\documentclass[11pt]{article}
\usepackage[margin=1in]{geometry}
\usepackage{amsmath,amsthm,amsfonts,amssymb}
\usepackage{graphicx,enumerate}
\graphicspath{ {./figures/} }
\usepackage[title]{appendix}
\usepackage[
            CJKbookmarks=true,
            bookmarksnumbered=true,
            bookmarksopen=true,
            colorlinks=true,
            citecolor=blue,
            linkcolor=blue,
            anchorcolor=red,
            urlcolor=blue
            ]{hyperref}

\usepackage{multirow}

\usepackage{tikz}
\usepackage{tikz-qtree}

\usetikzlibrary{calc}
\usetikzlibrary{decorations.pathreplacing,decorations.markings}

\tikzstyle{dot}=[circle,fill,black,inner sep=1pt]

\tikzset{
  on each segment/.style={
    decorate,
    decoration={
      show path construction,
      moveto code={},
      lineto code={
        \path [#1]
        (\tikzinputsegmentfirst) -- (\tikzinputsegmentlast);
      },
      curveto code={
        \path [#1] (\tikzinputsegmentfirst)
        .. controls
        (\tikzinputsegmentsupporta) and (\tikzinputsegmentsupportb)
        ..
        (\tikzinputsegmentlast);
      },
      closepath code={
        \path [#1]
        (\tikzinputsegmentfirst) -- (\tikzinputsegmentlast);
      },
    },
  },
  mid arrow/.style={postaction={decorate,decoration={
        markings,
        mark=at position .5 with {\arrow[#1]{stealth}}
      }}},
  early arrow/.style={postaction={decorate,decoration={
        markings,
        mark=at position .2 with {\arrow[#1]{stealth}}
      }}},
}

\newcommand{\bluedoubleedge}[2]{\path #1 edge[bend left,blue] #2 edge[bend right,blue] #2; \node[dot] at #1 {}; \node[dot] at #2 {};}
\newcommand{\reddoubleedge}[2]{\path #1 edge[bend left,red] #2 edge[bend right,red] #2; \node[dot] at #1 {}; \node[dot] at #2 {};}
\newcommand{\reddoubleedgeD}[2]{\path[red] #1 edge[bend left,postaction={on each segment={mid arrow=red}}] #2; \path[red] #2 edge[bend left,postaction={on each segment={mid arrow=red}}] #1; \node[dot] at #1 {}; \node[dot] at #2 {};}
\newcommand{\bluedoubleedgeD}[2]{\path[blue] #1 edge[bend left,postaction={on each segment={mid arrow=blue}}] #2; \path[blue] #2 edge[bend left,postaction={on each segment={mid arrow=blue}}] #1; \node[dot] at #1 {}; \node[dot] at #2 {};}
\newcommand{\redsingleedge}[2]{\path #1 edge[red] #2 ; \node[dot] at #1 {}; \node[dot] at #2 {};}
\newcommand{\bluesingleedge}[2]{\path #1 edge[blue] #2 ; \node[dot] at #1 {}; \node[dot] at #2 {};}

\newcommand{\bluesingleedgeD}[2]{\path #1 edge[blue,postaction={on each segment={mid arrow=blue}}] #2; \node[dot] at #1 {}; \node[dot] at #2 {};}
\newcommand{\redsingleedgeD}[2]{\path #1 edge[red,postaction={on each segment={mid arrow=red}}] #2; \node[dot] at #1 {}; \node[dot] at #2 {};}

\def\alternatecolorred{%
    \pgfkeysalso{red}%
    \global\let\alternatecolor\alternatecolorblue 
}
\def\alternatecolorblue{%
    \pgfkeysalso{blue}%
    \global\let\alternatecolor\alternatecolorred 
}

\newcommand{\altred}{\let\alternatecolor\alternatecolorred 
\tikzset{every edge/.append code = {%
    \global\let\currenttarget\tikztotarget 
    \pgfkeysalso{append after command={(\currenttarget)}}
			\alternatecolor
}}
}
\newcommand{\altblue}{\let\alternatecolor\alternatecolorblue 
\tikzset{every edge/.append code = {%
    \global\let\currenttarget\tikztotarget 
    \pgfkeysalso{append after command={(\currenttarget)}}
			\alternatecolor
}}
}

\tikzstyle{vertexdot}=[circle, draw, fill=black, minimum size=3,inner sep=0pt]

\newtheorem{theorem}{Theorem}
\newtheorem{lemma}{Lemma}

\newtheorem{corollary}{Corollary}
\newtheorem{definition}{Definition}

\newtheorem{assumption}{Assumption}

\usepackage{color}
\usepackage{subfigure}
\usepackage{algorithm}
\usepackage{algorithmic}

\usepackage{xspace,prettyref}

\newcommand{\diverge}{\to\infty}
\newcommand{\eqdistr}{\overset{d}{=}}

\newcommand{\ones}{\mathbf 1}

\newcommand{\reals}{{\mathbb{R}}}
\newcommand{\integers}{{\mathbb{Z}}}
\newcommand{\naturals}{{\mathbb{N}}}

\newcommand{\diff}{{\rm d}}
\newcommand{\red}{\color{red}}

\newcommand{\nbr}[1]{{\sf\red[#1]}}

\newcommand{\expect}[1]{\mathbb{E}\left[ #1 \right]}
\newcommand{\eexpect}[1]{\mathbb{E}[ #1 ]}

\newcommand{\pprob}[1]{ \mathbb{P}\{ #1 \} }
\newcommand{\prob}[1]{ \mathbb{P}\left\{ #1 \right\} }

\newcommand\indep{\protect\mathpalette{\protect\independenT}{\perp}}
\def\independenT#1#2{\mathrel{\rlap{$#1#2$}\mkern2mu{#1#2}}}
\newcommand{\Bern}{{\rm Bern}}
\newcommand{\Binom}{{\rm Binom}}
\newcommand{\Pois}{{\rm Pois}}

\newcommand{\ie}{i.e.\xspace}

\newrefformat{eq}{(\ref{#1})}
\newrefformat{chap}{Chapter~\ref{#1}}
\newrefformat{sec}{Section~\ref{#1}}
\newrefformat{alg}{Algorithm~\ref{#1}}
\newrefformat{fig}{Fig.~\ref{#1}}
\newrefformat{tab}{Table~\ref{#1}}
\newrefformat{rmk}{Remark~\ref{#1}}
\newrefformat{clm}{Claim~\ref{#1}}
\newrefformat{def}{Definition~\ref{#1}}
\newrefformat{cor}{Corollary~\ref{#1}}
\newrefformat{lmm}{Lemma~\ref{#1}}
\newrefformat{prop}{Proposition~\ref{#1}}
\newrefformat{app}{Appendix~\ref{#1}}
\newrefformat{hyp}{Hypothesis~\ref{#1}}
\newrefformat{thm}{Theorem~\ref{#1}}
\newrefformat{ass}{Assumption~\ref{#1}}
\newrefformat{conj}{Conjecture~\ref{#1}}

\newcommand{\sth}[1]{\left\{ #1 \right\}}

\newcommand{\iprod}[2]{\left \langle #1, #2 \right\rangle}

\newcommand{\indc}[1]{{\mathbf{1}_{\left\{{#1}\right\}}}}

\newcommand{\calB}{{\mathcal{B}}}

\newcommand{\calD}{{\mathcal{D}}}
\newcommand{\calE}{{\mathcal{E}}}
\newcommand{\calF}{{\mathcal{F}}}

\newcommand{\calH}{{\mathcal{H}}}

\newcommand{\calN}{{\mathcal{N}}}

\newcommand{\calP}{P}
\newcommand{\calQ}{Q}

\newcommand{\calT}{{\mathcal{T}}}
\newcommand{\calU}{{\mathcal{U}}}

\DeclareMathAlphabet{\varmathbb}{U}{bbold}{m}{n}

\renewcommand{\hat}{\widehat}
\renewcommand{\tilde}{\widetilde}

\usepackage{bbm}

\def \exppath {figures/}

\begin{document}

\pgfdeclarelayer{background}
\pgfdeclarelayer{foreground}
\pgfsetlayers{background,main,foreground}

\title{Hidden Hamiltonian Cycle Recovery via Linear Programming}
\author{Vivek Bagaria, Jian Ding, David Tse, Yihong Wu, Jiaming Xu\thanks{V.\ Bagaria is with Stanford University, Stanford, USA\texttt{ vbagaria@stanford.edu}.
J.\ Ding is with University of Pennsylvania, Philadelphia, USA \texttt{ dingjian@wharton.upenn.edu}.
D.\ Tse is with Stanford University, Stanford, USA \texttt{ dntse@stanford.edu}.
Y.\ Wu is with Yale University, New Haven, USA\texttt{ yihong.wu@yale.edu}.
J.\ Xu is with Purdue University, West Lafayette, USA\texttt{ xu972@purdue.edu}.
}}
\maketitle

\begin{abstract}

We introduce the problem of hidden Hamiltonian cycle recovery, where there is an unknown Hamiltonian cycle in an $n$-vertex complete graph that needs to be inferred from noisy edge measurements. The measurements are independent and distributed according to $\calP_n$ for edges in the cycle and $\calQ_n$ otherwise. This formulation is motivated by a problem in genome assembly, where the goal is to order a
set of contigs (genome subsequences) according to their positions on the genome using long-range linking measurements between the contigs.  Computing the maximum likelihood estimate in this model reduces to a Traveling Salesman Problem (TSP). Despite the NP-hardness of TSP, we show that a simple linear programming (LP) relaxation, namely the fractional $2$-factor (F2F) LP, recovers the hidden Hamiltonian cycle with high probability as $n \to \infty$ provided that $\alpha_n - \log n \to \infty$, where $\alpha_n \triangleq -2 \log \int \sqrt{d P_n d Q_n}$ is the R\'enyi divergence of order $\frac{1}{2}$. This condition is information-theoretically optimal in the sense that,  under mild distributional assumptions, $\alpha_n \geq (1+o(1)) \log n$ is necessary for any algorithm to succeed regardless of the computational cost. 

Departing from the usual proof techniques based on dual witness construction, the analysis relies on the combinatorial characterization (in particular, the half-integrality) of the extreme points of the F2F polytope. Represented as bicolored multi-graphs, these extreme points are further decomposed into simpler ``blossom-type'' structures for the large deviation analysis and counting arguments. Evaluation of the algorithm on real data shows improvements over existing approaches.

\end{abstract}



\newpage

\section{Introduction}
\label{sec:introduction}
Given an input graph, the problem of finding a subgraph satisfying certain properties has diverse
applications. MAX CUT, MAX CLIQUE, TSP are a few canonical examples. Traditionally these problems have been studied in theoretical computer science from the worst-case perspective, and many such problems have been shown to be NP-hard. However, in machine learning applications, many such problems arise when an underlying \emph{ground truth} subgraph needs to be recovered from the noisy measurement data represented by the entire graph.  Canonical models to study such problems
include planted partition models~\cite{Condon01} (such as planted clique~\cite{Jer92}) in community detection and  planted ranking models (such as Mallows model~\cite{mallows1957non})
in rank aggregation. In these models, the planted or hidden subgraph represents the ground-truth,  and one is not necessarily interested in the worst-case instances but rather in only instances for which there is enough information in the data to recover the ground-truth sub-graph, i.e. when the amount of data is above the information limit.  The key question is whether there exists an efficient recovery algorithm that can be successful all the way to the information limit. 

 In this paper,  we pose and answer this question for a hidden Hamiltonian cycle recovery model: 
 
\begin{definition}[Hidden Hamiltonian cycle recovery]\label{def:model}\hfill

\noindent
{\bf Given}: $n \ge 1$, and two distributions $\calP_n$ and $\calQ_n$, parameterized by $n$.\\
{\bf Observation}: A randomly weighted, undirected complete graph 
$G=([n], E)$ with a hidden Hamiltonian cycle $C^*$ such that every edge has an
independent weight distributed as $\calP_n$ if it is on $C^*$ and  as $\calQ_n$ otherwise. \\
{\bf Inference Problem:} Recover the hidden Hamiltonian cycle $C^*$ from the observed random graph.
\end{definition}


Our problem is motivated from {\em de novo} genome assembly, the reconstruction of an organism's long sequence of A,G,C,T nucleotides  from fragmented sequencing data.  The first step of the standard assembly pipeline stitches together short, overlapping fragments (so-called shotgun reads) to form longer subsequences called contigs,  of lengths typically tens to hundreds of thousands  of nucleotides (\prettyref{fig:assembly}).  Due to coverage gaps and other issues, these individual contigs cannot be extended to the whole genome. To get a more complete picture of the genome, the contigs need to be ordered according to their positions on the genome, a process called {\em scaffolding}. Recent advances in sequencing assays~\cite{lieberman2009comprehensive, putnam2016chromosome}  aid this process by providing long range linking information between these contigs in the form of randomly sampled {\em Hi-C reads}. This data can be summarized by a contact map (\prettyref{fig:contact}),  tabulating the counts of Hi-C reads linking each pair of contigs. The problem of ordering the contigs from the contact map data can be modeled by the  hidden Hamiltonian cycle recovery problem, where  the vertices of the graph are the contigs, the hidden Hamiltonian cycle is the true ordering of the contigs on the genome,\footnote{Strictly speaking, this only applies to genomes which are circular. For genomes which are linear, the ordering of the contigs would correspond to a hidden Hamiltonian {\em path}. We show in Section \ref{sec:cycle_path} that our results extend to a hidden Hamiltonian path model as well.} and the weights on the graph are the counts of the Hi-C reads linking the contigs. As can be seen in \prettyref{fig:contact}(a), there is a much larger concentration of Hi-C reads between contigs adjacent on the genome than between far-away contigs. A first order model is to choose $\calP_n = \Pois(\lambda_n)$ and $\calQ_n = \Pois(\mu_n)$, where $\lambda_n$ is the average number of Hi-C reads between adjacent contigs and $\mu_n$ is the average number between non-adjacent contigs. The parameter $n \lambda_n + \frac{n(n-1)}{2} \mu_n$ increase with the coverage depth\footnote{The coverage depth is the average number of Hi-C reads that include a given nucleotide (base pair).} of the Hi-C reads and are part of the design of the sequencing experiment.

\begin{figure}[ht]
\centering
\includegraphics[width=0.9\columnwidth]{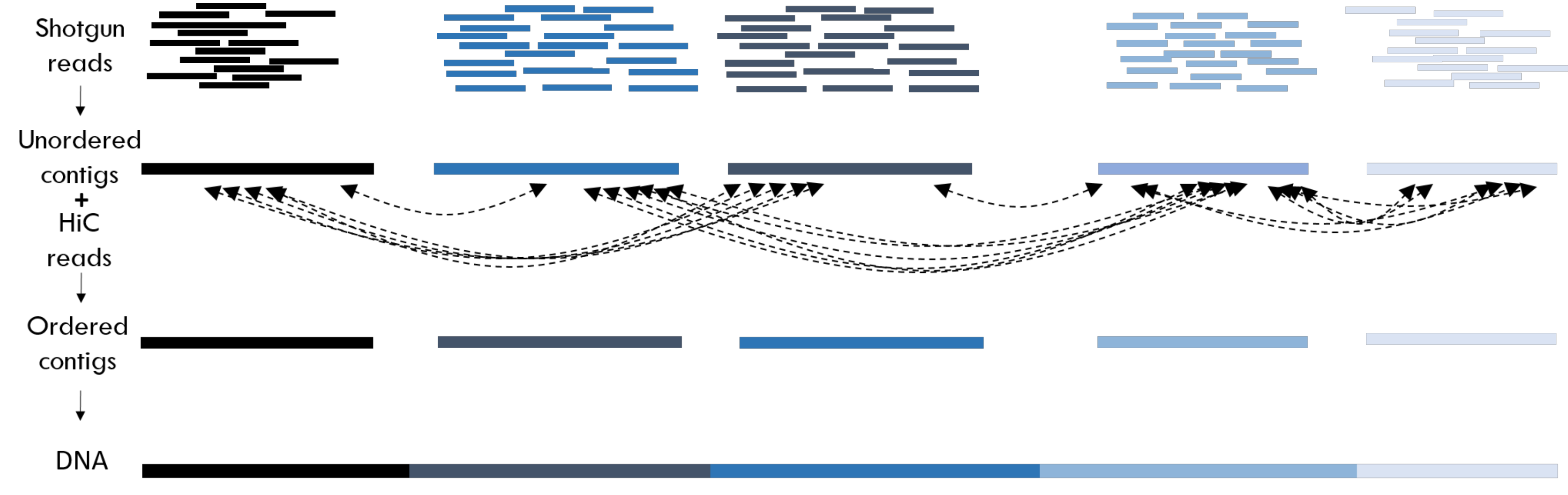}
\caption{Short reads are assembled to form contigs which are then ordered by using long-range linking Hi-C reads.}
\label{fig:assembly}
\end{figure}

\begin{figure}[ht]
\centering
\includegraphics[width=0.8\columnwidth]{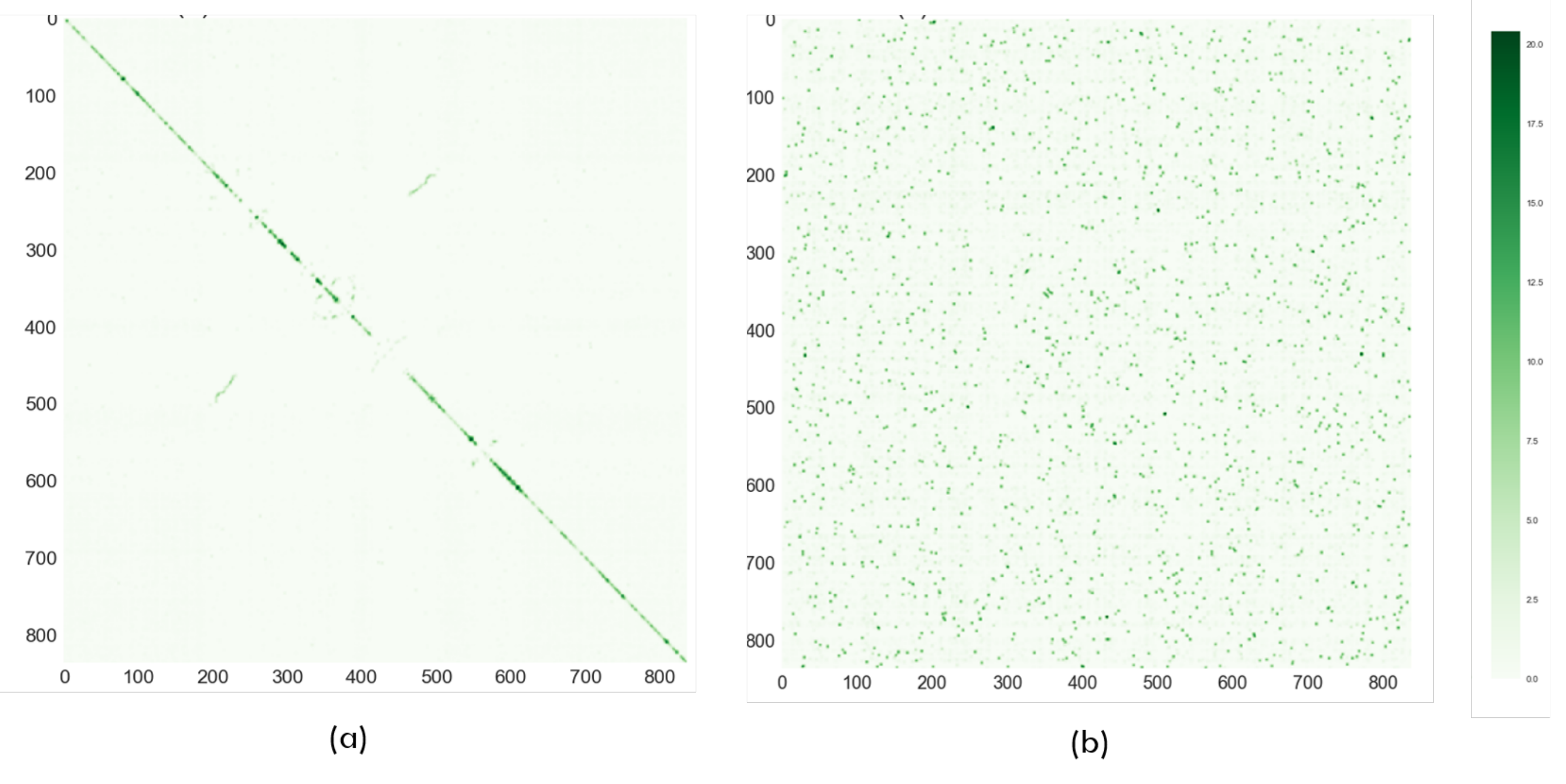}
\caption{(a): Contact map where the rows (and columns) correspond to \emph{ordered} contigs of human chromosome $1$~\cite{putnam2016chromosome}
and the value at entry $(i,j)$ corresponds to the number 
of HiC reads  between contig $i$ and contig $j$.
 (b): Contact map of the \emph{unordered} matrix in (a), where the contigs are randomly ordered. This is the data from which the ordering of the contigs is inferred.}
 \label{fig:contact}
 \end{figure}

 The hidden Hamiltonian cycle can be represented as an
adjacency vector $x^* \in \{0,1\}^{\binom{n}{2}}$ such that
$x^*_e=1$ if edge $e$ is on the Hamiltonian cycle, and
$x^*_e=0$ otherwise. 
Let $A$ denote the weighted adjacency matrix of $G$, so that $A_e$ is distributed according to $P$ (resp.~$Q$) if $x^*_e=1$ (resp.~0).
The maximum likelihood (ML) estimator  for the hidden Hamiltonian cycle recovery problem is equivalent to solving the traveling salesman problem (TSP) on a transformed weighted graph, where each edge weight $w_e=\log \frac{\diff P}{\diff Q}(A_e)$ is the log likelihood ratio  evaluated on the weights of the observed graph:
\begin{align}
\widehat{x}_{\rm ML} = \arg \max_{x}  & \;  \iprod{ w} {x}   \label{eq:mle} \\
\text{s.t.  } & \; \text{ $x$ is the adjacency vector of a Hamiltonian cycle in $G$}. \nonumber
\end{align}
In the Poisson or Gaussian model where the log likelihood ratio is an affine function, we can simply take $w$ to be $A$ itself.
 
 Solving TSP is NP-hard, and a natural approach is to look for a tractable relaxation. It is well-known that TSP  \prettyref{eq:mle}  can be cast as an integer linear program (ILP)~\cite{SWZ13}:
\begin{align}
\widehat{x}_{\rm TSP} = \arg \max_{x}  & \;  \iprod{w}{x} \label{eq:tsp_ilp}\\
\text{s.t.  } & \; x \left(  \delta(v) \right) =2  
\label{eq:tsp_degree}\\
& \; x \left( \delta(S) \right) \ge 2, \; \forall S \subset [n],  \; 3 \le |S| \le n-3  \label{eq:subtour_elim} \\
& \; x_e \in \{ 0, 1\}, \label{eq:tsp_integer}
\end{align}
where $\delta(S)$ denotes the set of all edges in $G$ with exactly one endpoint in $S \subset [n]$, and $\delta(v) \triangleq \delta(\{v\})$.
In particular, 
\prettyref{eq:tsp_degree} are called \emph{degree} constraints, enforcing each vertex to have exactly two incident edges in the graph represented by the adjacency vector $x$, while \prettyref{eq:subtour_elim} are \emph{subtour elimination} constraints, eliminating solutions whose corresponding graph is a disjoint union of subtours of length less than $n$ . Note that there are small number of degree constraints but exponentially large number of subtour elimination constraints. If we drop the subtour elimination constraints as well as relax the integer constraints on $x$, we obtain
 the \emph{fractional $2$-factor (F2F)} LP relaxation:\footnote{
A $2$-factor is a spanning subgraph consisting of disjoint cycles. }
\begin{align}
\widehat{x}_{\rm F2F} = \arg \max_{x}  & \;   \iprod{w}{x}   \label{eq:F2F}  \\
\text{s.t.  }  & \;    x \left(  \delta(v) \right) =2  \nonumber \\ 
& \;  x_e \in [0,1]. \nonumber
\end{align}

The main result of the paper is the following. We abbreviate $P_n$ and $Q_n$ as $P$ and $Q$, respectively. 
\begin{theorem}
\label{thm:LP_opt}
Define:
\begin{equation}
\alpha_n \triangleq - 2    \log \int \sqrt{\diff P \diff Q}
\label{eq:alpha}
\end{equation}
to be  the R\'enyi divergence of order $\frac{1}{2}$ between distributions $\calP$ and $\calQ$.\footnote{
The R\'enyi divergence of order $\rho>0$ from $P$ to $Q$ is defined as \cite{Renyi61}
\begin{equation}
D_{\rho}(P\|Q) \triangleq \frac{1}{\rho-1} \log \int (dP)^{\rho} (dQ)^{1-\rho}.
\label{eq:renyi}
\end{equation}
It particular, for $\rho=1/2$ it is related to the so-called Battacharyya distance $B(P,Q)$ via $D_{\frac{1}{2}}(P\|Q) = 2 B(P,Q)$.
} If
\begin{equation}
\alpha_n - \log n  \to + \infty, 
\label{eq:ITlimit}
\end{equation}
then the optimal solution of the F2F LP \prettyref{eq:F2F} satisfies 
$
\min_{x^*} \prob{ \hat x_{\rm F2F} = x^* } \to 1
$
as $n \to \infty$. 
\end{theorem}

For Gaussian, Poisson, or Bernoulli weight distribution, the explicit expressions of 
$\alpha_n$ are given as follows
\begin{align}
\alpha_n=
\begin{cases}
\mu^2/4 & \text{ if } \calP=N(\mu,1), \; \calQ=N(0,1) \\
\left(\sqrt{\lambda} - \sqrt{\mu} \right)^2 & \text { if } \calP=\Pois(\lambda), \; \calQ=\Pois(\mu) \\
- 2 \log \left(\sqrt{pq} + \sqrt{(1-p)(1-q)} \right) & \text { if } \calP=\Bern(p), \; \calQ=\Bern(q). 
\end{cases}
\label{eq:alphaexample}
\end{align}

Although the relaxation from TSP to F2F LP is quite drastic, the resulting algorithm is in fact information theoretically optimal for the hidden Hamiltonian cycle recovery problem. 
Specifically, under an assumption which can be easily verified for Poisson, Gaussian or Bernoulli weight distribution,  we show  in \prettyref{sec:IT} that  if there exists any algorithm, efficient or not,  which exactly recovers $x^*$ with
high probability, then it must hold that
\[
\alpha_n \geq (1+o(1)) \log n.
\]
This necessary condition together with sufficient condition \prettyref{eq:ITlimit}
implies that the optimal recovery threshold is at 
$$
\liminf_{n\diverge} \frac{\alpha_n}{\log n} =1,
$$ 
achieved by the F2F LP. 

Applying the results back to the geome scaff
We discuss two consequences of \prettyref{thm:LP_opt}. First, as a corollary of the integrality and the optimality of the F2F LP, it can be shown that the Max-Product belief propagation algorithm introduced in \cite{bayati2011belief} can be used to solve the F2F LP exactly, which, for the Gaussian or Poisson weight distribution, requires $o(n^2 \log n)$ iterations (see \prettyref{sec:other_efficient_algorithms} for details). Second, note that we do not require the edge weights to be real-valued. Thus the formulation also encompasses the case of partial observation, by letting the weight of every edge in $G$ takes on a special ``erasure'' symbol with some probability. See~\prettyref{app:partial} for details.



 
In related work, a version of the hidden Hamiltonian cycle model was studied in~\cite{broder1994finding}, where the observed graph is the superposition of a hidden Hamiltonian cycle and an Erd\"os-R\'enyi random graph with constant average degree $d$. Our measurement model is more general than the one in~\cite{broder1994finding}, but more importantly, the goal in~\cite{broder1994finding} is not to recover the hidden Hamiltonian cycle but rather to find \emph{any} Hamiltonian cycle in  the observed graph, which may not coincide with the hidden one.  (In fact in the regime considered there, exact recovery of the hidden cycle is information theoretically impossible.\footnote{To see this, suppose the hidden Hamiltonian cycle is given by sequence of vertices $(1,2,\ldots, n,1)$. If $q=\Omega(1/n)$, then with a non-vanishing probability there exist $1 \le i \le n-4$ and $ i+2\le j \le n-2$ such that $(i,j)$ and $(i+1,j+1)$ are edges in $G$. Thus we have a new Hamiltonian cycle by deleting edges $(i,i+1)$ and $(j,j+1)$ in the  hidden one and adding edges $(i,j)$ and $(i+1,j+1)$, leading to the impossibility of exact recovery. See \prettyref{fig:c4} for an illustration.})   The fractional $2$-factor relaxation of TSP has been well-studied in the worst case \cite{dantzig1954,BC99,SWZ13}.
It has been shown that for metric TSP (the cost minimization formulation) where the costs are symmetric and satisfy the triangle
inequality, the \emph{integral gap} of F2F is $4/3$; here the integrality gap is defined as the worst-case ratio of the cost
of the optimal integral solution to the cost of the optimal relaxed solution. In contrast, our model does not make any  metric assumption on the graph weights. 


The rest of the paper is organized as follows. In Section \ref{sec:other_efficient_algorithms}, we describe a few other computationally efficient algorithms for the hidden Hamiltonian cycle problem and benchmark their performance against the information-theoretic limit. 
In Section \ref{sec:related_work}, we discuss related work in more details. Sections \ref{sec:warmup} and \ref{sec:proof_theorem1} are devoted to the proof of Theorem \ref{thm:LP_opt}, while Section \ref{sec:IT} characterizes the information theoretic  limit for the recovery problem. In \prettyref{sec:cycle_path} we describe the closely related hidden Hamiltonian path problem and show that it can be reduced to and from the hidden Hamiltonian cycle problem both statistically and computationally. 
Empirical evaluation of various algorithms on both simulated and real DNA datasets are given in Section \ref{sec:emp_results}.

\section{Performance of Other Algorithms}
\label{sec:other_efficient_algorithms}
It is striking to see that the simple F2F LP relaxation of the TSP achieves the
optimal recovery threshold in the hidden Hamiltonian cycle model. 
A natural question to ask is whether there exists other efficient and perhaps even simpler estimator with provable optimality. We 
have considered various efficient algorithms and derived their performance guarantees. 
\begin{table}[ht]
\begin{center}
\begin{tabular}{l|c} 
Efficient Algorithms         & Performance Guarantee \\
\hline \hline
F2F LP &  \multirow{ 2}{*}{$\alpha_n - \log n \to +\infty$}  \\
MaxProduct BP &   \\
\hline
Greedy Merging & $\beta_n - \log n \to + \infty$ \\
\hline
Simple Thresholding & \multirow{ 2}{*}{$\alpha_n - 2 \log n \to +\infty $} \\
Nearest Neighbor & \\
\hline
Spectral Methods     & $\alpha_n \gg n^5$ (Gaussian)  \\
\hline
\end{tabular}
\caption{Sufficient conditions for various efficient algorithms to achieve exact recovery.}
\label{table:various_alg}
\end{center}
\end{table}
As summarized in Table~\ref{table:various_alg}, spectral algorithms is orderwise suboptimal; greedy methods including thresholding achieve the optimal scaling but not the sharp constant. Finally, MaxProduct Belief Propagation also achieves the sharp threshold as a corollary of our result on the F2F LP.
In Table \prettyref{table:various_alg}, 
\begin{align}
\beta_n \triangleq - \frac{3}{2} \log \int  (\diff P)^{2/3} (\diff Q)^{1/3}  \label{eq:def_beta_n}
\end{align}
is the $\frac{1}{3}$-R\'{e}nyi divergence from $Q$ to $P$; cf.~\prettyref{eq:renyi}. By Jensen's and H\"older's inequality, for any distinct $P$ and $Q$, we have 
\begin{equation}
\frac{1}{2} \alpha_n < \beta_n < \alpha_n.
\label{eq:alphabeta}
\end{equation}
For Gaussian weights with $\calP = \calN(\mu,1)$ and $\calQ = \calN(0,1)$, we have $\beta_n = \frac{1}{6} \mu^2 = \frac{2}{3} \alpha_n.$ Simulation of these algorithms confirm these theoretical results. See Figure \ref{fig:simulations_a} in Section \ref{subsec:sim}.

\paragraph{Spectral Methods} Spectral algorithms are powerful methods for recovering the underlying structure
in planted models based on the principal eigenvectors of the 
observed adjacency matrix $A$. Under planted models
such as planted clique~\cite{Alon98} or planted partition models~\cite{McSherry01}, spectral algorithms and their
variants have been shown to achieve either the optimal recovery thresholds~\cite{Massoulie13,BordenaveLelargeMassoulie:2015dq,AbbeSandon15} or the best possible
performance within certain relaxation hierarchies~\cite{Meka15,DeshpandeMontanari15,barak-etal-planted-clique}. 
The rationale behind spectral
algorithms is that the principal eigenvectors of $\expect{A}$ contains information
about underlying structures and the principal eigenvectors of $A$ are close to
those of $\expect{A}$, provided that the spectral gap (the gap between the largest
few eigenvalues and the rest of them) is much larger than the spectral norm of
the perturbation $\|A-\expect{A}\|$.  In our setting, indeed the principal eigenvectors
of $\expect{A}$ contain
information about the ground truth Hamiltonian cycle $C^*$. To see this, let us
consider the Gaussian case where $P=\calN(\mu,1)$ and $Q=\calN(0,1)$ as an
illustrating example. Then the observed matrix can be expressed as
$$
A = \mu C^* + Z,
$$
where with a slight abuse of notation we use $A$ to denote the weighted adjacency matrix
of $G$ and $C^*$ to denote 
the adjacency matrix of the true Hamiltonian cycle; 
$Z$ is a symmetric Gaussian matrix with zero diagonal and 
$Z_{ij}=Z_{ji}$ independently drawn from~$\calN(0,1)$ for $i<j$.
Since $C^*$ is a circulant matrix, its eigenvalues and the corresponding
eigenvectors can be explicitly derived via discrete Fourier transform. It turns out that
the eigenvector corresponding to the second largest eigenvalue of $C^*$ contains perfect information about the
true Hamiltonian cycle. Unfortunately, in contrast to the planted
clique and planted partition models under which $\expect{A}$ is low rank and
has a large eigen-gap, here $C^*$ is full-rank and the gap
between the second and the third largest eigenvalue is on the order of $1/n^2$,
which is much smaller than $\|Z\| = \Theta(\sqrt{n})$. Therefore, for spectral
algorithms to succeed, it requires a very high signal level $\mu^2 \gg n^5$. This agrees with the empirical performance on simulated data and is highly suboptimal as compared to the sufficient condition \prettyref{eq:ITlimit} of F2F LP: 
$\mu^2 -4 \log n \to \infty$.

\paragraph{Greedy Methods} To recover the hidden Hamiltonian cycle, we can also resort to greedy methods. It turns out that the following 
simple thresholding algorithm achieves the optimal recovery threshold \prettyref{eq:ITlimit} within a factor of two:
for each vertex, keep the two incident edges with the two largest weights and delete the other $n-3$ edges. 
The resulting graph has degree at most $2$. It can be shown  that the resulting graph coincides with $C^*$ 
with high probability provided that  $\alpha_n - 2 \log n \to + \infty$.  

Another well-known greedy heuristic is the following nearest-neighbor algorithm. Start on an arbitrary vertex as the current vertex and 
find the edge with the largest weight connecting the current vertex and an unvisited vertex $v$,
set the current vertex to $v$ and mark $v$ as visited. Repeat until all vertices have been visited.
Let $ v_1, \ldots, v_{n}$ denote the sequence of visited vertices and output the Hamiltonian cycle formed by 
$(v_1, \ldots, v_n, v_1)$. It can be shown (see \prettyref{app:greedy}) that the resulting Hamiltonian cycle coincides with $C^*$ 
with high probability provided that  $\alpha_n - 2 \log n \to + \infty$.  

Finally, we consider a greedy merging algorithm proposed in~\cite{motahari2013information}: connect pairs of vertices with the largest edge weights until all vertices have degree two. 
The output is a $2$-factor, and it can be shown  that the output $2$-factor coincides with $C^*$ with high probability provided that
$\beta_n - \log n \to  +\infty $, strictly improving on the performance guarantee of previous two greedy algorithms.

Notice that the aforementioned greedy algorithms only exploit local information, and do not take the into account the global cycle structure. Naturally, none of them achieves the optimal threshold \prettyref{eq:ITlimit}. See \prettyref{app:greedy}
for further details.

\paragraph {Max-Product Belief Propagation}
We can improve on the simple thresholding algorithm using an iterative message-passing algorithm 
known as max-product belief propagation. Specifically, at each time $t=0, 1, \ldots, t_f$, each vertex $i$ sends a real-valued message $m_{i\to j}(t)$
to each of its neighbors $j$. Messages are initialized by $m_{i\to j}(0)=w_e$ for all $e =(i,j)$. For $t\ge 1$, messages transmitted by 
vertex $i$ in iteration $t$ are updated based on messages received in iteration $t-1$ recursively as follows:
$$
m_{ i \to j} (t) = w_e - \mathop{\text{2nd}\max}_{\ell \neq j} \left\{ m_{\ell \to i} (t-1) \right\}, 
$$
where $\text{2nd}\max $ denotes the second largest value. At the end of the final iteration $t_f$,  for every 
vertex, keep the two incident edges with the two largest received message values and delete the other $n-3$ edges,
and output the resulting graph. Note that BP with one iteration $t_f=1$ reduces to the simple thresholding algorithm. 
 
Belief propagation algorithm is studied in~\cite{bayati2011belief} to find the $b$-factor with the maximum weight for $b \ge 1$; 
it is shown that if the fractional
$b$-factor LP relaxation has no fractional optimum solution, then 
the output of BP coincides with the optimal $b$-factor when $t_f \ge \lceil  2 n w^*/\epsilon \rceil$,
where $w^*$ is the weight of the optimal $b$-factor and $\epsilon$ is the difference between the weight
of the optimal $b$-factor and the second largest weight of $b$-factors. 
Our optimality result of F2F implies that  if $\alpha_n - \log n \to +\infty$, then with high probability, 
F2F has no fractional optimum solution and  the optimal $2$-factor coincides with the ground truth $x^*$;
Therefore, by combining our result with results of BP in \cite{bayati2011belief}, we immediately conclude
that  the output of BP coincides with $x^*$ with high probability after $t_f$ iterations, 
provided that $\alpha_n - \log n \to +\infty$. 
For both the Gaussian and the Poisson model, with high probability the number of iterations $t_f$ of the BP algorithm is in fact $o(n^2 \log n)$, nearly linear in the problem size (see \prettyref{app:BPtime} for a justification).

\section{Related Work}
\label{sec:related_work}

We discuss additional related work before presenting the proof
of our main results. Because of the NP-hardness of TSP, researchers have imposed 
structural assumptions on the costs (weights) and devised
efficient approximation algorithms. One natural 
assumption is the metric assumption under which 
the costs are symmetric ($c_{ij}=c_{ji}$ for all $i, j \in V$)
and satisfy the triangle
inequality ($ c_{ik} \leq c_{ij}+c_{jk}$ for all $i,j, k \in V$).  Metric TSP turns out
to be still NP-hard, as shown by
reduction from the NP-hard Hamiltonian cycle problem~\cite[Theorem 58.1]{schrijver2003combinatorial}.
The best approximation algorithm for metric TSP currently known is
Christofides' algorithm which finds a Hamiltonian cycle of cost
at most a factor of $3/2$ times the cost of an optimal Hamiltonian cycle. 

\paragraph{Integrality gap of LP relaxations of TSP}
Various relaxations of TSP has also been extensively studied under the metric 
assumption. 
To measure the tightness of LP relaxations, a commonly used figure of merit is the \emph{integrality gap}. 
As is the convention in the TSP literature, the optimization is formulated as a minimization problem with nonnegative costs.
In general, the integrality gap is defined as the supremum of the
ratio $\mathsf{OPT}/\mathsf{FRAC}$,
over all instances of the problem, where $\mathsf{FRAC}$ denotes the 
objective value of the optimal fractional solution and $\mathsf{OPT}$ denotes
the objective value of the optimal integral solution~\cite{chlamtac2012convex}. 
Note that by definition the integrality gap is always at least one. 
Dropping the integer constraints \prettyref{eq:tsp_integer} in ILP formulation
of TSP \prettyref{eq:tsp_ilp} leads to a LP relaxation known as \emph{Subtour LP}~\cite{dantzig1954,held1970traveling}.
The integrality gap of the subtour LP  is known to be between 
$4/3$ and $3/2$.
The integrality gap of fractional $2$-factor LP \prettyref{eq:F2F} is  shown in
\cite{BC99,SWZ13} to be $4/3$.
In contrast to the previous worst-case approximation results on metric TSP, this paper focuses on a
planted instance of TSP, where we impose probabilistic assumption on the costs (weights)
and the goal is to recover the hidden Hamiltonian cycle. In particular, the metric assumption is
not fulfilled in our hidden Hamiltonian cycle model and hence the previous results do not apply. Our results
imply that when $\alpha_n -\log n \to +\infty$, the optimal solution of F2F coincides with the optimal solution
of TSP with probability tending to $1$, where the probability is taken over the randomness of weights $w$ 
in the hidden Hamiltonian cycle model. 
%
In other words, for ``typical'' instances of the hidden Hamiltonian cycle model, the optimal objective value of 
TSP is the same as that of F2F. 

\paragraph{SDP relaxations of TSP}
Semidefinite programming (SDP) relaxations of the traveling salesman problem have also been extensively studied in the literature. 
A classical SDP relaxation of TSP  due to \cite{cvetkovic1999semidefinite} is obtained by imposing an extra constraint on
the second largest eigenvalue of a Hamiltonian cycle in F2F LP \prettyref{eq:F2F}.  A more sophisticated SDP relaxation is derived in \cite{zhao1998semidefinite} by viewing the TSP as a quadratic assignment problem, 
from which one can obtain a simpler SDP relaxation of TSP based on association schemes~\cite{de2008semidefinite}. 
This SDP relaxation in~\cite{de2008semidefinite} is shown to dominate that of \cite{cvetkovic1999semidefinite}.
Since all these SDP relaxations are tighter than the F2F LP, our results immediately imply that the optimal solutions of 
these SDP relaxations coincide with the true Hamiltonian cycle $x^*$ with high probability provided $\alpha_n-\log n \to +\infty$.

\paragraph{Data seriation}
The problem of recovering a hidden Hamiltonian cycle (path) in a weighted complete graph 
falls into a general problem known as data seriation~\cite{kendall1971abundance} or data stringing~\cite{DataStringing11}.
In particular, we are given a similarity matrix $Y$ for $n$ objects,
and are interested in seriating or stringing the data, by ordering the $n$ objects
 so that similar objects $i$ and $j$ are near each other. 
 Data seriation has diverse applications
 ranging from data visualization, DNA sequencing to functional data analysis~\cite{DataStringing11} and archaeological dating~\cite{robinson1951method}.  
Most previous work on data seriation focuses on the noiseless case~\cite{robinson1951method,kendall1971abundance}, 
where there is an unknown ordering of $n$ objects so that if object $j$ is closer  than object $k$ to object $i$ in the ordering, then
 $Y_{ij} \ge Y_{ik}$, \ie, the similarity between $i$ and $j$ is always no less than the similarity between $i$ and $k$. Such a matrix $Y$ is called
 \emph{Robinson matrix}.   It is shown in~\cite{SeriationSpectral98} that
one can recover the underlying true ordering of objects up to a global shift by component-wisely sorting the second eigenvector of 
 the Laplacian matrix associated with $Y$ if $Y$ is a Robinson matrix. 
The data seriation problem has also been formulated as a quadratic assignment problem
and convex relaxations are derived~\cite{fogel2013convex,lim2014beyond} and references therein).

One interesting generalization of the hidden Hamiltonian cycle model is to 
extend the hidden structure from a cycle to $k$-regular graph for general $k\ge 2$, e.g., nearest-neighbor graphs, which can potentially better fit the genome assembly data.
The underlying $k$-regular graph can represent the hidden geometric
structure and  the observed graph can be viewed as a realization of the Watts-Strogatz small-world graph~\cite{watts1998collective} if
the weight distribution is Bernoulli.  
Recent work~\cite{cai2017detection} has studied the problem of detecting and recovering
the underlying $k$-regular graph under the small-world graph model and derived 
conditions for reliable detection and recovery; however, the information limit and the optimal algorithm
remain open. 

Finally, we mention that the R\'enyi divergence of order $1/2$ also plays a key role in determining the exact recovery threshold for community detection under the stochastic block models ~\cite{Abbe14,mossel2015consistency,zhang2016minimax,JL15,AbbeSandon15}.



\section{Proof Techniques and a Simpler Result}
\label{sec:warmup}

The proof of the main result, Theorem \ref{thm:LP_opt}, is quite involved. In this section, we will discuss the high-level ideas and the difference with the conventional proof using dual certificates. As a warm-up we also prove a weaker version of the result on the 2-factor ILP. The proof of the full result is given in Section \ref{sec:proof_theorem1}.

\subsection{Proof Techniques}

A standard technique for analyzing convex relaxations is the \emph{dual certificate argument}, which amounts to constructing the dual variables so that the desired KKT conditions are satisfied for the primal variable corresponding to the ground truth. 
This type of argument has been widely used, for instance, for proving the optimality of SDP relaxations for community detection under stochastic block models~\cite{Abbe14,HajekWuXuSDP14,Bandeira15,HajekWuXuSDP15,ABBK,perry2015semidefinite,HajekWuXu_one_sdp15}. 
However, for the F2F LP \prettyref{eq:F2F}, we were only able to find explicit constructions of dual certificates that attain the optimal threshold within a factor of two. Instead, the proof of \prettyref{thm:LP_opt} is by means of a direct \emph{primal argument}, which shows that with high probability no other vertices of the F2F polytope has a higher objective value than that of the ground truth. Nevertheless, it is still instructive to describe this dual construction before explaining the ideas of the primal proof.

 To certify the optimality of $x^*$ for F2F LP, it reduces to constructing a dual variable $u \in \reals^n$ (corresponding to
the degree constraints)
such that for every edge $(i,j)$,
\begin{align}
u_i + u_j & \le w_{ij} , \quad \text{ if $x^*_{ij}=1$},  \label{eq:dual_certificate_1}\\
u_i + u_j & \ge w_{ij},  \quad \text{ if $x^*_{ij}=0$}  \label{eq:dual_certificate_2}. 
\end{align}
A simple choice of $u$ is
\begin{equation}
u_i = \frac{1}{2} \min_j \left\{  w_{ij}: x^*_{ij} = 1 \right\}.
\label{eq:dual}
\end{equation}
 Then \prettyref{eq:dual_certificate_1} is fulfilled automatically and \prettyref{eq:dual_certificate_2} can be shown (see \prettyref{app:dual}) to hold with high probability provided that 
\begin{align}
\beta_n  - \log n \to + \infty. \label{eq:dual_condition}
\end{align}
where $\beta_n$ is the $\frac{1}{3}$-R\'{e}nyi divergence defined in \prettyref{eq:def_beta_n}.
Since $\alpha_n /2 < \beta_n < \alpha_n$ by \prettyref{eq:alphabeta}, this construction shows that 
F2F achieves the optimal recovery threshold by at most a multiplicative factor of $2$. For specific distributions, this factor-of-two gap can be further improved, e.g., to $\frac{3}{2}$ for Gaussian weights, for which we have $P=N(\mu,1)$ and $Q=N(0,1)$ and $\beta = \frac{1}{6}\mu^2 = \frac{2}{3} \alpha$. However, this certificate does not get us all the way to the information limit \prettyref{eq:ITlimit}.


Departing from the usual dual certificate argument, our proof of the optimality of F2F relaxation 
relies on delicate primal analysis. In particular, we show that  
$\langle w, x - x^* \rangle < 0$ for any vertex (extremal point) of the F2F polytope $x \neq x^*$ with high probability.
It is known that  F2F polytope is not integral in the sense that some of its vertices is fractional. Fortunately, it turns out for any vertex $x$, its fractional entry $x_e$ must be $1/2$. Thanks to this \emph{half-integrality} property,
we can encode the difference $y \triangleq 2(x-x^*)$ as a
bicolored multigraph $G_y$ with a total weight $w(G_y)=2\langle w, x - x^* \rangle$. 
Finally, we bound $w(G_y)$ via a \emph{divide-and-conquer} argument by first decomposing $G_y$ into an edge-disjoint union of graphs in a family with simpler structures 
and then proving that for every graph $H$ in this family, its total weight $w(H)$ is negative with high probability under condition \prettyref{eq:ITlimit}.
Our decomposition of $G_y$ heavily exploits the fact that $G_y$ is a balanced multigraph in the sense that every vertex has an equal number of incident red
edges and blue edges, and the classical graph-theoretic result that every connected balanced multigraph has an Eulerian
circuit with edges alternating in colors. 




\subsection{$2$-factor (2F) Integer Linear Programming Relaxation}
\label{sec:2f}


The $2$-factor (2F) Integer Linear Programming  relaxation of the TSP is 
\begin{align}
\widehat{x}_{\rm 2F} = \arg \max_{x}  & \;   \iprod{w}{x}   \label{eq:2factor}  \\
\text{s.t.  }  & \;    x \left(  \delta(v) \right) =2 \nonumber \\  
& \; x_e \in \{ 0, 1\}. \nonumber 
\end{align}

The $2$-factor ILP is the same as the F2F LP \prettyref{eq:F2F} except that the $x_e$'s have integrality constraints, and is therefore a tighter relaxation of the original TSP than F2F LP. As a warm-up for the optimality proof of F2F LP, 
we provide a much simpler proof, showing that the optimal
solution of the 2F ILP coincides with the true cycle $x^*$ with high probability, under the same condition that
$\alpha_n - \log n \to +\infty$. 
We note that although it is not an LP,  the $2$F ILP can be solvable in $O(n^4)$ time using
a variant of the blossom algorithm \cite{edmonds65,edmonds1965maximum,letchford2008odd}; however, 
finding an efficient and scalable implementation of the blossom algorithm can be challenging in practice~\cite{Grötschel1987}.


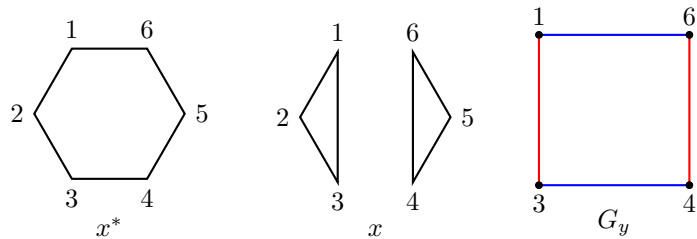
\begin{figure}[ht]%
\centering
\begin{tikzpicture}[scale=1,every edge/.append style = {thick,line cap=round},font=\small]
\draw[thick] (0:1) node[right] {$5$} --(60:1) node[above] {$6$} -- (120:1) node[above] {$1$} -- (180:1) node[left] {$2$} -- (240:1) node[below] {$3$} -- (300:1) node[below] {$4$} --cycle;
		\node at (0,-1.5) {$x^*$};
 \end{tikzpicture}
~~~
\begin{tikzpicture}[scale=1,every edge/.append style = {thick,line cap=round},font=\small]
\draw[thick]  (0:1) node[right] {$5$} --(60:1) node[above] {$6$} -- (300:1) node[below] {$4$} --cycle;
\draw[thick]  (120:1) node[above] {$1$}  -- (180:1) node[left] {$2$}  -- (240:1) node[below] {$3$}  --cycle;
		\node at (0,-1.5) {$x$};
 \end{tikzpicture}
~~
\begin{tikzpicture}[scale=1,every edge/.append style = {thick,line cap=round},font=\small]
  \node[coordinate,label=above:$6$] (v2) at (0,1)  {};
  \node[coordinate,label=above:$1$] (v1) at (-2,1)  {};
  \node[coordinate,label=below:$4$] (v3) at (0,-1)  {};
  \node[coordinate,label=below:$3$] (v4) at (-2,-1)  {};
 \redsingleedge{(v2)}{(v3)}
    \bluesingleedge{(v1)}{(v2)}
    \bluesingleedge{(v3)}{(v4)}
    \redsingleedge{(v1)}{(v4)}
    \node at (-1,-1.5) {$G_y$};
 \end{tikzpicture}
\caption{The ground truth $x^*$ is a cycle of length $6$, $x$ is a feasible solution to \prettyref{eq:2factor} corresponding to two disjoint triangles, and the graph $G_y$ for $y=x-x^*$ is an alternating 4-cycle.}%
\label{fig:2fexample}%
\end{figure}

Let $x$ denote the adjacency vector of a given $2$-factor. 
To prove that $x^*$ is the unique optimal solution to the 2F
ILP, it suffices to show $\iprod{w}{x-x^*}
<0$ for the adjacency vector of any $2$-factor 
$x \neq x^*$. To capture the difference between $x$ and $x^*,$
we define $y \in  \{0,\pm 1\}^{\binom{n}{2}}$ by
\begin{align}
y=x-x^*. \label{eq:y_def_2f}
\end{align}
Define a simple graph $G_y$ with bicolored edge
whose adjacency matrix is $|y|$ with isolated vertices removed and each edge is colored red if
$y_e=-1$ and blue if $y_e=+1$ (see \prettyref{fig:2fexample} for an example).
Furthermore, for a given bicolored graph $B$, 
we define its weight as
$$
w(B) \triangleq  \sum_{\text{blue $e\in E(B)$}} w_e - \sum_{\text{red $e \in E(B)$}}  w_e \, .
$$
Then $w(G_y)= \iprod{w}{x-x^*}$. 

A bicolored graph is \emph{balanced} if
 for every vertex the number of red incident edges is equal to the number of blue incident edges.
 Since $x(\delta(v)) =2$ and $x^*(\delta(v))=2$ for every vertex $v$, it follows that $y(\delta(v))=0$
and thus $G_y$ is balanced.
Define
\begin{align*}
\calB &=\left\{ B: \text{$B$ is a simple, connected, and balanced bicolored graph} \right\} \\
\calB^* & =\{B \in \calB: V(B) \subset [n], x_e^*=1 \text{ for every red edge $e \in E(B)$} \},
\end{align*}
where $V(B)$ and $E(B)$ denote the vertex set and edge set of $B$, respectively. 

Let $B_1, \ldots, B_m$ denote the connected components of $G_y$. 
Since each connected component of $G_y$ is balanced, it follows that 
$B_i \in \calB^*$ and 
$$
w\left( G_y \right)  = \sum_{i=1}^m w(B_i).
$$
Hence to show $w(G_y)<0$ for all possible $G_y$, it reduces to proving that
$w(B)<0$ for all $B \in \calB^*$. 

Fix an even integer $\ell \ge 4$ and let 
$$
\calB^*_\ell = \{ B \in \calB^*: |E(B)|=\ell\}.
$$
Fix any $B\in \calB^*_\ell$. By the balancedness, $B$ has $\ell/2$ red 
edges and $\ell/2$ blue  edges. Hence, 
$$
w(B) \eqdistr  \sum_{i=1}^{\ell/2} Y_i - \sum_{i=1}^{\ell/2} X_i ,
$$
where $X_i$'s and $Y_i$'s are independent sequences of random variables such that  $X_i$'s are i.i.d.\ copies of $\diff P/\diff Q$ under distribution $P$ and $Y_i $'s are i.i.d.\ copies of $\diff P/\diff Q$
under distribution $Q$; the notation $\eqdistr$ denotes equality in distribution. 
It follows from the Chernoff's inequality (cf.\ the large-deviation bound \prettyref{eq:LDXY} in \prettyref{app:LD}) that 
\begin{align}
\prob{ w(B) \ge 0} = \prob{ \sum_{i=1}^{\ell/2} Y_i - \sum_{i=1}^{\ell/2} X_i \ge 0}  \le \exp \left( -\alpha_n \ell/2 \right). \label{eq:chernoff_renyi}
\end{align}

Next we claim that there are at most $(2n)^{\ell/2}$ different graphs $B$ in $\calB_\ell^*$.
The proof of the claim is deferred to the end of this section. 
Combining the union bound with \prettyref{eq:chernoff_renyi} gives that 
$$
\prob{ \max_{B \in \calB_\ell^* } w(B) \ge 0  } \le |\calB_\ell^*| \exp \left( -\alpha_n \ell/2 \right)
\le \exp \left\{ - \left(\alpha_n - \log (2n) \right)  \ell /2 \right\}. 
$$
Taking another union bound over all integers $\ell \ge 4,$ 
we get the desired result. 
\begin{align*}
\prob{ \max_{B \in \calB^*} w(B) \ge 0} & \le \sum_{\ell=4}^\infty
\prob{ \max_{B \in \calB_\ell^* } w(B) \ge 0  } \\
& \le \sum_{\ell=4}^\infty \exp \left\{ - \left(\alpha_n - \log (2n) \right) \ell /2 \right\} \\
& \le \frac{\exp \left\{ -2 \left(\alpha_n - \log (2n) \right) \right\} }{1- \exp \left\{ - \left(\alpha_n-\log (2n) \right)/2 \right\} } \overset{\prettyref{eq:ITlimit}}{\to} 0.
\end{align*}

We are left to show $|\calB_\ell^*| \le (2n)^{\ell/2}$. 
This follows from the following classical graph-theoretic result that every connected balanced multigraph 
$G$ has an \emph{alternating} Eulerian circuit, that is, the edges in the circuit alternate in colors. 
%
\begin{lemma}
\label{lmm:eulerian}  
 Every connected balanced bicolored multigraph $G$ has an alternating Eulerian circuit.
\end{lemma}
The lemma is proved in~\cite[Theorem 1]{Kotzig1968} in a more general form 
(see also \cite[Corollary 1]{Pevzner95}).
For completeness, we provide a short proof in \prettyref{app:eulerian}. 

In view of \prettyref{lmm:eulerian}, for every $B \in \calB_\ell^*$,
it must have a Eulerian circuit $T$ given by the sequence of
 $(v_0,v_1, \ldots, v_{\ell-1},v_{\ell}=v_0)$ of vertices (vertices may repeat) 
 such that $v_i \in [n]$ and $(v_{i},v_{i+1})$ is
a red edge for even $i$ and blue edge for odd $i$ in $B$. Let $\calT$ denote
the set of all possible such Eulerian circuits. Moreover, every Eulerian circuit $T \in \calT$ uniquely
determines a $B \in \calB_\ell^*$, because the vertex set $V(B)$ is the union
of vertices $v_i$'s and the colored edge set $E(B)$ is the union of colored edges $(v_{i},v_{i+1})$'s in
$T$. Hence, $|\calB_\ell^*| \le |\calT|$. To enumerate all possible $T \in \calT$, 
it suffices to enumerate all the possible labelings of vertices in $T.$
Recall that by definition, for every red edge $e=(v_{i},v_{i+1})$ in $T$,
$x^*_e=1$. Thus the two endpoints $v_{i}$ and $v_{i+1}$ 
must be neighbors in the true cycle corresponding to $x^*$, and hence once the vertex labeling of 
$v_{i}$ is fixed, there are at most $2$ different choices for the vertex labeling of $v_{i+1}$. 
Therefore, we enumerate all possible 
Eulerian circuits $T\in \calT$ 
by sequentially choose the vertex labeling of $v_{i}$ from $i=0$ to $i=\ell-1.$
Given the vertex labelings of $(v_0, \ldots, v_{i-1})$, 
the number of choices of the vertex labeling of $v_{i}$ is at most $n$ 
for even $i$ and $2$ for odd $i.$ Hence, $|\calT| \le (2n)^{\ell/2}$,
which further implies that $|\calB_\ell^*| \le (2n)^{\ell/2}$.

\section{Proof of Theorem \ref{thm:LP_opt}}
\label{sec:proof_theorem1}
In this section, 
we prove that the optimal solution of the fractional $2$-factor 
coincides with $x^*$ with high probability, provided that $\alpha_n-\log n \to \infty$. 
This is the bulk of the paper.

\subsection{Graph Notations}\label{sec:graph_notations}
We describe several key graph-theoretic notations used in the proof. 
We start with multigraphs. 
Formally,
a multigraph $G$ is an ordered pair $(V,E)$  with a vertex set $V=V(G)$ and
an edge \emph{multiset} $E=E(G)$ consisting of subsets of $V(G)$ of size $2$. 
Note that by definition multigraphs do not have self-loops. 
A multi-edge is a set of edges in $E(G)$ with the same end points.  
The multiplicity of an edge is its multiplicity as an element in $E(G)$.   
We call a multi-edge single and double if its edge multiplicity
is $1$ and $2$, respectively. Note that a double edge $(u,v)$ 
refers to the set of two edges connecting vertices $u$ and $v$. 
We say a multigraph $G$ is bicolored, if every distinct element in $E(G)$ is colored in 
either red or blue and the repeated copies of an element all have the same color.

 For two multigraphs $G$ and $H$ on the same set of vertices,
we define $G-H$ to be the multigraph induced by the edge multiset $E(G) \setminus E(H)$.
The union of multigraphs $G$ and $H$  is the multigraph $G\cup H$ with vertex set $V(G) \cup V(H)$ and 
edge multiset $E(G) \cup 
E(H)$.\footnote{Here the union of multisets is defined 
so that the multiplicity of each elements adds up. For example, $\{a,a,b\} \cup \{a,b,c\}=\{a,a,a,b,b,c\}$.}
By definition, the multiplicity of an element in $E(G) \cup E(H)$ is the sum of its
multiplicity in $E(G)$ and $E(H)$.  When $E(G) \cap E(H) =\emptyset$, $G\cup H$ is called an edge-disjoint union. When 
 $V(G) \cap V(H) =\emptyset$, $G\cup H$ is called an vertex-disjoint union.

 A \emph{walk} in a multigraph $G$ is a sequence  $(v_0, v_1, \ldots, v_m)$ of vertices (which may repeat)   such that
$(v_{i-1},v_i) \in E(G)$ for $ 1 \le i\le m$.
A \emph{trail} in a multigraph $G$ is a walk 
$(v_0, v_1, \ldots, v_m)$ such that for all $1 \le i \le m$, the number of times that edge $(v_{i-1},v_i)$ appears in the walk is
no more than its edge multiplicity in $E(G)$. 
 %
 A trail is \emph{closed} if the starting and ending vertices are the same. 
 A \emph{circuit} is a closed trail. 
 An \emph{Eulerian trail} in a multigraph $G$ is a trail $(v_0, v_1, \ldots, v_m)$ such that
for every $e \in E(G)$, the number of times that it appears in the trail coincides with its edge multiplicity in 
$E(G)$. An \emph{Eulerian circuit} is a closed Eulerian trail. 
A \emph{path} is a trail with no repeated vertex. A \emph{cycle} consists a path
plus an edge from its last vertex to the first.



\subsection{Proof Outline} 
 Let $x^*=(x^*_e)$ denote the adjacency vector of the hidden Hamiltonian cycle (ground truth). 
The feasible set of the F2F LP \prettyref{eq:F2F} is the \emph{F2F polytope}:
\begin{equation}
\calQ \triangleq 
\sth{ x \in [0,1]^{\binom{n}{2}}: 
x(\delta(v)) =2, \forall v \in [n]}.
\label{eq:F2Fpoly}
\end{equation}

To prove that $x^*$ is the unique optimal solution to the F2F LP with high probability,
it suffices to show that  $\langle w, x - x^* \rangle < 0$ holds with high probability for any vertex (extremal 
point) of the F2F polytope $x$ other than $x^*$.   It turns out that the vertices 
of  the F2F polytope $\calQ$ has the following simple characterization~\cite{balinski1965integer,BC99,SWZ13}. 
First of all, for any vertex $x$, its fractional entry must be a half-integer, i.e., 
\begin{equation}
x_e \in \{0,1/2,1\}, \quad \forall e. 
\label{eq:halfint}
\end{equation} 
Furthermore, if we define the \emph{support graph} of $x$ 
as be the graph with vertex set $[n]$
and edge set $\{e: x_e \neq 0\}$, then
each connected component of
the support graph of $x$  must be one of the following two cases: it is
either
\begin{enumerate}
\item a cycle of at least three vertices with 
$x_e=1$ for all edges $e$ in the cycle, 
\item or consists of an even number of odd-sized cycles  with $x_e=1/2$ for all edges $e$ in the cycles that are connected by
paths of edges $e$ with $x_e=1$.  In this case, if we remove the edges in the odd cycles, the resulting graph is a spanning disjoint set of paths formed by edges $e$ with $x_e=1$. 
\end{enumerate}
It turns out that, among the aforementioned characterizations of the vertices of the F2F polytope, 
our analysis of the LP relaxation uses only the half-integrality property \prettyref{eq:halfint}.

\begin{figure}[H]
\centering
\begin{tabular}{cc}
 \begin{tikzpicture}[scale=1,every edge/.append style = {thick,line cap=round}]
 \bluesingleedge{(-1,2)}{(1,2)}
 \redsingleedge{(-1,2)}{(-1,-2)}
 \redsingleedge{(1,2)}{(1,-2)}
 \bluesingleedge{(-1,-2)}{(1,-2)}
   \redsingleedge{(-1,2)}{(0,1)}
   \redsingleedge{(1,2)}{(0,1)}
     \redsingleedge{(-1,-2)}{(0,-1)}
   \redsingleedge{(1,-2)}{(0,-1)}
   \bluesingleedge{(0,1)}{(0,-1)}
    \node at (0,0) [left]  {$1$};
     \node at (-1,0) [left]  {$1$};
     \node at (1,0) [right]  {$1$};
       \node at (0,2) [below]  {$\frac{1}{2}$};
       \node at (0,-2) [above]  {$\frac{1}{2}$};
       \node at (-0.3,1.2) [left]  {$\frac{1}{2}$};
       \node at (0.3,1.2) [right]  {$\frac{1}{2}$};
       \node at (-0.3,-1.2) [left]  {$\frac{1}{2}$};
       \node at (0.3,-1.2) [right]  {$\frac{1}{2}$};			
\end{tikzpicture}
\qquad \qquad &
\begin{tikzpicture}[scale=1,every edge/.append style = {thick,line cap=round}]
 \bluesingleedge{(3,2)}{(5,2)}
  \bluesingleedge{(3,-2)}{(5,-2)}
   \redsingleedge{(3,2)}{(4,1)}
   \redsingleedge{(5,2)}{(4,1)}
   \redsingleedge{(3,-2)}{(4,-1)}
   \redsingleedge{(5,-2)}{(4,-1)}
   \bluedoubleedge{(4,1)}{(4,-1)}
 \end{tikzpicture}
\\
(a) \qquad \qquad  & (b)  
\end{tabular}
\caption{(a): The support graph of a fractional 
vertex $x$ of the F2F polytope with
$n=6$. The edges in the support graph of $x^*$
are highlighted in red. (b): The multigraph representation of $y=2(x-x^*)$.}
\end{figure}
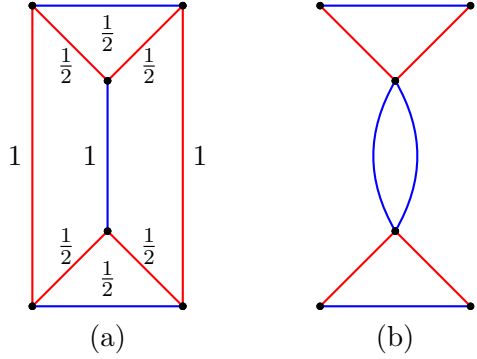

To capture the difference between a given vertex $x$ and the true solution $x^*$, we use the following multigraph representation: define $y \in \reals^{\binom{n}{2}}$ by 
\begin{equation}
y = 2(x-x^*),
\label{eq:y2x}
\end{equation}
with $y_{e} = 2 \left( x_{e }-x^*_{e} \right) \in \{0, \pm 1, \pm 2\}$.
Define a multi-graph $G_y$ whose adjacency matrix is $|y|$ with isolated vertices removed 
and each edge $e$ is colored red if $y_e<0$ or blue if $y_e>0$.
In particular, the edge multiplicity of $G_y$ is at most $2$. Comparing to \prettyref{eq:y_def_2f}, the extra factor of 2 in \prettyref{eq:y2x} is to ensure that $y$ is still integral;
as a consequence, 
$G_y$ may be a multigraph with multiplicity $2$ instead of a simple graph. 
For any given bicolored multigraph $F$, we define its weight as  
$$
w(F) \triangleq \sum_{\text{blue $e \in E(F)$}} w_e - \sum_{\text{red $e \in E(F)$}}  w_e \, ,
$$
where the summation above includes all repeated copies of $e$ in $E(F).$
Then $w(G_y)=\iprod{w}{x-x^*}$. Hence, to prove that $x^*$ is the unique optimal solution to the LP program with high probability, 
it reduces to showing that $w(G_y) <0$ for all possible $G_y$ constructed from the
extremal point $x \neq x^*$ with high probability. 

Instead of first calculating the probability of $w(G_y) \leq 0$ and then taking a union bound on all possible $G_y$,
our proof crucially relies on a decomposition of $G_y$ into 
some suitably defined simpler graphs. In the next subsection, we will describe a family 
$\calF^*$ of graphs, and show that every  possible $G_y$ can be decomposed as a union of
graphs in $\calF^*$: $G_y = \cup_{i=1}^m F_i$ with $F_i \in \calF^*$ for each $1 \le i \le m$. 
Since the multiplicity of $e$ in $E(G_y)$ is equal to the sum of multiplicities of $e$ in $E(F_i)$ over $ i \in [m]$, it follows that 
$$
w(G_y) = \sum_{i=1}^m w (F_i).
$$
Therefore, to  show  that $w(G_y) <0$ for all possible $G_y$ with high probability, it suffices
to show $w(F) <0$ for all graphs $F$ in family $\calF^*$.

We remark that in the analysis of the 2F ILP in \prettyref{sec:2f}, we have $y=x-x^*$ as opposed to $y=2(x-x^*)$ and thus $G_y$ is a balanced simple graph. Consequently, we can simply decompose $G_y$ into its connected components which are connected, balanced simple graphs. In contrast, here the decomposition of $G_y$ as a multigraph is much more sophisticated due to the existence of double edges. In particular, the weight of a double edge in $F$ appears twice in $w(F)$ and hence its variance 
is twice the total variance of two independent edge weight. For this reason,
to control the deviation of $w(F)$ from its mean, it is essential to account for the contribution of double edges and single edges separately, which, in turn, requires us to separate the double edges from single edges in our decomposition.

\subsection{Edge Decomposition} \label{sec:edge_decomposition}

Our decomposition of $G_y$ relies on the notion of balanced multigraph and alternating Eulerian circuit. 
A bicolored multi-graph is \emph{balanced} if for every vertex the number of red incident edges is equal to the number of blue incident edges. 
Since $x(\delta(v)) =2$ and $x^*(\delta(v))=2$ for every vertex $v$, it follows that $y(\delta(v))=0$
and thus $G_y$ is balanced. As a result, the vertices in $G_y$ all have even degree (in fact, either $2,4,6$ or $8$). Therefore each connected component of $G_{y}$ has an Eulerian circuit. 
Recall that Eulerian circuit is alternating if the edges in the Eulerian circuit alternate in colors. 
In view of \prettyref{lmm:eulerian}, each connected component of $G_{y}$ has an Eulerian circuit. 
In the remainder, we suppress the subscript $y$ in $G_y$ whenever
the context is clear.

Next, we describe a family $\calF$ of graphs and show that $G$ is a union of graphs in this family. 
First, we need to introduce a few notations. 
For any pair of two vertices $u, v$ in graph $G$,
vertex identification (also known as vertex contraction) produces a graph
by removing all edges between $u,v$ and replacing $u,v$ with a single vertex $w$ incident
to all edges formerly incident to either $u$ or $v$. 
When $u$ and $v$ are adjacent, \ie, sharing two endpoints of edge
$e,$ vertex identification specializes to edge contraction of $e$ and the resulting graph is
denoted by $G \cdot e$; visually, $e$ shrinks to a vertex. 
Note that edge contraction may introduce multi-edges.
 We define a \emph{stem} as a path $(v_0, v_1, \ldots, v_{k-1})$ for some $k$
 distinct vertices such that $(v_{i-1},v_i)$ is a double edge for all $1\le i \le k$
 and the double edges alternate in colors.
 The two endpoints $v_0$ and $v_{k-1}$ of the stem are identified as the tips of the stem. 
  We say a tip of the stem is red if it is incident to the red double edge; otherwise we say it is blue. 
  Given a stem and an even cycle  
 $C_0$ consisting of only single edges of alternating colors,
 we define the following \emph{blossoming} procedure to connect the stem with $C_0$: 
 first contract any single blue (red) edge in $C_0$ to a vertex $v$ 
 and attached to $v$ the stem by identifying $v$ with a blue (red) tip of the stem. 
 The resulting graph known as \emph{flower} has an alternating circuit and 
 the contracted $C_0$ is called \emph{blossom}. 
 The tip of the stem not incident to the blossom is called the tip of the flower. 
 We say a flower is red (blue) if its tip is red (blue).
 For example, a red flower are shown 
	\prettyref{fig:sub}.
Similar notions of stem, flower, blossom were introduced
 in~\cite{edmonds65} in the context of simple graphs.

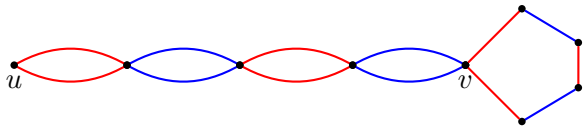
\begin{figure}[ht]
	\centering
	\begin{tikzpicture}[scale=1.5,every edge/.append style = {thick,line cap=round},baseline=(anc.base)]
	\reddoubleedge{(-1,0)}{(0,0)}
\bluedoubleedge{(0,0)}{(1,0)}
\reddoubleedge{(1,0)}{(2,0)}
\bluedoubleedge{(2,0)}{(3,0)}
\node at (-1,0) [below] {$u$};
\node at (3,0) [below] {$v$};
\altred
\draw (3,0) node[dot] {} edge (3.5,0.5) node[dot] {} edge (4,0.2) node[dot] {} edge (4,-0.2) node[dot] {} edge (3.5,-0.5) node[dot] {} edge (3,0) node[dot] {};
\node (anc) at (0,-1) {};
\end{tikzpicture}
	\caption{A red flower consisting of a stem of $4$ alternating double edges followed by a blossom of $5$ single edges.}
	\label{fig:sub}
\end{figure}


Then we introduce a family $\calU$ of balanced graphs.  
We start with an even cycle $G_0$ in alternating colors.
At each step $t \ge 1$, construct a new balanced graph $G_t$ from $G_{t-1}$ as follows.
Fix any cycle consisting of at least $4$ edges in $G_{t-1}$.
In this cycle, pick any edge and
apply the  following \emph{flowering} procedure: 
contract the red (blue) edge to one vertex $w$ and attach to $w$ a flower by identifying $w$ with 
the root of a blue (red) flower.\footnote{Since we allow contracting an 
edge incident to a stem, it is possible to have a vertex with 
multiple stems attached.}
Let $\calU$ denote the collection of all graphs 
obtained from applying the flowering procedure recursively for finitely many times.
In particular, $\calU$ includes all even cycles in alternating colors.
For example, the graph in \prettyref{fig:ex} is in $\calU$ which is obtained by
starting with a $10$-cycle and applying the flowering procedure $4$ times.
By construction, any graph $H  \in \calU$ must contain even $\ell \ge 4$ number of single edges. 

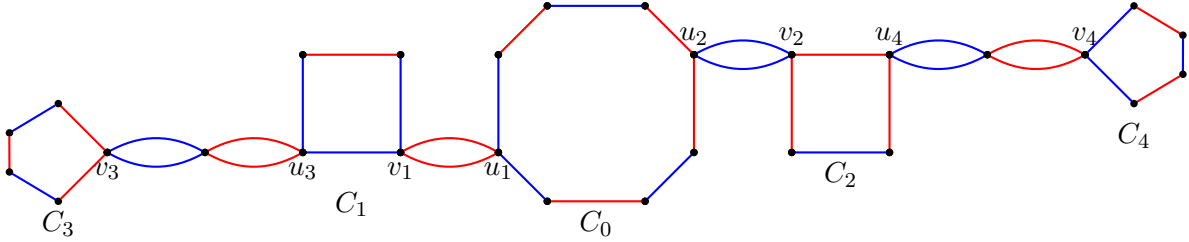
\begin{figure}[ht]%
\centering
\begin{tikzpicture}[scale=1.3,every edge/.append style = {thick,line cap=round},baseline=(anc.base)]
	\bluesingleedge{(-0.5,1)}{(0.5,1)}
	\redsingleedge{(0.5,1)}{(1,0.5)}
	\redsingleedge{(1,0.5)}{(1,-0.5)}
	\bluesingleedge{(1,-0.5)}{(0.5,-1)}
	\redsingleedge{(0.5,-1)}{(-0.5,-1)}
		\bluesingleedge{(-0.5,-1)}{(-1,-0.5)}
			\bluesingleedge{(-1,-0.5)}{(-1,0.5)}
	\redsingleedge{(-1,0.5)}{(-0.5,1)}
	\node at (0, -1)[below]{$C_0$};
	\node at (1,0.5) [above] {$u_2$};
	\bluedoubleedge{(1,0.5)}{(2,0.5)}
	\node at (2,0.5) [above] {$v_2$};
	\redsingleedge{(2,0.5)}{(3,0.5)}
	\redsingleedge{(3,0.5)}{(3,-0.5)}
	\bluesingleedge{(3,-0.5)}{(2,-0.5)}
	\redsingleedge{(2,-0.5)}{(2,0.5)}
	\node at (2.5, -0.5)[below]{$C_2$};
	
	\node at (-1,-0.5) [below]{$u_1$};
	\reddoubleedge{(-1,-0.5)}{(-2,-0.5)}
	\node at (-2,-0.5) [below]{$v_1$};
	\bluesingleedge{(-2,-0.5)}{(-3,-0.5)}
	\bluesingleedge{(-3,-0.5)}{(-3,0.5)}
	\redsingleedge{(-3,0.5)}{(-2,0.5)}
	\bluesingleedge{(-2,0.5)}{(-2,-0.5)}
	\node at (-2.5, -0.8)[below]{$C_1$};

	\node at (3,0.5) [above]{$u_4$};
	\bluedoubleedge{(3,0.5)}{(4,0.5)}
	\reddoubleedge{(4,0.5)}{(5,0.5)}
	\node at (5,0.5) [above]{$v_4$};
	\node at (5.5, -0.1)[below]{$C_4$};
	
	\node at (-3,-0.5) [below]{$u_3$};
	\reddoubleedge{(-4,-0.5)}{(-3,-0.5)}
	\bluedoubleedge{(-5,-0.5)}{(-4,-0.5)}
	\node at (-5,-0.5) [below]{$v_3$};
	\node at (-5.5, -1)[below]{$C_3$};
	
	\altblue
	\draw (3+2,0+0.5) node[dot]{} edge (3.5+2,0.5+0.5) node[dot]{} edge (4+2,0.2+0.5) node[dot]{} edge (4+2,-0.2+0.5) node[dot]{} edge (3.5+2,-0.5+0.5) node[dot]{} edge (3+2,0+0.5) node[dot]{};
	\draw (-5,-0.5) node[dot]{} edge (-5.5,0) node[dot]{} edge (-6,-0.3) node[dot]{} edge (-6,-0.7) node[dot]{} edge (-5.5,-1) node[dot]{} edge (-5,-0.5) node[dot]{};

\end{tikzpicture}
\caption{Example of graph in family $\calU$}%
\label{fig:ex}%
\end{figure}

\begin{center}
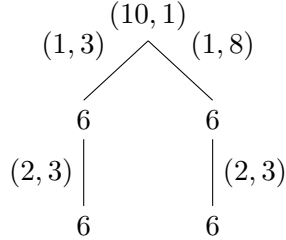
\begin{figure}[ht]%
\centering
\begin{tikzpicture}[level distance=40pt,sibling distance=36pt]
\Tree [.$(10,1)$ 
\edge node[auto=right]{$(1,3)$}; 
[.$6$ 
\edge node[auto=right]{$(2,3)$}; 
$6$ ] 
\edge node[auto=left]{$(1,8)$};
[.$6$ 
\edge node[auto=left]{$(2,3)$}; 
$6$ ]]
\end{tikzpicture}
\caption{Tree representation of graph in \prettyref{fig:ex}.}%
\label{fig:tree}%
\end{figure}
\end{center}

Alternatively, note that each graph in the family $\calU$ can be viewed as 
cycles interconnected by stems. 
Thus, we can represent $G_t$ using a tree $T_t$,
whose nodes correspond to even cycles and 
links corresponds to stems. 
Next we describe this enumeration scheme in details. 
Every node in the tree represents a cycle in alternating colors, 
with a mark $\ell$ being the length of the cycle. 
The cycle corresponding to the root node is assumed to have 
a fixed ordering of edges, and the root node has an extra mark which is $1$ if  
the color of the first edge in the corresponding cycle is blue and $0$ otherwise. 
Every link $(u,v)$ represents a stem consisting of only double edges of alternating colors 
with mark $(k, i)$, where $k$ is the length of the stem, and $i$ is the index of the contracted edge
of the parent vertex $u$. 
We start with tree $T_0$ with a single root node corresponding to $G_0=C_0$ with the mark being the length of 
$C_0$. At each step  $t \ge 1$, we view the flowering procedure as 
growing to a new tree $T_t$ as follows. For any
vertex $u$ in $T_{t-1}$ that corresponds to an alternating cycle $C$,
a new vertex $v$ that corresponds to an alternating cycle $C'$, and 
a stem, we connect $u$ and $v$ with a edge corresponding to the stem.
The edge is marked with $(k,i)$, where $k$ is the length of the stem, and $i$
is the index of the contracted edge in $C$. Then the edges in the alternating cycle $C'$
is indexed by $1, 2,\ldots$ by starting from the contracted edge in $C'$ and traversing
 $C'$ in clockwise direction.


Finally, we need to introduce the
notion of homomorphism between two bicolored multigraphs $H$ and $F$. There exist multiple definitions of 
homomorphism between multigraphs; here we follow the convention in \cite[Section 5.2.1]{Lovasz12}.
A node-and-edge homomorphism $H \to F$ is a pair of 
vertex map $\phi: V(H) \to V(F)$ and \emph{bijective}\footnote{Let $A$ and $B$ denote two multisets. 
Let $A'$ and $B'$ denote the set of distinct elements in $A$ and $B$, respectively. We say $\psi: A \to B$
is bijective if $\psi: A' \to B'$ is bijective and for every element $a \in A'$, the multiplicity of $a$ in $A$
is the same as the multiplicity of $\psi(a)$ in $B$. For example, if $A= \{a,a, b, c\}$ and $ B= \{x, x, y, z\}$. Let
$\psi(a)=x$, $\psi(b)=y$, and $\psi(c)=z$, then $\psi: A \to B$ is bijective.} edge map 
$\psi: E(H) \to E(F)$ such that if $e \in E(H)$ connects $i$ and $j$, then
$\psi(e)$ connects $\phi(i)$ and $\phi(j)$ and has the same color as 
$e$.  We say $H$ is homomorphic to $F$ if such a node-to-edge homomorphism exists.
By construction, an edge $e$ is incident to $u$ in $H$ 
if and only if $\psi(e)$ is incident to
$\phi(u)$ in $F$. Therefore, if $H$ is homomorphic to $F$, then they are either both balanced or both unbalanced.
Moreover, since $\psi$ is bijective,  $H\to F$ is edge-multiplicity preserving, \ie, 
the multiplicity of $\psi(e)$ in $E(F)$ is the same  
as the multiplicity of $e$ in $E(H)$. Hence, the number of double  (single) edges in $H$ and $F$ are the same.
Furthermore, note that for two node-and-edge homomorphisms $(\phi,\psi): H \to F$
and $(\phi, \psi'): H \to F'$ with the same vertex map $\phi$, it holds that $F=F'$. 
Hence, when the context is clear, we simply write $\phi: H \to F$ or $\phi(H)=F$ 
by suppressing the underlying edge map.


Let $\calF$ denote the collection of all graphs $F$ such that
$H \to F $ for some $H \in \calU$. In particular, $\calF \supseteq \calU$, and this inclusion is strict as the example in \prettyref{fig:homo} shows.
\begin{figure}[ht]%
\centering
\begin{tikzpicture}[scale=1,every edge/.append style = {thick,line cap=round},font=\small]
	\node[coordinate,label=above:$2$] (v2) at (0,0)  {};
	\node[coordinate,label=above:$1$] (v1) at (-1,1)  {};
	\node[coordinate,label=below:$3$] (v3) at (-1,-1)  {};
	\node[coordinate,label=above:$9$] (v9) at (1,1)  {};
	\node[coordinate,label=below:$8$] (v8) at (1,-1)  {};
	\node[coordinate,label=right:$11$] (v11) at (1,0)  {};
	\node[coordinate,label=above:$10$] (v10) at (2.5,1)  {};
	\node[coordinate,label=below:$7$] (v7) at (2.5,-1)  {};
	\node[coordinate,label=left:$12$] (v12) at (2.5,0)  {};
	\node[coordinate,label=above:$4$] (v4) at (3.5,0)  {};
	\node[coordinate,label=above:$5$] (v5) at (4.5,1)  {};
	\node[coordinate,label=below:$6$] (v6) at (4.5,-1)  {};
   \reddoubleedge{(v2)}{(v11)}
	\reddoubleedge{(v4)}{(v12)}
	\redsingleedge{(v9)}{(v10)}
	\redsingleedge{(v8)}{(v7)}
	\redsingleedge{(v3)}{(v1)}
	\redsingleedge{(v6)}{(v5)}
		\bluesingleedge{(v1)}{(v2)}
		\bluesingleedge{(v2)}{(v3)}
		\bluesingleedge{(v11)}{(v9)}
		\bluesingleedge{(v11)}{(v8)}
		\bluesingleedge{(v12)}{(v10)}
		\bluesingleedge{(v12)}{(v7)}
		\bluesingleedge{(v4)}{(v5)}
		\bluesingleedge{(v4)}{(v6)}
		\node at (1.5,-2) {$H$};
 \end{tikzpicture}
~~~\raisebox{6em}{{\Large $\xrightarrow{~~\phi~~}$}}~~~
\begin{tikzpicture}[scale=1,every edge/.append style = {thick,line cap=round},font=\small]
	\node[coordinate,label=above:$2$] (v2) at (0,0)  {};
	\node[coordinate,label=above:$1$] (v1) at (-1,1)  {};
	\node[coordinate,label=below:$3$] (v3) at (-1,-1)  {};
	\node[coordinate,label=above:$9$] (v9) at (0.5,1)  {};
	\node[coordinate,label=below:$8$] (v8) at (0.5,-1)  {};
	\node[coordinate,label=right:$11$] (v11) at (1,0)  {};
	\node[coordinate,label=above:$10$] (v10) at (1.5,1)  {};
	\node[coordinate,label=below:$7$] (v7) at (1.5,-1)  {};
	\node[coordinate,label=above:$4$] (v4) at (2,0)  {};
	\node[coordinate,label=above:$5$] (v5) at (3,1)  {};
	\node[coordinate,label=below:$6$] (v6) at (3,-1)  {};
   \reddoubleedge{(v2)}{(v11)}
	\reddoubleedge{(v4)}{(v11)}
	\redsingleedge{(v9)}{(v10)}
	\redsingleedge{(v8)}{(v7)}
	\redsingleedge{(v3)}{(v1)}
	\redsingleedge{(v6)}{(v5)}
		\bluesingleedge{(v1)}{(v2)}
		\bluesingleedge{(v2)}{(v3)}
		\bluesingleedge{(v11)}{(v9)}
		\bluesingleedge{(v11)}{(v8)}
		\bluesingleedge{(v11)}{(v10)}
		\bluesingleedge{(v11)}{(v7)}
		\bluesingleedge{(v4)}{(v5)}
		\bluesingleedge{(v4)}{(v6)}
		\node at (1,-2) {$F$};
 \end{tikzpicture}
\caption{Example of a graph $F$ in $\calF$ but not in $\calU$. Here $F$ is homomorphic to $H\in\calU$, where the homomorphism $\phi$ maps both vertices $11$ and $12$ to $11$.}%
\label{fig:homo}%
\end{figure}
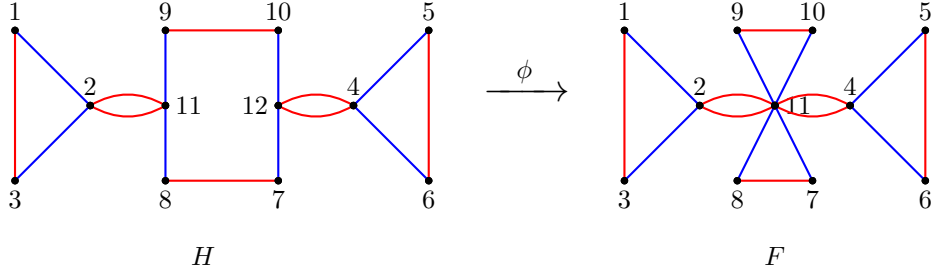
The next lemma shows that $\calF$ includes all connected balanced simple graphs.
This result serves as the base case of the induction proof of the decomposition lemma.

\begin{lemma} \label{lmm:decom_simple}
An alternating cycle is homomorphic to any connected balanced simple graph $G$ with equal number of edges.
In particular, $G \in \calF$. 
\end{lemma}
\begin{proof}
By \prettyref{lmm:eulerian}, $G$ has an Eulerian circuit 
$T = (v_0, v_1, \ldots, v_{m-1}, v_m=v_0)$ of alternating colors 
where $m$ is the total number of edges in $G$ and vertices $v_i$'s may repeat.
Let $C$ denote any alternating cycle with $m$ edges. 
We write $C=(u_0, u_1, \ldots, u_{m-1}, u_m=u_0)$ such that 
the edge $(u_0,u_1)$ has the same color as $(v_0,v_1)$. 
Then we define a pair of vertex and edge map $(\phi, \psi)$ from $C$ to $G$
such that $\phi(u_i )=v_i$ and $\psi\left((u_i, u_{i+1} ) \right)=( v_{i},v_{i+1} )$ for all $0 \le i \le m-1$. 
Since both $G$ and $C$ are simple graphs,  $\psi:E(C) \to E(G)$ is bijective.
Hence, $(\phi, \psi):C \to G$ is a node-and-edge homomorphism and the conclusion follows. 
\end{proof}

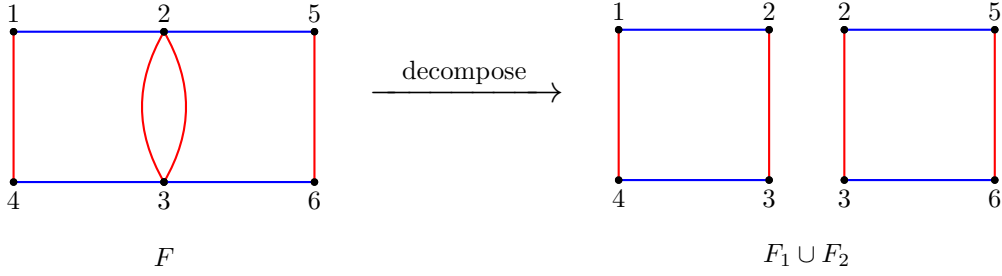
\begin{figure}[ht]%
\centering
\begin{tikzpicture}[scale=1,every edge/.append style = {thick,line cap=round},font=\small]
  \node[coordinate,label=above:$2$] (v2) at (0,1)  {};
  \node[coordinate,label=above:$1$] (v1) at (-2,1)  {};
  \node[coordinate,label=below:$3$] (v3) at (0,-1)  {};
  \node[coordinate,label=below:$4$] (v4) at (-2,-1)  {};
  \node[coordinate,label=above:$5$] (v5) at (2,1)  {};
  \node[coordinate,label=below:$6$] (v6) at (2,-1)  {};
   \reddoubleedge{(v2)}{(v3)}
    \bluesingleedge{(v1)}{(v2)}
    \bluesingleedge{(v3)}{(v4)}
    \bluesingleedge{(v2)}{(v5)}
    \bluesingleedge{(v3)}{(v6)}
    \redsingleedge{(v1)}{(v4)}
    \redsingleedge{(v5)}{(v6)}
    \node at (0,-2) {$F$};
 \end{tikzpicture}
~~~\raisebox{6em}{{\Large $\xrightarrow{~~\text{decompose}~~}$}}~~~
\begin{tikzpicture}[scale=1,every edge/.append style = {thick,line cap=round},font=\small]
  \node[coordinate,label=above:$2$] (v2) at (0,1)  {};
  \node[coordinate,label=above:$1$] (v1) at (-2,1)  {};
  \node[coordinate,label=below:$3$] (v3) at (0,-1)  {};
  \node[coordinate,label=below:$4$] (v4) at (-2,-1)  {};
  \node[coordinate,label=above:$2$] (v7) at (1,1)  {};
  \node[coordinate,label=below:$3$] (v8) at (1,-1)  {};
  \node[coordinate,label=above:$5$] (v5) at (3,1)  {};
  \node[coordinate,label=below:$6$] (v6) at (3,-1)  {};
 \redsingleedge{(v2)}{(v3)}
 \redsingleedge{(v7)}{(v8)}
    \bluesingleedge{(v1)}{(v2)}
    \bluesingleedge{(v3)}{(v4)}
    \bluesingleedge{(v7)}{(v5)}
    \bluesingleedge{(v8)}{(v6)}
    \redsingleedge{(v1)}{(v4)}
    \redsingleedge{(v5)}{(v6)}
    \node at (0.5,-2) {$F_1 \cup F_2$};
 \end{tikzpicture}
\caption{Example of a graph $F$ which is not in $\calF$, but 
can be decomposed as $F=F_1 \cup F_2$ with $F_1,F_2 \in \calF$.}%
\label{fig:F_not_in_family}%
\end{figure}

In contrast to connected balanced simple graphs, if a balanced graph $G$ contains double edges then certainly no alternating
cycle is homomorphic to $G$. 
What's more, it is possible that $G$ is not homomorphic to any graph  in the class $\calU$, \ie, $G$ may not belong to $\calF$. See \prettyref{fig:F_not_in_family} for such an example. 
Nevertheless, the next lemma shows that 
if $G$ has edge multiplicity at most $2$, then it 
can be decomposed as a \emph{union of elements} in $\calF$.

	\begin{lemma}[Decomposition]
	\label{lmm:decomp}	
    Every balanced multigraph $G$ with edge multiplicity at most $2$ 
    can be decomposed as a union of elements in $\calF$.
	\end{lemma}
	

\begin{proof}
It suffices to prove the lemma for the case where 
$G$ is connected. 
By \prettyref{lmm:eulerian}, $G$ has an Eulerian circuit 
$T = (v_0, v_1, \ldots, v_{m-1}, v_m=v_0)$ of alternating colors 
where $m$ is the total number of edges and vertices $v_i$ may repeat.  
We proceed by induction on the number of double edges $k$ in $G$.
If $k=0$, $G$ is simple. By \prettyref{lmm:decom_simple}, 
$G \in \calF$ and thus the conclusion holds. 



Suppose the conclusion holds for $k \ge 0$, we aim to prove it also holds for $k+1$. 
We call a double edge $(u,v)$ bidirectional if the Eulerian
circuit $T$ traverses it twice in two different directions, \ie,
there exists an $s$ and $t \ge s+3 $ such that $(v_{s-1}, v_{s}) = (u,v)$ and $(v_{t-1},v_{t})=(v,u)$.
Similarly, we call a double edge $(u,v)$ unidirectional if $T$
traverses it twice in the same direction. We divide our induction into two cases according
to whether there exists a unidirectional double edge.

\paragraph{Case 1: There exists at least one unidirectional double edge.}
Let $(u,v)$ be an arbitrary unidirectional double edge. See \prettyref{fig:F_not_in_family} for
an illustration. By definition,
in the alternating Eulerian circuit $T=(v_0,v_1,\ldots, v_{m-1}, v_0)$, 
there must exist an $s$ and $t \ge s+3$, such that $(v_{s-1}, v_{s}) =(v_{t-1}, v_{t})=(u,v)$.
Since $T$ is alternating, it follows that $T'\triangleq (v_{s-1}, v_{s}, \ldots, v_{t-1} ) $ is also an alternating  circuit.
Also define $T-T'$ to be the resulting graph by deleting the edges in the circuit $T'$ from $T$, \ie,
$$
T-T' = \left(v_0,v_1,v_{s-1}, v_{t}, \ldots, v_{m-1}, v_0 \right).
$$
It follows that $T-T'$ is also a circuit in alternating colors. 
Hence both $T'$ and $T-T'$ are balanced multigraphs with
edge multiplicity at most $2$. Also, since the two edges connecting $u$ and $v$
appear separately in $T'$ and $T-T'$, both 
$T'$ and $T-T'$ have at most $k$ double edges.
Applying the induction hypothesis to each connected component of $T'$ and $T-T'$, 
we conclude that both $T'$ and $T-T'$ can be decomposed as a union of elements in 
$\calF$. Hence, $T$ can be decomposed as a union of elements in $\calF$.




\paragraph{Case 2: All double edges are bidirectional.}
 In this case, we pick an arbitrary bidirectional double edge $(u,v)$. 
In the Eulerian circuit $T$, we find a trail $S$ of maximal length
 that contains the double edge $(u,v)$ and consists of only double edges in alternating colors (see \prettyref{fig:decomp_first_layer}). 
 More precisely, we find the smallest $s$ and the largest $t \ge s+1$ such that
$(v_{i-1}, v_{i})$ all have edge multiplicity $2$ for $ s \le i \le t $, and $(v_{i-1}, v_i)=(u,v)$ for 
some $ s \le i \le t$, and there exists $t' \ge t+3$ and $s' \ge t'$ such that
$$
\left( v_{t' } , v_{t'+1}, \ldots, v_{s'-1}  \right) = \left( v_{t }, v_{t-1}, \ldots, v_{s-1} \right), 
$$
in which the Eulerian circuit $T$ traverses $S$ in two different directions.


By the maximality of $S$, 
$v_{t+1}$ and $v_{t'-1}$ must be  two distinct vertices. 
Since the Eulerian circuit $T$ has alternating colors, $( v_t, v_{t+1})$ and $(v_{t'-1}, v_{t'})$
must have the same color, and thus the circuit 
$$
R=(  v_t, v_{t+1}, \ldots, v_{t'}=v_t )
$$ 
must have
an odd length. Entirely analogously,  $v_{s-2} $ and $ v_{s'}$ must be two distinct vertices,
and $(v_{s-2}, v_{s-1})$ and $(v_{s'-1}, v_{s'})$ must be of the same color; 
thus the circuit 
$$
L= ( v_0, \ldots, v_{ s-1 }, v_{ s' }, \ldots, v_{m-1}, v_m=v_0)
$$ 
must
have an odd length. In particular, both $L$ and $R$ cannot be empty. Therefore we have 
$T=L \cup S  \cup R$ as shown in \prettyref{fig:decomp_first_layer}.
We further consider two subcases according to whether $L$ and $R$ share any edge or not.
\begin{figure}[ht]%
\centering
\begin{tikzpicture}[scale=1.5,every edge/.append style = {thick,line cap=round},baseline=(anc.base)
]
  \reddoubleedgeD{(-1,0)}{(0,0)}
\bluedoubleedgeD{(0,0)}{(1,0)}
\reddoubleedgeD{(1,0)}{(2,0)}
\bluedoubleedgeD{(2,0)}{(3,0)}
\node at (-0.9,0) [below=.1cm] {$v_{s'-1}$};
\node at (-0.9,0) [above=.1cm] {$v_{s-1}$};
\node at (3,0) [below=.1cm] {$v_{t'}$};
\node at (3,0) [above=.1cm] {$v_t$};
\redsingleedgeD{(3,0)}{(3.5,0.5)}
\redsingleedgeD{(3.5,-0.5)}{(3,0)}
\node at (3.5,0.5) [above] {$v_{t+1}$};
\node at (3.5,-0.5) [below] {$v_{t'-1}$};
\path (3.5,0.5) edge[blue,dashed,bend left=110] (3.5,-0.5); 
\bluesingleedgeD{(-1.5,0.5)}{(-1,0)}
\bluesingleedgeD{(-1,0)}{(-1.5,-0.5)}
\node at (-1.5,0.5) [above] {$v_{s-2}$};
\node at (-1.5,-0.5) [below] {$v_{s'}$};
\path (-1.5,0.5) edge[red,dashed,bend right=110] (-1.5,-0.5); 
\node at (3.5,0)[right=0.5cm]{$R$};
\node at (-1.5,0) [left=0.5cm] {$L$};
\draw[dashed] (-1.2,-0.5) rectangle (3.2,0.5);
\node at (1,-0.6) {$S$};
\end{tikzpicture}
\caption{The decomposition $T=L\cup S \cup R$ where $L$ and $R$ are circuits of alternating colors shown in dashed lines}%
\label{fig:decomp_first_layer}%
\end{figure}
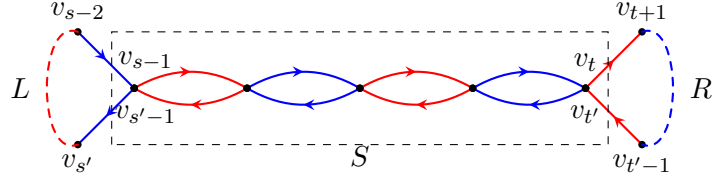

{\bf Case 2.1: $E(L) \cap E(R) \neq \emptyset$.}
In this case, there exists a bidirectional 
double edge whose two simple edges appear separately in $E(L)$ and $E(R)$.
In particular, there exist $ p \in \{t+1,\ldots,t'\}$  and
$q \in \{1,\ldots, s-1\} \cup \{s',\ldots,m\}$ 
such that $(v_{p-1}, v_p)=(v_{q}, v_{q-1})$.

Define a circuit $T'=(v_{s-1}, \ldots, v_{p-1}, v_{q-1}, v_{q-2}, \ldots, v_{s'-1})$ by first traversing 
$S$ to $R$ until reaching $v_{p-1}=v_q$ and then traversing $L$ in the reverse direction until reaching
$v_{s'-1}=v_{s-1}$. Note that $T'$ is an alternating circuit. It follows that both $T'$ and
$T-T'$ are balanced multigraphs
with edge multiplicity at most $2$. Moreover, for each double edge in $S$, its two simple edges appear separately in $T'$ and $T-T'$.
Thus, both $T$ and $T-T'$ have at most $k$ double edges. 
Applying the induction hypothesis to each connected component of $T'$ and $T-T'$, 
we conclude that both $T'$ and $T-T'$ can be decomposed as a union of elements in 
$\calF$. Hence, $T$ can be decomposed as a union of elements in $\calF$.

{\bf Case 2.2:  $E(L) \cap E(R) = \emptyset$.} In this case, $T=L \cup S  \cup R$ is an edge-disjoint
union. 


\underline{Step 1}. First we argue that, without loss of generality, we can and will assume that 
$(v_t, v_{t+1})$, $(v_{t'-1}, v_{t'})$, $(v_{s-2}, v_{s-1})$, and $(v_{s'-1}, v_{s'})$ in \prettyref{fig:decomp_first_layer} are all single edges.
Below we only consider $(v_t, v_{t+1})$. Suppose $(v_t, v_{t+1})$ is not a single edge. Then by assumption, it must be a bi-directional double edge. 
Thus there exists a $t''\in (t,t')$
such that  $(v_{t''-1}, v_{t''}) = (v_{t+1}, v_t)$.
Thus $v_t=v_{t'}=v_{t''}$.
Let $T'=(v_0, \ldots, v_{t''},v_{t'+1}, \ldots, v_0)$ and $T-T'=(v_{t''},\ldots,v_{t'})$ (cf.~\prettyref{fig:counter_example}). Then $T'$ is balanced and thus $T-T'$ is also balanced.
Moreover, $T-T'$ has at most $k$ double edges. Applying the induction hypothesis to each connected component of $T-T'$,
we conclude that $T-T'$ can be decomposed as a union of elements in $\calF$. Then it remains to
show $T'$ can be decomposed as a union of elements in $\calF$. Note that  in $T'$ the trail of maximal length that
contains the double edge $(u,v)$ is given by $(v_{s-1}, \ldots, v_{t+1})$. Redefine $S$ by including
the double edge $(v_t, v_{t+1})$ and redefine $v_t$ and $R$ accordingly. 
Applying the above procedure in finitely many times ensures that $(v_t, v_{t+1})$ becomes a single edge. 

\begin{figure}
\centering
\begin{tikzpicture}[scale=1.5,every edge/.append style = {thick,line cap=round},baseline=(anc.base),font=\small]
	\draw[dashed](-1.5,0)--(-1,0);
	\reddoubleedgeD{(-1,0)}{(0,0)}
\bluedoubleedgeD{(0,0)}{(1,0)}
\reddoubleedgeD{(1,0)}{(2,0)}
\bluesingleedgeD{(1,0)}{(1.5,-0.5)}
\redsingleedgeD{(1.5,-0.5)}{(1,-1)}
\bluesingleedgeD{(1,-1)}{(0.5,-0.5)}
\redsingleedgeD{(0.5,-0.5)}{(1,0)}
\bluesingleedgeD{(2,0)}{(2.5,0.5)}
\redsingleedgeD{(2.5,0.5)}{(2.5,-0.5)}
\bluesingleedgeD{(2.5,-0.5)}{(2,0)}
\begin{pgfonlayer}{background}
        \draw[rounded corners=0.5em,line width=0.5em,black!20,cap=round]
								(1,0)-- (1.5,-0.5) -- (1,-1) -- (0.5,-0.5) --cycle;
    \end{pgfonlayer}
\node at (1,0.1) [above] {$v_t$};
\node at (1,-0.1) [below] {$v_{t''}$};
\node at (0,0.1) [above] {$v_{t-1}$};
\node at (0,-0.1) [below] {$v_{t'+1}$};
\node at (2,0.1) [above] {$v_{t+1}$};
\node at (2,-0.1) [below] {$v_{t''-1}$};
\end{tikzpicture}
\caption{Example of the decomposition in Step 1 in Case 2.2. The highlighted circuit is $T-T'$, which needs to be decomposed so that the double edge $(v_t,v_{t+1})$ can be included in the stem $S$.}
\label{fig:counter_example}
\end{figure}
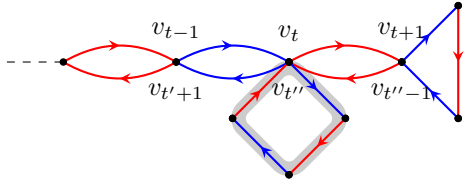


\underline{Step 2}. 
Next we capitalize on the fact that $T=L \cup S  \cup R$ is an edge-disjoint
union to complete the proof of the decomposition.

In $R$, we expand the vertex $v_t$ to an edge $e=(u_t, v_t)$ whose color is different from the
color of edge $(v_t, v_{t+1})$. In particular, we add a distinct vertex $u_t$ and edge $e$, and
reconnect $v_{t'-1}$ from $v_{t'}$ to $u_t$.  Let $R'$ be the resulting circuit, \ie,
$R' = (  u_t, v_t, v_{t+1}, \ldots, v_{t'-1}, u_t)$. 
Then $R'$ is a circuit of alternating colors.  Hence $R'$ is a balanced multigraph
with edge multiplicity at most $2$.  Moreover, $R'$ has at most $k$ double edges.
Applying the induction hypothesis on each connected component of $R'$, we conclude that 
$R'$ can be decomposed as a union of elements in $\calF$. 
In this union, denote by $P$ the element in $\calF$ that contains the edge $e$. 
Since $u_t$ is only incident to $v_t$ and $v_{t'-1}$, it follows that $P$ must also contain
the edge $(v_{t'-1},u_t)$. Furthermore,  since $(v_{t'-1},v_{t'})$ is a single edge in $R$,
it follows that $v_{t'-1}$ and $v_{t}$ are disconnected in $P$
and thus $e$ is contained in a cycle of length at least $4$ in $P$. 
Let $P\cdot e$ be the resulting graph by contracting $e$ in $P$ to  vertex $v_t$.



Analogously, in $L$, we expand the vertex $v_{s-1}$ to an edge $e'=(u_{s-1}, v_{s-1})$ 
whose color is different from the color of edge $(v_{s-2}, v_{s-1})$. Let $L'$ be the resulting circuit, \ie,
\[
L' = (v_0, \ldots, v_{s-2}, u_{s-1}, v_{s-1}, v_{s'}, \ldots, v_{m-1},v_0).
\]
 Then 
$L'$ is a circuit of alternating colors. By the same argument as in the case of $R'$, we
get that $L'$ can be decomposed as a union of elements in $\calF$.
In this union, denote by $Q$ the element in $\calF$ that contains the edge $e'$. 
Then $e'$ is contained in a cycle of length at least $4$ in $Q$. 
Let $Q\cdot e'$ be the resulting graph by contracting $e'$ in $Q$ to  vertex $v_{s-1}$. 

To show $T$ is an  union of elements in $\calF$, 
it suffices to show that $M \triangleq (Q\cdot e') \cup S \cup (P \cdot e) \in \calF$. 
Note that vertices may repeat in $S$. Nevertheless, there exists a 
homomorphism $(\phi_S, \psi_S)$ such that $S_0 \to S$ for some stem $S_0$. Let $w'$ denote an endpoint of $S_0$ such that 
$\phi_S(w')=v_{s-1}$ and  $w$ denote the other endpoint of $S_0$ such that $\phi_S(w)=v_{t}$.
Moreover, since $P \in \calF$,
it follows that there exists a homomorphism $(\phi_P, \psi_P)$ such that $P_0 \to P$
for some $P_0 \in \calU$. Similarly, there exists a homomorphism $(\phi_Q, \psi_Q)$ such that
$Q_0 \to Q$ for some $Q_0 \in \calU$. 
Note that $e$ (resp.~$e'$) is a simple edge in $P$ (resp.~$Q$).
Therefore, there exist simple edge $e_0 \in E(P_0)$ and $e'_0 \in E(Q_0)$) such that $\psi_P(e_0)=e$ and $\psi_Q(e_0')=e' $.
By the tree representation of graphs in $\calU$,  $P_0$ can be represented as a 
tree $I$ with vertex $i$ corresponding to the cycle $C_i$ that contains the edge $e_0$.
Entirely analogously, $Q_0$ can be represented as a tree $J$ with vertex $j$
corresponding to the cycle $C_j$ that contains the edge $e'_0$. 
We form a new tree $K$ by connecting vertex $i$ and  vertex $j$ with an edge corresponding to 
 stem $S_0$ and identifying the root vertex of $I$ as the root vertex of $K$.    
 In other words, let $Q_0 \cdot e'_0$ denote the graph obtained by contracting $e'_0$ in $Q_0$ to vertex $w'$,
and $P_0 \cdot e_0$ denote the graph obtained by contracting $e_0$ in $P_0$ to vertex $w$.
Note that $e_0$ (resp.~$e_0'$) is contained in a cycle of length at least $4$ in $P_0$ (resp.~$Q_0$).
Let $M_0 \triangleq (Q_0 \cdot e'_0) \cup S_0 \cup (P_0 \cdot e_0)$.  
Then $K$ is a tree representation of  $M_0$ starting from a cycle represented by the root vertex of tree $I$
and thus $M_0 \in \calU$.




Finally, it remains to show that $M_0$ is homomorphic to $M$.
Since the vertex sets $V(P_0)$, $V(Q_0)$, and $V(S_0)$ are disjoint,
we can define a vertex map $\phi: V(M_0)\to V(M)$ as follows: 
for any vertex $v \in V(M_0)$, 
\begin{align*}
\phi(v) &=
  \begin{cases}
   \phi_P(v)      & \text{if } v \in V(P_0) \\
   \phi_Q(v)       & \text{if } v \in V(Q_0) \\
   \phi_S(v) & \text{if } v \in V(S_0).
  \end{cases}
  \end{align*}
  Moreover, since $M_0=(Q_0 \cdot e'_0) \cup S_0 \cup (P_0 \cdot e_0)$ is an edge-disjoint union,  we can define an edge map $\psi: E(M_0)\to E(M)$ as follows: 
	for  any edge $e \in E(M_0)$,
\begin{align*}
\psi(e) &=
  \begin{cases}
   \psi_P(e)     & \text{if } e \in E\left(P_0 \cdot e_0 \right) \\
   \psi_Q(e)       & \text{if } e \in E\left(Q_0 \cdot e'_0 \right) \\
   \psi_S(e) & \text{if } e \in E(S_0).
  \end{cases}
  \end{align*}
By assumption, $E(L) \cap E(R) =\emptyset$. 
Therefore, $(Q\cdot e') \cup S \cup (P \cdot e)$ 
is also an edge-disjoint union. As a consequence, 
$\psi: E(M_0) \to E(M)$
is bijective, and for any $e \in E(M_0)$, 
the multiplicity of $\psi(e)$ is the same as the multiplicity of $e$.
Thus $(\phi,\psi): M_0 \to M$ is a homomorphism. Hence,
$M\in \calF$, concluding the proof.
\end{proof}

For $k \ge 0$ and $\ell \geq 3$, we define $\calU_{k,\ell} \subset \calU$
as the bicolored balanced multigraphs $H \in \calU$ with \emph{$k$ double edges} and \emph{$\ell$ single edges}.
The following lemma upper bounds the number of unlabeled graphs  in $\calU_{k,\ell}$.

	\begin{lemma}[Enumeration of isomorphism classes]
	\label{lmm:enum1}	
	Let $k \geq 0$ and even $\ell \geq 4$. 
  Then the number of unlabeled graphs in $\calU_{k,\ell}$ is at most $17^{k} 4^{\ell}$.
	\end{lemma}	
\begin{proof}
If $k=0$, then by \prettyref{lmm:decom_simple}, all graphs in $\calU_{0,\ell}$ are isomorphic to an alternating $\ell$-cycle. Thus
the number of unlabeled graphs in $\calU_{0,\ell}$ is $1$.  Hence the conclusion trivially holds.  
Next we focus on the case where $k \ge 1$. Note that if two unlabeled multigraphs in $\calU$ have the same tree representation, 
they must be the same.  
Hence it suffices to upper bound 
the number of possible tree representations of multigraphs $H \in \calU_{k,\ell}$.


Fix a marked tree $T$ of $m$ vertices that is a tree representation of
 a multigraph $H \in \calU_{k,\ell}$. Since $k \ge 1$, it follows that $2 \le m \le k+1$. 
Starting with the root node, we order the nodes and links in $T$ via breadth first search. 
For $i \in [m]$, let $\ell_i$ denote the length
of the cycle represented by node $i$ 
and $S_i$ denote the set of indices of contracted edges in the cycle.
For $j \in [m-1]$, let $k_j$ the length of the stem represented by link $j$.
Let $c$ denote the extra mark of the root vertex. 
Note that from $\{ c, (\ell_i, S_i)_{i=1}^m, (k_j)_{j=1}^m \}$, we
can uniquely determine the marked tree $T$. In particular, $\ell_i$
determines the mark of vertex $i$ and $k_j$ determines the first mark of 
edge $j$. Moreover, $S_i$ determines the number of children of node $i$ 
and the second mark on every link connecting $i$ to its child. 
Therefore, to bound the number of all possible marked trees $T$, 
it suffices to bound the number of all possible
$\{ c, (\ell_i, S_i)_{i=1}^m, (k_j)_{j=1}^{m-1} \}$.
Note that 
$$
\sum_{i=1}^m \ell_i = \ell + 2(m-1), \quad \sum_{j=1}^{m-1} k_j = k. 
$$
Hence, the number of possible choices of sequences $(\ell_1, \dots, \ell_m)$ is at most $\binom{\ell+2m-3}{m-1}$,
and the number of possible choices of sequences $(k_1, \dots, k_{m-1})$ is at most $\binom{k-1}{m-2}$. 
Moreover, for each $i$, there are at most $2^{\ell_i}$
different choices of $S_i$. Hence, the number of all possible
$\{ c, (\ell_i, S_i)_{i=1}^m, (k_j)_{j=1}^m \}$ is at most $2^{\ell+2m-1} \binom{\ell+2m-1}{m-1} \binom{k-1}{m-2}$.
Therefore, 
the number of possible tree representations of multigraphs $H \in \calU_{k,\ell}$ 
is at most 
\begin{align*}
\sum_{m = 2}^{k+1} 2^{\ell+2m-1}  \binom{\ell+2m-3}{m-1} \binom{k-1}{m-2} 
&\le \sum_{m = 2}^{k+1} 2^{\ell+2m-1}   2^{\ell+2m-3} \binom{k -1}{m-2} \\
& = 2^{2\ell+4} \sum_{m = 2}^{k+1} 2^{4(m-2)} \binom{k-1}{m-2}  \le 2^{2\ell+4} 17^{k-1} \le 4^\ell 17^k.
\end{align*}

\end{proof}


Define\footnote{The constraint that $|E(F)| \leq 4n$ is due to the fact that, for any vertex $x$ of the F2F polytope, the multigraph $G_y$ obtained from $y=2(x-x^*)$ has maximal degree at most 8.}
\[
 \calF^* = \{ F \in \calF: V(F) \subset [n] \text{ and } \left|E(F)\right| \le 4n \text{ and for every red edge $e \in E(F),$ $x^*_e=1$ } \}.
\]
Given $H \in \calU$, we say a homomorphism $\phi: H \to F$ is \emph{compatible with $x^*$} if 
$\phi(H) \in \calF^*$.
Denote by $\Phi^\ast_{H}$ the set of all homomorphism $\phi: H \to F$ that is compatible with $x^*$.
Then 
$$
\calF^\ast = \left\{ \phi(H): \phi \in \Phi^\ast_{H}, H \in \calU  \right \}.
$$






In the following, we upper bound the number of elements in $\Phi^\ast_{H}$ for a given
$H \in \calU$. We need to set up a few notations. Let $H_d$ and $H_s$ denote the subgraph of $H$ induced by all the double edges and all the single edges, respectively. 
Then we have an edge-disjoint union $H=H_d \cup H_s$. 
For a vertex map $\phi \in \Phi^*_{H}$, let $\phi_d$ and $\phi_s$ 
denote $\phi$ restricted to $V(H_d)$ and $V(H_s)$, respectively.
Note that $\phi_d (v) = \phi_s(v)$ for all $v \in V(H_d) \cap V(H_s)$.
We write $\phi=(\phi_s,\phi_d)$.  

\begin{lemma}[Enumeration of homomorphisms]
	\label{lmm:enum2}	
Let $ k \ge 0$ and let $\ell \geq 4$ be an even integer.
Fix a bicolored balanced multigraph $H \in \calU_{k,\ell}$.
 \begin{itemize}
\item There exists an integer $0 \le r \le \ell/2$ such that
\begin{align}
\log |\Phi^*_{H_d}|
\le \frac{1}{2} (k+r) \log (2n),  \label{eq:count_2_double}
\end{align}
where
$$
\Phi^*_{H_d} \triangleq \left\{ \phi_d: \exists \phi_s, \text{ s.t. } (\phi_d,\phi_s) \in \Phi^*_{H} \right\};
$$
\item For any fixed vertex map $\phi_d: V(H_d) \to [n]$,
\begin{align}
\log \left| \left\{ \phi_s:  \left(\phi_s, \phi_d \right) \in \Phi^*_{H} \right\} \right|
\le \left( \ell/2-r\right)  \log n +  \left( \ell/2+k  \right) \log 2. \label{eq:count_2_single}
\end{align}

\end{itemize}
\end{lemma}


\begin{proof}
 Let $n_0$ denote the number of stems  and $m_0$ denote the
number of distinct vertices that are tips of some stem in $H$. 
Then $n_0 \le k$ and $m_0 \le 2n_0$. Recall that in the tree representation of $H$, each link corresponds to a stem in $H$ and
each node corresponds to a cycle. Hence, there are $n_0+1$ nodes
in the tree. Each cycle after contraction will have at least $3$ single edges in $H$,
and thus the total number of single edges $\ell \ge 3(n_0+1)$.

We will use the following crucial property in enumerating the vertex maps  $\phi: V(H)\to [n]$ 
that are compatible with $x^*$. 
Given a red edge $(u,v)$, $\phi(u)$ and $\phi(v)$ must be neighbors in the cycle 
corresponding to $x^*$.   
Therefore, once the label $\phi(u)$  is fixed, there are at most 
two different choices for $\phi(v)$. 
For every vertex  in $H$, we will
assign a mark in a certain order such that the
mark on a vertex is an upper bound on 
the number of choices of its labeling given the labelings of 
previous vertices. In particular, 
for every red edge $(u,v)$, if one endpoint has been assigned a mark before the other endpoint, 
then we assign mark $2$ to the other endpoint.

First, we count the vertex maps for double edges $\phi_d: V(H_d) \to [n]$ that are
compatible with $x^*$. Fix a stem $S=(v_0, \ldots, v_{s})$ of  
$s$ double edges in alternating colors. We sequentially choose
  the vertex labelings from $v_0$ to $v_{s}$. 
  We distinguish three types of stems.

\begin{enumerate}[Type 1:]
	\item Both tips are red. In this case, 
	 the number of different values of $\phi_d(v_i)$ given $\left( \phi_d(v_0), \ldots, \phi_d(v_{i-1}) \right)$  is 
at most  $n$ for even $i$ and $2$ for odd $i$.  See \prettyref{fig:stemtype1}.



 \item The two tips have different colors.
 Without loss of generality, assume $(v_0,v_1)$ is a red double edge.   
 Again  the number of choices of $\phi_d(v_i)$ given $\left( \phi_d(v_0), \ldots, \phi_d(v_{i-1}) \right)$ 
is at most  $n$ for even $i$ and $2$ for odd $i$. 
See \prettyref{fig:stemtype2}.

\item Both tips are blue. 
In this case,  
  the number of choices of $\phi_d(v_0)$  is 
at most  $n$, and the remaining part of stem $(v_1,\ldots,v_s)$ is of Type II. 
Thus, for $i \ge 1$, the number of possible maps $\phi_d(v_i)$ is given $\left( \phi_d(v_0), \ldots, \phi_d(v_{i-1}) \right)$ at most  $n$ for odd $i$ and $2$ for even $i$. 
See \prettyref{fig:stemtype3}.

\end{enumerate}
In summary, the number of possible vertex labelings for a stem of length $s$ of Type $i$ is at most 
\[
n^{\frac{s+i}{2}} 2^{\frac{s+2-i}{2}} , \quad i=1,2,3.
\]
 
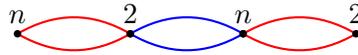
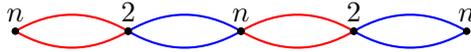
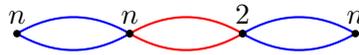
\begin{figure}[ht]
\centering
\subfigure[Type 1.]{\label{fig:stemtype1}
\begin{tikzpicture}[scale=1.5,every edge/.append style = {thick,line cap=round}]
  \reddoubleedge{(-1,0)}{(0,0)}
\bluedoubleedge{(0,0)}{(1,0)}
\reddoubleedge{(1,0)}{(2,0)}
\node at (-1,0) [above] {$n$};
\node at (0,0) [above] {$2$};
\node at (1,0) [above] {$n$};
\node at (2,0) [above] {$2$};
\end{tikzpicture}
}

\subfigure[Type 2.]{\label{fig:stemtype2}
\begin{tikzpicture}[scale=1.5,every edge/.append style = {thick,line cap=round}]
  \reddoubleedge{(-1,0)}{(0,0)}
\bluedoubleedge{(0,0)}{(1,0)}
\reddoubleedge{(1,0)}{(2,0)}
\bluedoubleedge{(2,0)}{(3,0)}
\node at (-1,0) [above] {$n$};
\node at (0,0) [above] {$2$};
\node at (1,0) [above] {$n$};
\node at (2,0) [above] {$2$};
\node at (3,0) [above] {$n$};
\end{tikzpicture}
}

\subfigure[Type 3.]{\label{fig:stemtype3}
\begin{tikzpicture}[scale=1.5,every edge/.append style = {thick,line cap=round}]
  \bluedoubleedge{(-1,0)}{(0,0)}
\reddoubleedge{(0,0)}{(1,0)}
\bluedoubleedge{(1,0)}{(2,0)}
\node at (-1,0) [above] {$n$};
\node at (0,0) [above] {$n$};
\node at (1,0) [above] {$2$};
\node at (2,0) [above] {$n$};
\end{tikzpicture}
}
\caption{Three types of stems. The number on  each vertex shows the number of different choices for its labeling under $\phi_d$.
}
\label{fig:stemtype}
\end{figure}


 Suppose in $H$ there are 
$n_i$ stems of Type $i$, for $i=1,2,3$.
Then the total number of 
 different vertex maps for $H_d$ is at most 
\[
n^{(k+n_1+2n_2+3n_3)/2} 2^{(k+n_1-n_3)/2}
\] 
This bound can be further improved by taking into account the fact that some of the tips are either identical or connected by a red single edge.
In the following we give a tighter upper bound by incorporating these constraints.
 
Let $\ell_r$ and $\ell_b$ (resp.~$k_r$ and $k_b$) denote the number of red and blue single (resp.~double) edges in $H$ respectively. 
 Then we have $k_r+k_b=k$
 and $k_r-k_b=n_1-n_3$. Hence, 
\begin{align}
k_r =  ~ \frac{1}{2}(k+n_1-n_3), \quad k_b = ~\frac{1}{2}(k -n_1+n_3).
\label{eq:krkb}
\end{align}
Analogously, $\ell_r+\ell_b =\ell$ and  $2k_r+\ell_r=2k_b+\ell_b$ by the balancedness of $H$. Hence,
 \begin{align}
\ell_r = \frac{\ell}{2} -n_1+n_3, \quad \ell_b = ~\frac{\ell}{2} +n_1-n_3.
\label{eq:lrlb}
\end{align}
Moreover, the number of red (resp.~blue) tips is $2n_1+n_2$ (resp.~$n_2+2n_3$). 
To count the total number of vertex maps for $H_d$, it suffices to count the labelings of tip and non-tip vertices separately. We will count  the labelings sequentially: 
first for tip vertices and then for non-tip vertices. 

\paragraph{Tips.}

Without loss of generality, we assume 
there is no cycle in $H$ consisting of only red single edges. Suppose, for the sake of contradiction, there is a cycle $(u_0, \ldots, u_{\ell_1-1}, u_0)$
consisting of $\ell_1$ red single edges. By the construction of $\calU$,
each vertex is attached to at least one flower. 
For any homomorphism $\phi: H \to F$,  
$(u_0, \ldots, u_{\ell_1-1}, u_0)$ is mapped to the ground truth cycle $x^*$. 
Therefore $\ell_1=n$. A flower has at least $5$ edges. Therefore,
the total number of edges (counting with multiplicity) in $E(F)$ is at least $n+5n=6n$, and thus 
$\phi$ is not compatible with $x^*$ by definition. 
Therefore, there is no homomorphism $H \to F$ compatible with $x^*$.

In a given cycle of $H$, we define a red path to be a maximal path consisting of all red single edges. 
We define a vertex as a red path of length $0$, if it is incident to two blue single edges in $H$.
Then all cycles in $H$ are segmented by blue single edges into a total of $\ell_b$ red paths (cf.~\prettyref{fig:cycle_binary} for an illustration). 
 We distinguish two types of red paths depending on whether it is incident to a double edge or not:
 \begin{enumerate}[Type I:]
 \item The red path is not incident to any double edge. In this case, the red path must consist of only one red single edge.
 \item The red path is incident to at least one double edge. In this case, at least one of the vertex in the path is a tip of some stem. 
 \end{enumerate}

To count the vertex maps for tips of stems, it suffices to focus on red paths of Type II. 
Note that two red paths of Type II may be connected by a stem consisting of only one red double edge.
To account for the total number of constraints induced by red paths, let $n_4$  denote the total number of red paths of Type II and let $m_1$  denotes the total number of stems consisting of one red double edge.
Define a graph whose vertices are red paths of Type II, and two red paths are connected if they are connected by a one red double edge. By definition of the tree representation, such a graph must be a forest; thus the number of connected components therein is $n_4-m_1$.
Call each connected component a \emph{red component}. 
The key observation is that, for each red component, once we fix the labeling for a given vertex in it, any remaining vertex has at most $2$ labelings. 
Moreover,  each red component must contain at least one vertex as the tip of some stem, and
each tip is contained in some red component. 
Therefore, we can assign marks to tips in two steps: First, for each red component, pick an arbitrary tip and assign
a mark $n$ to it; next, we assign mark $2$ to each of the remaining tips.
Since there are $n_4-m_1$ red component, there are in total $n_4-m_1$ tips having mark $n$. 

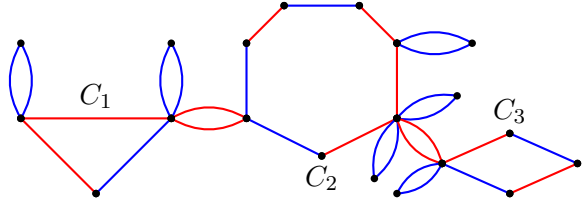
\begin{figure}[ht]
\centering
\begin{tikzpicture}[scale=1,every edge/.append style = {thick,line cap=round}]
  \bluesingleedge{(-1,1)}{(0,1)}
    \redsingleedge{(0,1)}{(0.5,0.5)}
      \redsingleedge{(0.5,0.5)}{(0.5,-0.5)}
           \redsingleedge{(0.5,-0.5)}{(-0.5,-1)}
            \bluesingleedge{(-0.5,-1)}{(-1.5,-.5)}
           \redsingleedge{(-1.5,0.5)}{(-1,1)} 
					\bluesingleedge{(-1.5,0.5)}{(-1.5,-0.5)} 
					\reddoubleedge{(-1.5,-0.5)}{(-2.5,-0.5)}
					\redsingleedge{(-2.5,-0.5)}{(-4.5,-0.5)}
					\redsingleedge{(-4.5,-0.5)}{(-3.5,-1.5)}
					\bluesingleedge{(-3.5,-1.5)}{(-2.5, -0.5)}
					\bluedoubleedge{(-2.5,-0.5)}{(-2.5,0.5)}
					\bluedoubleedge{(-4.5,-0.5)}{(-4.5,0.5)}
          \bluedoubleedge{(0.5,0.5)}{(1.5,0.5)}         
      \bluedoubleedge{(0.5,-0.5)}{(1.3,-0.2)}
       \bluedoubleedge{(0.5,-0.5)}{(0.2,-1.3)}
  \reddoubleedge{(0.5,-0.5)}{(1.1,-1.1)}
    
    \redsingleedge{(1.1,-1.1)}{(2,-0.7)}
    \bluesingleedge{(2,-0.7)}{(2.9,-1.1)}
    \redsingleedge{(2.9,-1.1)}{(2,-1.5)}
     \bluesingleedge{(2,-1.5)}{(1.1,-1.1)}
       \bluedoubleedge{(1.1,-1.1)}{(0.5,-1.5)}
       
       \node at (-3.5,-0.5) [above] {$C_1$};
       \node at (-0.5,-1) [below] {$C_2$};
      \node at (2,-0.7) [above] {$C_3$};
     
\end{tikzpicture}
%
\caption{Three cycles with attached double edges in $H$. There are one red path with length $2$ of type II in cycle $C_1$.
There are three red paths in cycle $C_2$:  one with length $1$ of type I, one with length $0$ and the other of length $3$ both of type II.
There are two red paths in cycle $C_3$: one with length $1$ of type I and the other with length $1$ of type II. 
In total, there are  $2$ red components formed by $4$ red paths of type II: 
one formed by the red path in cycle $C_1$ and the red path with length $0$ in cycle $C_2$,
and the other formed by the red path with length $3$ in cycle $C_2$ and the red path with length $1$ of type II in cycle $C_3$.
}
\label{fig:cycle_binary}
\end{figure}

\paragraph{Non-tips.}
Given the labeling of tips, we count the labeling of non-tips. Note that each non-tip vertex $u$ is incident to 
at least one red double edge. Consider two cases separately:
\begin{itemize}
\item If the red double edge is incident to a tip, 
since the labeling of tips have been assigned, there are at most $2$ possible labelings for $u$.
Hence we assign a mark $2$ to vertex $u$.

	\item If the red double edge is
in the interior of the stem, \ie, not incident to any tip, we assign  mark $n$ to one of its endpoint
and mark $2$ to the other endpoint.
Recall that $m_1$ denotes the number of stems consisting of only one
red double edge. Then there are exactly 
\[
k_r- 2(n_1-m_1) -n_2-m_1=k_r- 2n_1 -n_2+m_1 
\]
number of red
double edges in the interior of the stems. 


\end{itemize}


By now we have assigned marks -- either $n$ or $2$ -- to every vertex in $H_d$. 
Once we multiply all the marks together, we obtain an upper bound on the number of vertex maps $\phi_d$ for $H_d$.
Recall that the total number of stems is $n_0=n_1+n_2+n_3$ 
and $m_0$ denotes the number of distinct tip vertices in $H$. 
Then in total there are $k-n_0+m_0$ vertices in $H_d$. 
 Recall that there are $n_4-m_1$ marks of $n$ assigned to tips and $k_r-2n_1-n_2+m_1$ marks of $n$ assigned to non-tips. 
Hence, in total we have $k_r-2 n_1-n_2+n_4$ marks
of $n$. The rest of marks all take values $2$. 
In view of \prettyref{eq:krkb}, we have
$$
k_r - 2n_1 -n_2 + n_4 = \frac{1}{2} \left( k + n_1 - n_3\right) - 2n_1 -n_2 + n_4 = \frac{1}{2} \left( k -3n_1 -2 n_2 -n_3 + 2n_4 \right).
$$
Therefore, 
\begin{align}
\log \left| \left\{ \phi_d:  \phi \in \Phi^*_{H} \right\} \right|
 \le \frac{1}{2} \left( k -3 n_1 -2n_2 - n_3 +2 n_4 \right) \log (n/2) + \left( k-n_0 +m_0 \right) \log 2. 
\label{eq:phid0}
\end{align}

Next, we fix a vertex map $\phi_d$ for $H_d$ and 
count all the possible vertex maps for $H_s$ that
are compatible with $x^*$.
To this end, a key observation is that it suffices to count the possible labelings for the endpoints of red single edges in $H_s$. 
Indeed, for any vertex $u$ not incident to any red single edges, by the balancedness and connectedness of $H$, $u$ 
must be incident to at least one red double edge, which has been accounted for by $\phi_d$.  

For each red path of Type I, it consists of only one
red single edge and thus we assign mark $n$ to one of its endpoint and $2$ to its other endpoint. 
For each red path of Type II, since it contains at least one tip whose vertex map has already been fixed by
the vertex map for double edges, each of its other vertex has at most $2$ different vertex maps. Hence, we
assign mark $2$ to every vertex other than tips in the red path.  

Recall that there are $\ell_b - n_4$ red paths of Type I. Hence the total number of marks of value $n$, in view of \prettyref{eq:lrlb}, is
$$
\ell_b - n_4  = \frac{\ell}{2} + n_1 - n_3 - n_4 .
$$
The rest of marks are $2$. 
Recall that $m_0$ denotes the total number of distinct tip vertices in $H$, and each tip is incident to at least one single edge.
Note that there are in total $\ell$ vertices in $H_s$. Thus $m_0$ of them are tips.
Hence, there are $\ell-m_0$ vertices in $H_s$ we need to assign marks to and thus 
the total number of marks of $2$ in $H_s$ is at most $\ell-m_0 - (\ell_b-n_4)$.
In total, we have 
\begin{align}
\log \left| \left\{ \phi_s:  (\phi_s, \phi_d) \in \Phi^*_{H}  \right\} \right|
 \le &~ \left( \ell/2 + n_1 -n_3 -n_4 \right)  \log (n/2)  +\left( \ell-m_0   \right) \log 2
\nonumber \\
 =&~\left( \ell/2-r \right)  \log n  +\left( \ell/2-m_0  +r \right) \log 2, \label{eq:count_2_single_exact}
\end{align}
where we defined
$$
r \triangleq n_3+n_4 - n_1.
$$

Furthermore, recall that $n_0$ denotes the total number of stems.
Note that $n_4 \le m_0 \le 2n_0$, because every red path of Type II must contain at least one  tip,
and each stem has at most $2$ tip vertices that are distinct from the tip vertices of the other stems.  
Using $n_0=n_1+n_2+n_3$ and $n_4 \le 2n_0$, we get that $r \ge -3n_1 -2n_2 -n_3 + 2n_4$.
In view of \prettyref{eq:phid0}, we have
\begin{align}
\log \left| \left\{ \phi_d:  \phi \in \Phi^*_{H} \right\} \right|
 & \le \frac{1}{2} \left( k + r \right) \log (n/2) + \left( k-n_0 +m_0 \right) \log 2  \nonumber \\
 & =  \frac{1}{2} \left( k + r \right) \log n + \frac{1}{2} \left( k -r -2n_0+ 2m_0 \right) \log 2  \nonumber \\
 & \le   \frac{1}{2} \left( k + r \right) \log n + \frac{1}{2} \left( k  -r + m_0 \right)  \log 2,
 \label{eq:count_2_double_exact}
\end{align}
where the last inequality holds because $m_0 \le 2n_0$.

Next we show that $m_0/2 \le r \le \min\{ \ell/2, m_0+n_0 \}$.
Indeed, by definition, we have $r \le n_3 +n_4 \le n_0+m_0$. Since $ 0 \le \ell_b-n_4 = \ell/2 - r$, 
it follows that $r \le \ell/2$. Furthermore, for any red path of Type I,
it consists of precisely one
red single edge, where neither of the two endpoints is a tip. 
Thus $\ell-m_0 \ge 2(\ell_b-n_4)$. Recall that $\ell_b-n_4= \ell/2 + n_1 - n_3 - n_4=\ell/2-r$.
It follows that $r \ge m_0/2$.

Finally, the desired \prettyref{eq:count_2_double} follows from \prettyref{eq:count_2_double_exact}
in view of $m_0/2 \le r$ and the desired \prettyref{eq:count_2_single} follows from 
\prettyref{eq:count_2_single_exact} in view of $r \le m_0+n_0 \le m_0+k.$
\end{proof}


\subsection{Proof of \prettyref{thm:LP_opt}}
We prove that if 
\begin{align}
\alpha - \log n \ge 16 \log 17, \label{eq:alpha_condition}
\end{align}
then
\begin{align}
\min_{x^*} \prob{ \hat x_{\rm F2F} = x^* } \ge 1- 8 \exp \left( - (\alpha-\log n) /8 \right). \label{eq:2f2_success_prob}
\end{align}
Then \prettyref{thm:LP_opt} readily follows by taking $\alpha-\log n \to +\infty.$
By \prettyref{lmm:decomp}, 
$$
G_y= \bigcup_{i=1}^m F_i, \quad F_i \in \calF
$$ 
for each $1 \le i \le m$ and some finite $m$. 
Note that for each red edge $e$ in $G_y,$ $x^*_e=1$.
Therefore, $F_i \in \calF^*$. 
Thus, to prove \prettyref{eq:2f2_success_prob}, it suffices to show
\begin{align}
\prob{ \max_{F \in \calF^*} w(F) <0 } \ge 1- 8 \exp \left( - (\alpha-\log n) /8 \right). \label{eq:2f2_success_prob_2}
\end{align}

Fix $k \ge 0$ and $\ell \ge 4$,  define
\begin{align*}
 \calF^*_{k, \ell} = \big\{ F \in \calF^*: E(F) & \text{ consists of $k$ double edges and $\ell$ single
 edges
 } 
 \big \}.
\end{align*}
Then 
\begin{align}
\calF^*_{k,\ell} = \{ \phi(H): H \in \calU_{k,\ell} \text{ and } \phi\in \Phi^*_H \}
\label{eq:calF_kell_map}
\end{align}
and 
$$
\calF^\ast=\bigcup_{k\ge 0} \bigcup_{\ell \ge 4} 
\calF^*_{k,\ell}.
$$

In view of \prettyref{eq:calF_kell_map}, we have
\begin{align}
\max_{F \in \calF^\ast_{k,\ell} } w(F)
= \max_{H \in \calU_{k,\ell}} \max_{\phi \in \Phi^\ast_{H}}  w( \phi(H) )
\label{eq:minF_H}
\end{align}
We first show a high-probability bound to the inner maximum for a given
$H \in \calU_{k,\ell}$. Since maximizing over $\phi$ is equivalent to 
first maximizing over $\phi_d$ and then maximizing over $\phi_s$ for a fixed $\phi_d$,
it follows that 
\begin{align*}
\max_{\phi \in \Phi^\ast_{H}}  w ( \phi(H) ) 
=  \max_{\phi_d  \in \Phi^\ast_{H_d}}  
\left(  w \left( \phi_d(H_d) \right)  + 
\max_{\phi_s: (\phi_d,\phi_s) \in \Phi^\ast_{H} }  
w \left( \phi_s (H_s ) \right)  \right).
\end{align*}
Recall that $X_i$'s and $Y_i$'s are two independent sequences of random variables, where $X_i$'s are i.i.d.~copies of $\diff P/\diff Q$ under distribution $P$ and $Y_i $'s are i.i.d.~copies of $\diff P/\diff Q$ under
distribution $Q$. 
Recall that $\ell_r$ and $\ell_b$ (resp.~$k_r$ and $k_b$) denote the number of red and blue single (resp.~double) edges in $H$ respectively. Let $\delta=k_b-k_r = (\ell_r-\ell_b)/2$. Then $\delta \le \min\{k,\ell/2\}$.
In view of \prettyref{eq:krkb} and \prettyref{eq:lrlb}, for a fixed $\phi_d$, 
$$
w \left( \phi_d(H_d) \right) \eqdistr 2 \left(\sum_{i=1}^{(k+\delta)/2} Y_i -\sum_{i=1}^{(k-\delta)/2} X_i  \right)
$$ 
and for a fixed $\phi_s$, 
$$
w \left( \phi_s (H_s ) \right) \eqdistr \sum_{i=1}^{\ell/2-\delta } Y_i - \sum_{i=1}^{\ell/2+ \delta} X_i  \; ,
$$ 
where $\eqdistr$ denotes equality in distribution. 
Moreover, for a fixed $\phi_d$, $w( \phi_d(H_d) )$ is the sum of the weights on double edges,
which is independent of 
the collection of $ w \left( \phi_s (H_s ) \right)$ ranging over
all possible $\phi_s$ such that $(\phi_s,\phi_d) \in \Phi^*_H$. 

Recall from \prettyref{lmm:enum2} that there exists an integer $0 \le r \le \ell/2$ such that
\begin{align*}
\log | \Phi^*_{H_d} | & \le  \frac{1}{2} (k+r) \log (2n), \\
\log \left| \left\{ \phi_s:  \left(\phi_s, \phi_d \right) \in \Phi^*_{H} \right\} \right|
& \le \left( \ell/2-r \right) \log n +  \left( k+\ell/2 \right) \log 2.
\end{align*}
Invoking the large deviation bound \prettyref{lmm:LD} in \prettyref{app:LD} with $s=(k-\delta)/2$, $t=\ell/2-\delta$, $u=\delta$, and $v=r-\delta$, 
and noting that
\begin{align*}
s+u+v/2 & = (k+r)/2 \\
t-v & = \ell/2 - r \\
s+t+u-v/2  & = (k+\ell-r)/2 \ge k/2 + \ell/4,  
\end{align*}
and \prettyref{eq:alpha_condition},
we get 
\begin{equation}
 \prob{ \max_{\phi \in \Phi^\ast_{H}}  w ( \phi(H) )  \ge 0 }
 \le 5 \exp \left( - (\alpha-\log n) (k/8+\ell/16) \right). 
\end{equation}

It follows that
\begin{align*}
 \prob{\max_{F \in \calF^*_{k,\ell} } w(F) \ge 0 }
 & \overset{(a)}{\le} 5 | \calU_{k,\ell} | \exp \left( - (\alpha-\log n) (k/8+\ell/16) \right) \\
 & \overset{(b)}{\le} 5 \times 17^k 4^{\ell} \exp \left( - (\alpha-\log n) (k/8+\ell/16) \right) \\
 & \overset{(c)}{\le} 5 \exp \left( - (\alpha-\log n) (k/16+\ell/32) \right),
\end{align*}
where $(a)$ follows from union bound; 
 $(b)$ follows from \prettyref{lmm:enum1};
 $(c)$ holds because $\alpha-\log n \ge 16 \log 17$ by assumption \prettyref{eq:alpha_condition}. 
 Taking another union bound over $k \ge 0$ and $\ell \ge 4$, we get 
 \begin{align*}
 \prob{\max_{F \in \calF^* }  w(F) \ge 0  }
 =&~ \prob{\max_{k\geq 0,\ell\geq 4}  \max_{F \in \calF^*_{k,\ell}}  w(F)   \ge 0  } \\
 \leq & ~ \sum_{k \ge 0} \sum_{\ell \ge 4} \prob{\max_{F \in \calF^*_{k,\ell}} w(F) \ge 0}  \\
 \le  & ~ 5 \sum_{k \ge 0} \sum_{\ell \ge 4} \exp \left( - (\alpha-\log n) (k/16+\ell/32) \right)\\
 \le & ~  \frac{5}{1-e^{-(\alpha-\log n)/8} } \frac{e^{-(\alpha-\log n)/8} }{1-e^{-(\alpha-\log n)/32}} \\
 \le & ~ \frac{5}{1-1/17} \frac{e^{-(\alpha-\log n)/8}}{1-1/4 } \le 8 e^{-(\alpha-\log n)/8}.
 \end{align*}
 Therefore, we arrive at the desired \prettyref{eq:2f2_success_prob_2}, completing  
 the proof of \prettyref{thm:LP_opt}.

\section{Information-theoretic Necessary Conditions}
\label{sec:IT}

We first present a general necessary condition needed for \emph{any} algorithm to succeed in recovering the hidden Hamiltonian cycle 
with high probability.  Recall that $X$ and $Y$ are two independent random variables distributed as the log likelihood ratio $\log (\diff P/\diff Q)$ under $P$ and $Q$, respectively. 

\begin{theorem}[Information-theoretic conditions] \label{thm:converse}
If there exists a sequence of estimators $\hat{C}$ such that $\min_{C^*} \pprob{ \hat{C} = C^*} \to 1$ as $n \to \infty,$ 
then
\begin{equation}
 \sup_{\tau \in \reals} \left\{ \log \prob{X \le \tau} + \log \prob{ Y \ge \tau} \right\} + \log n  \le  O(1).
\label{eq:ITlimit_converse_full}
\end{equation}

\end{theorem}

Next, we state a regularity assumption on $P$ and $Q$ under which it immediately follows from \prettyref{thm:converse}
that $\alpha_n \ge (1+o(1)) \log n$ is necessary information-theoretically, thereby establishing the optimality of
F2F LP. 

\begin{assumption}\label{ass:chernoff_inverse}
$$
\sup_{\tau \in \reals} \left\{ \log \prob{X \le \tau} + \log \prob{ Y \ge \tau} \right\} 
\ge - \left( 1+o(1) \right) \alpha_n + o(\log n). 
$$
\end{assumption}

\begin{corollary} \label{corollary:converse}
Suppose \prettyref{ass:chernoff_inverse} holds. If there exists 
a sequence of estimators $\hat{C}$ such that $\min_{C^*} \pprob{ \hat{C} = C^*} \to 1$ as $n \to \infty,$ 
then
\begin{equation}
\alpha_n \geq \left(1+o(1) \right) \log n.
\label{eq:ITlimit_converse}
\end{equation}
\end{corollary}

\prettyref{ass:chernoff_inverse} is very general and fulfilled when the weight distributions are either Poisson, Gaussian or Bernoulli as the following result shows (see \prettyref{app:assumption_verify} for a proof):
\begin{lemma}\label{lmm:assumption_verify}
\prettyref{ass:chernoff_inverse}  holds in the Gaussian case with $\calP=\calN(\mu,1)$ and
$\calQ=\calN(0,1)$, the Poisson case with $\calP=\Pois(\lambda)$ and $\calQ=\Pois(\mu)$ for
for $\lambda \ge \mu$ such that 
$$
\log (\lambda \mu) = o  \left( (\sqrt{\lambda} -\sqrt{\mu})^2 \right)  + o(\log n), 
$$
and the Bernoulli case with $\calP=\Bern(p)$ and $\calQ=\Bern(q)$ for $p\ge q$. 
\end{lemma}

Let us now explain the intuition behind \prettyref{ass:chernoff_inverse}:
Denote the log moment generating function of $X$ and $Y$ as
\begin{equation}
\psi_P(\theta)=\log \eexpect{e^{\theta X}} , \quad \psi_Q(\theta)=\log \eexpect{e^{\theta Y} } = \psi_P(\theta-1).
\label{eq:logMGF}
\end{equation}
Denote the Legendre transform of $\psi_P$ and $\psi_Q$ as
 \begin{equation}
E_P(\tau) = \sup_{\theta \ge 0} \left\{  - \theta \tau - \psi_P( -\theta) \right \}, \quad
E_Q(\tau) = \sup_{\theta \ge 0} \left\{   \theta \tau - \psi_Q(\theta) \right \}.
\label{eq:legendre}
\end{equation}
Then Chernoff's inequality gives the following large deviation bounds: for any $\tau\in\reals$,
\begin{align}
 \prob{X \leq \tau} \leq \exp \left( - E_P(\tau) \right), \quad \prob{Y \geq \tau} \leq \exp \left( - E_Q(\tau) \right), \label{eq:chernoff_XY} 
\end{align}
Therefore,
$$
\sup_{\tau \in \reals} \left\{ \log \prob{X \le \tau} + \log \prob{ Y \ge \tau}  \right\}  \le -  \inf_{\tau \in \reals} \left\{ E_P(\tau) + E_Q(\tau)  \right\}.
$$
The infimum on the right-hand side is in fact equal to $\alpha_n$. Indeed,
\begin{align*}
\inf_{\tau\in\reals}E_P(\tau)+E_Q(\tau)
= & ~  \inf_{\tau\in\reals} \sup_{\theta_1,\theta_2\geq 0} \sth{- \theta_1 \tau - \psi_P(-\theta_1)  + \theta_2 \tau - \psi_Q(\theta_2) } \\
\geq & ~  \sup_{\theta_1,\theta_2\geq 0} \sth{\inf_{\tau\in\reals} (\theta_2- \theta_1) \tau - \psi_P(-\theta_1) - \psi_Q(\theta_2)}  \\
= & ~  \sup_{\theta \geq 0} \sth{- \psi_P(-\theta) - \psi_Q(\theta)} \\
=& ~  - \psi_P(-1/2) - \psi_Q(1/2) = -2 \log   \int \sqrt{\diff P \diff Q} =  \alpha_n,
\end{align*}
and the infimum over $\tau$ is in fact achieved by
$$
 \tau^*=\psi_P'\left(-1/2\right) =\psi'_Q\left(1/2 \right), 
$$ 
so that $ E_P\left(\tau^* \right) + E_Q\left(\tau^* \right)  = \alpha_n$. 
Hence, 
$$
\sup_{\tau \in \reals} \left\{ \log \prob{X \le \tau} + \log \prob{ Y \ge \tau} \right\}  \le -  E_P\left(\tau^* \right) - E_Q\left(\tau^* \right)  = - \alpha_n. 
$$
Therefore, the point of \prettyref{ass:chernoff_inverse} 
is to require that the large deviation exponents in Chernoff's inequalities \prettyref{eq:chernoff_XY} are asymptotically tight,
so that we can reverse the Chernoff bound in the lower bound proof.

\begin{proof}[Proof of \prettyref{thm:converse}]
To lower bound the worst-case probability of error, consider the Bayesian
setting where $C^*$ is drawn uniformly at random from all possible Hamiltonian cycles of 
$G.$ Since the prior distribution of $C^*$ 
is uniform, the ML estimator minimizes the error probability among all estimators.
Thus, without loss of generality, we can
assume the estimator $\hat{C}$ used is $\hat{C}_{\rm ML}$ and 
the true Hamiltonian cycle $C^*$ is
given by $(1,2,\ldots, n,1)$. Hence, by assumption $\pprob{\hat{C}_{\rm ML} = C^*} \to 1.$

Recall that the ML estimator is equivalent to finding a Hamiltonian cycle 
of the maximum weight. Given a Hamiltonian cycle $x$, 
define the simple graph $G_x$ with bicolored edge whose
  adjacency matrix is $|x-x^*|$ and each edge
  is colored in red if $(x-x^*)_e=-1$ and in blue if $(x-x^*)_e=+1$. 
  Also,  each edge $e$ has
  a weight $w_e(x-x^*)_e$ and hence $w(G_x)=\iprod{w}{x-x^*}.$
Note that if $G_x$ is a $4$-cycle of alternating colors given by $(i,i+1,j+1,j,i)$, then
 $x$ corresponds to a Hamiltonian cycle constructed 
 by deleting edges $(i,i+1),(j,j+1)$ in $C^*$ and adding edges $(i,j),(i+1,j+1)$ (see \prettyref{fig:c4} for an illustration).
Let $\calD$ denote the set of all possible $4$-cycles of alternating colors
given by $(i,i+1,j+1,j)$. Then $|\calD|=n(n-3)/2$, because for a given $i$, $j$ have
$(n-3)$ choices except $i-1,i,i+1.$

\begin{figure}[ht]
	\centering
	\begin{tikzpicture}[scale=1.5,every edge/.append style = {thick,line cap=round},font=\scriptsize]
\coordinate (i) at (112.5:1);
\coordinate (i1) at ({112.5-45}:1);
\coordinate (j) at (-112.5:1);
\coordinate (j1) at (-{112.5+45}:1);
	\foreach \a in {4,5,...,6}
	{
	\node[dot] at ({112.5+\a*360/8}:1) {};
	\path ({112.5+\a*360/8}:1) edge[red,postaction={on each segment={mid arrow=blue}}] ({112.5+(\a+1)*360/8}:1);
	}
	\foreach \a in {1,2,...,3}
	{
	\node[dot] at ({112.5+\a*360/8}:1) {};
	\path ({112.5+\a*360/8}:1) edge[red,postaction={on each segment={mid arrow=blue}}] ({112.5+(\a-1)*360/8}:1);
	}
\draw[red,dashed,thick] (i)--(i1);
\draw[red,dashed,thick] (j)--(j1);
\path node at (i) [above] {$i$} edge[blue,postaction={on each segment={early arrow=blue}}] (j1);
\path node at (i1) [above] {$i+1$} edge[blue,postaction={on each segment={early arrow=blue}}] (j);
\node at (j) [below] {$j+1$};
\node at (j1) [below] {$j$};
\end{tikzpicture}
~~~~~~~~ 
\raisebox{2em}{
	\begin{tikzpicture}[scale=1.5,every edge/.append style = {thick,line cap=round},font=\scriptsize]
\coordinate (i) at (0,1);
\coordinate (i1) at (1,1);
\coordinate (j) at (0,0);
\coordinate (j1) at (1,0);
\redsingleedge{(i)}{(i1)}
\redsingleedge{(j)}{(j1)}
\bluesingleedge{(i)}{(j)}
\bluesingleedge{(i1)}{(j1)}
\node at (i) [above] {$i$};
\node at (i1) [above] {$i+1$};
\node at (j) [below] {$j$};
\node at (j1) [below] {$j+1$};
\end{tikzpicture}
}
	\caption{The cycle $(1,2,\ldots,i,j,j-1,\ldots,i+1,j+1,j+2,\ldots,n)$ and the corresponding graph $G_x$ as a four-cycle.}
	\label{fig:c4}
\end{figure}
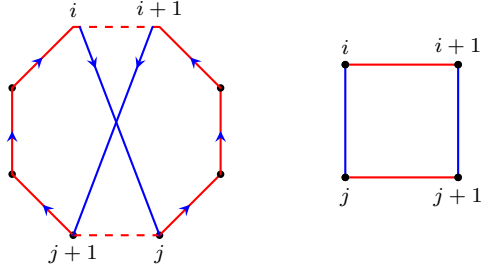

Define 
$$
S=\sum_{D \in \calD } \indc{w(D) \le 0}.
$$
If $S>0$, then there exists a Hamiltonian cycle $x \neq x^*$ whose weight
is at least as large as the weight of $C^*$; hence the likelihood function has
at least two maximizers, which in turn implies the probability of exact recovery by
ML estimator is at most $1/2.$
Therefore, $ \frac{1}{2} \prob{S>0} \le \prob{\text{ML fails} } =o(1)$. As
a consequence, $\prob{S=0} \to 1.$

To explain the intuition, suppose $w(D)$ are mutually independent for all $D \in \calD.$
Then
\begin{align*}
\prob{S=0} &= \prob{\forall D \in \calD, w(D) > 0} \\
&= \prod_{D \in \calD}  \prob{w(D)<0} \\
& \overset{(a)}{=} \left( 1 -  \prob{ Y_1 +Y_2 - X_1 - X_2 \ge 0} \right)^{|\calD|} \\
& \le \exp \left( - |\calD| \prob{ Y_1 +Y_2 - X_1 - X_2 \ge 0}  \right)
\end{align*}
where $(a)$ holds  because $w(D)$ has the same distribution 
as $Y_1+Y_2-X_1-X_2 $ and the last inequality holds in view of $1-x \le e^{-x}$.
In view of $\prob{S=0} \to 1,$ it follows from the last displayed equation that
$$
\log |\calD| + \log \prob{ Y_1 +Y_2 - X_1 - X_2 \ge 0} \to - \infty. 
$$
Furthermore,  for any $\tau \in \reals$, we have
\begin{align*}
\log \prob{ Y_1 +Y_2 - X_1 - X_2  \ge 0} 
& \ge \log \left( \prob{Y_1 \ge \tau } \prob{Y_2 \ge \tau}  \prob{X_1 \le \tau}
\prob{X_2 \le \tau } \right) \\
& =2 \log \prob{Y \ge \tau} + 2 \log \prob{X \le \tau}. 
\end{align*}
Combining the last two displayed equation and recalling 
that $|\calD|=n(n-3)/2$, we immediately get that 
$$
\log \prob{Y \ge \tau} + \log \prob{X \le \tau} + \log n  \to -\infty,
$$
Taking the supremum over $\tau \in \reals$ of the last displayed equation 
yields the desired \prettyref{eq:ITlimit_converse_full}.


However, $w(D)$ and $w(D')$ are dependent if $D$ and $D'$ share edges. To deal with this 
dependency, we focus on a subset of $\calD$.  
In particular, for any $\tau \in \reals$,  define 
$$
I=\left\{\text{odd } i: w_{i,i+1} \leq \tau \right\}
$$
and 
$$
J= \left\{ (i,j) \in I \times I: i \neq j, w_{i,j} + w_{i+1,j+1} \geq 2\tau  \right\}.
$$
Then for any $(i,j) \in J$, the alternating $4$-cycle given by $(i,i+1,j+1,j,i)$ belongs to 
$\calD$ and has a non-positive weight. Hence, $|J| \le S $ and thus 
$\prob{|J| =0} \ge \prob{S=0} \to 1.$

Note that $w_{i,i+1}$ has the same distribution as $X$. Thus for any $\tau \in \reals$, 
$$
\prob{ w_{i,i+1 } \leq \tau} = \prob{X \leq \tau} \triangleq p.
$$
Also, $w_{i,i+1}$ are mutually independent for different $i.$ 
Thus $| I | \sim \Binom ( \lceil n/2 \rceil, p) $. By the Chernoff's
bound for binomial distribution, 
$$
\prob{ | I | \le np/4 }  \le \exp \left( - np /8 \right).
$$
Thus, 
$$
\prob{|J| =0} \le \prob{   |J|=0, |I| > np/4 } + \prob{|I| \le np/4 }
\le \prob{   |J|=0 \mid |I| > np/4 }  + \exp \left( - np /8 \right).
$$
Let
$
q \triangleq \prob{Y \geq \tau}. 
$
Then $\prob{Y_1+Y_2 \geq 2\tau} \ge \prob{Y_1 \geq \tau} \prob{Y_2 \geq \tau} =q^2$
and hence
\begin{align*}
\prob{  |J|=0 \mid |I| > np/4  }  
& =\prob{ \forall i< j \in I, w_{(i,j)} + w_{(i+1,j+1)} < 2\tau  \mid |I| > np/4 } \\
& \overset{(a)}{\le} (1-q^2)^{\binom{np/4}{2}} \le e^{  -q^2 \binom{np/4}{2} },
\end{align*}
where $(a)$ holds because conditional on $I$, $w_{(i,j)} + w_{(i+1,j+1)}$ 
are i.i.d.~copies of $Y_1+Y_2.$
Combining the last three displayed equations yields that 
\begin{align}
\prob{|J| =0 } \le e^{  - q^2 \binom{np/4}{2} } +  e^{ - np /8 }. 
\end{align}
Recall that $\prob{|J|=0} \to 1$. It follows that
\begin{align}
e^{  - q^2 \binom{np/4}{2} } +  e^{ - np /8 } \ge 1+o(1). \label{eq:converse_key_equation}
\end{align}
Hence, 
$
\log n +  \log p +  \log q \le  O(1), 
$
or equivalently,
$$
\log \prob{X \le \tau} + \log \prob{Y \ge \tau} + \log n  \le O(1).
$$
Taking the supremum over $\tau \in \reals$ of the last displayed equation 
yields the desired \prettyref{eq:ITlimit_converse_full}. 


\end{proof}

\section{Reduction between Hamiltonian Cycle and Hamiltonian Path}
\label{sec:cycle_path}

In this section, we extend our results to the hidden
Hamiltonian path model. As opposed to \prettyref{def:model}, $G$ is a randomly weighted, 
undirected complete graph  with a  hidden Hamiltonian path $H^*$ 
such that every edge has an independent weight drawn from $\calP$ 
if it is on $H^*$ and from $\calQ$ otherwise.

The following theorem shows that the information-theoretic limits 
of hidden Hamiltonian cycle and path models coincide. 

\begin{theorem}\label{thm:optimal_threshold_path}
Consider the ML estimator $\hat{H}_{\rm ML} \in \arg \min_{H \in \calH} w(H)$, 
where $\calH$ denotes the set of all possible Hamiltonian paths in graph $G$ 
and $w(H)$ denote the total weights on the Hamiltonian path $H$. If  
$$
\alpha_n - \log n \to + \infty, 
$$
\ie, \prettyref{eq:ITlimit} holds, 
then $\min_{H^* \in \calH} \pprob{ \hat{H}_{\rm ML} = H^*} \to 1$ as $n\to \infty$.

Conversely, suppose \prettyref{ass:chernoff_inverse} holds. If there exists 
a sequence of estimators $\hat{H}$ such that $\min_{H^* \in \calH} \pprob{ \hat{H} = H^*} \to 1$ as $n \to \infty,$ 
then $$\alpha_n  \ge \left(1+o(1) \right) \log n,$$ \ie, \prettyref{eq:ITlimit_converse} holds. 
\end{theorem}

\begin{proof}
We first prove the sufficiency part. 
We abbreviate the ML estimator $\hat{H}_{\rm ML}$ as $\hat{H}$.
Let $L=| H^* \setminus \hat{H} |$.  Since $|H^*|=|\hat{H}|=n-1$, it follows that $ 0 \le L \le n-1$
and $|H^\ast \Delta \hat{H} | = 2L$. To prove the theorem, it suffices to show that 
$\prob{L \ge 1} =o(1)$. 
For any $ 1 \le \ell \le n-1$, 
$$
\{ L=\ell \}  \subset \{ \exists \text{ Hamiltonian path $H$}:  | H^* \setminus H| = \ell, w(H^*) < w(H)  \}.
$$
Note that for each fixed $H \in \calH_n$ with $ | H^* \setminus H| = \ell$, 
$$
w(H) -w(H^*) \eqdistr \sum_{i=1}^{\ell}  Y_i - \sum_{i=1}^{\ell} X_i,
$$ 
where $\{X_i\}$ and $\{Y_i\}$ are independent and iid copies of the log likelihood ratio $\log (\diff P/\diff Q)$ under $P$ and $Q$, respectively. 
Also, there are at most $\binom{n-1}{\ell} 2^{\ell+1}  (\ell+1)! $ different choices of 
Hamiltonian paths $H$ with $ | H^* \setminus H| = \ell$. To see this, note that there
are at most $\binom{n-1}{\ell}$ different choices of edges in 
$H^* \setminus H$. For each such choice, 
there are at most $\ell +1$ disjoint paths in $H^* \cap H$;
thus there are at most $2^{\ell+1}$ different orientations of these paths,
and at most $(\ell+1)!$ different permutations of these paths. 

Applying union bound gives that
\begin{align*}
\prob{L=\ell} &\le \binom{n-1}{\ell} 2^{\ell+1}  (\ell+1)!  \; \times \; \prob{  \sum_{i=1}^{\ell}  Y_i - \sum_{i=1}^{\ell} X_i < 0}  \\
& \le  n^{\ell} 2^{\ell+1} (\ell+1) e^{ - \ell \alpha } \le  e^{ -\omega(\ell) } ,
\end{align*}
where the second inequality follows due to  $\binom{n-1}{\ell} (\ell+1)! \le n^{\ell} (\ell+1)$ and \prettyref{lmm:LD}; 
the last inequality follows from the assumption \prettyref{eq:ITlimit}.  Therefore, 
$$
\prob{ L \ge 1} =\sum_{\ell=1}^{n-1} \prob{L=\ell} \le \sum_{\ell=1}^{n-1} e^{-\omega(\ell) } = e^{-\omega(1)},
$$
completing the proof for sufficiency part.

Next we show the necessity part. To lower bound the worst-case probability of error,
consider the Bayesian setting where $H^*$ is uniformly chosen among all possible Hamiltonian
paths of $G.$ Since the ML estimator minimizes the error probability among all estimators, 
without loss of generality, we can
assume the estimator $\hat{H}$ used is $\hat{H}_{\rm ML}$ and 
the true Hamiltonian path $H^*$ is
given by $(1,2,\ldots, n)$. Hence, by assumption $\prob{\hat{H}_{\rm ML} = H^*} \to 1.$
Note that under the hidden Hamiltonian path model, the ML estimator is equivalent to 
finding the maximum  weighted Hamiltonian path.    
For $i \in \{2, \ldots, n-1\}$, define the event that 
$$
\calE_i =\{ w_{1, i+1} \ge w_{i, i+1} \}.
$$
Let $\calF=\cup_{i=2}^{n -1} \calE_i$, which implies the existence of vertex $i$ such that
the total weight of the Hamiltonian path $H^* \setminus \{ (i,i+1) \} \cup \{ (1, i+1)  \}$ is 
at least as large as the total weight of $H^*$. Hence, on $\calF,$ the likelihood function has
at least two maximizers, which in turn implies the probability of exact recovery by
ML estimator is at most $1/2.$ Hence $\frac{1}{2} \prob{\calF} \leq \prob{\text{ML fails}}=o(1)$.
As a consequence, $\prob{\calF^c} \to 1.$ Note that
\begin{align*}
\prob{\calF^c} & \overset{(a)}{=} \prod_{i=2}^{n -1} \prob{ \calE_i^c } \\
& \overset{(b)}= \left(1 - \prob{ X \le Y  } \right)^{ n -2 } \\
&  \overset{(c)}{\le}   \exp \left( - (n-2) \prob{ X \le Y }  \right)
\end{align*}
where $(a)$ holds due the independence among $\calE_i$;
$(b)$ holds because $\prob{\calE_i } =  \prob{ X < Y }$;
$(c)$ follows due to $1+x \le e^x$ for all $x \in \reals$. 
It follows that
$$
\log n + \log \prob{X \le Y} = o(1). 
$$
Note that 
\begin{align*}
\log \prob{X \le Y} & \ge \sup_{\tau \in \reals} \left\{ \log \prob{X \le \tau}  + \log \prob{Y\ge\tau} \right\} \\
& \ge -\left(1+o(1) \right) \alpha + o(\log n),
\end{align*}
where the last inequality is exactly \prettyref{ass:chernoff_inverse}. Combining the last  two displayed
equations yields that $\alpha  \ge (1+o(1)) \log n$. 
\end{proof}


Next, we show that given a polynomial-time estimator which achieves exact recovery
in one model, one can construct another polynomial-time estimator which achieves 
exact recovery for the other model whenever it is information-theoretically possible. 

\prettyref{alg:cycle_to_path} gives a procedure to construct a polynomial-time estimator of 
Hamiltonian path given a polynomial-time estimator of Hamiltonian cycle.

\begin{algorithm}
\caption{From Hamiltonian cycle estimator to Hamiltonian path estimator}\label{alg:cycle_to_path}
\begin{algorithmic}[1]
\STATE {\bfseries Input:} Graph $G$, an estimator $\hat{C}(G)$ of Hamiltonian cycle in $G.$
\STATE {\bfseries Output:} An estimator $\hat{H}(G)$ of Hamiltonian path in $G.$
\STATE Generate a random variable $W \sim \calP.$ 
\STATE For every edge $e=(i,j)$ in $G$, 
construct a weighted graph $G^{e}$ by replacing the 
edge weight on $e$ in $G$ by $W$. If $\hat{C} \left( G^{e} \right)$ 
is a Hamiltonian cycle containing $e$, let 
$\hat{H}^{e}$ denote the Hamiltonian path by deleting edge $e$ from $\hat{C} \left( G^{e} \right) $;
otherwise, let $\hat{H}^{e}=\emptyset$.
\STATE  Output
$$
\hat{H} \in \arg \min_{e: \hat{H}^{e}\neq \emptyset} w\left( \hat{H}^{e} \right),
$$
where $w\left( \hat{H}^{e}\right)$ denotes the sum of log likelihood weights on edges in $\hat{H}^{e}$.
\end{algorithmic}
\end{algorithm}

Similarly, \prettyref{alg:path_to_cycle} gives a procedure to construct a polynomial-time estimator of 
Hamiltonian cycle given a polynomial-time estimator of Hamiltonian path.

\begin{algorithm}
\caption{From Hamiltonian path estimator to Hamiltonian cycle estimator}\label{alg:path_to_cycle}
\begin{algorithmic}[1]
\STATE {\bfseries Input:} Graph $G$, an estimator $\hat{H}(G)$ of Hamiltonian path in $G.$
\STATE {\bfseries Output:} An estimator $\hat{C}(G)$ of Hamiltonian cycle in $G.$
\STATE Generate a random variable $W \sim \calQ.$ 
\STATE For every edge $e=(i,j)$ in $G$, 
construct a weighted graph $G^{e}$ by replacing the 
edge weight on $e$ in $G$ by $W$. If $\hat{H} \left( G^{e} \right)$ 
is a Hamiltonian path connecting vertices $i,j$, let 
$\hat{C}^{e}$ denote the Hamiltonian cycle by adding edge $e$ in $\hat{H} \left( G^{e} \right)$;
otherwise, let $\hat{C}^{e}=\emptyset$.
\STATE  Output
$$
\hat{C} \in \arg \min_{e: \hat{C}^{e}\neq \emptyset} w\left( \hat{C}^{e} \right),
$$
where $w\left( \hat{C}^{e}\right)$ denotes the sum of log likelihood weights on edges in $\hat{C}^{e}$.
\end{algorithmic}
\end{algorithm}



\begin{theorem}\label{thm:path_cycle_conversion}
Given an estimator $\hat{C}(G)$ with time complexity $O(T_n)$ such that $\pprob{\hat{C} = C^*} \to 1$
as $n \to \infty $ under the hidden Hamiltonian cycle model. Let  $\hat{H}(G)$ denote the output of
\prettyref{alg:cycle_to_path} with input $\hat{C}(G).$ Then $\hat{H}(G)$ has time complexity $O(n^2 T_n)$.
Moreover, under the hidden Hamiltonian path model if $\pprob{\hat{H}_{\rm ML} = H^*} \to 1$ as $n \to \infty$, 
then $\pprob{\hat{H} = H^*} \to 1$.  

Conversely, given an estimator $\hat{H}(G)$ with time complexity $O(T_n)$ such that $\pprob{\hat{H} = H^*} \to 1$
as $n \to \infty $ under the hidden Hamiltonian path model. Let $\hat{C}(G)$ 
denote the output of
\prettyref{alg:path_to_cycle} with input $\hat{H}(G).$ Then $\hat{C}(G)$ has time complexity $O(n^2 T_n)$.
Moreover,  under the hidden Hamiltonian cycle model if $\pprob{\hat{C}_{\rm ML} = C^*} \to 1$ as $n \to \infty$, 
then  $\pprob{\hat{C} = C^*} \to 1$.  
\end{theorem}
\begin{proof}
By construction in \prettyref{alg:cycle_to_path}, the time complexity of $\hat{H}(G)$ is $O(n^2T_n).$
We show that $\pprob{\hat{H} =H^*} \to 1$. 
Without loss of generality, assume that $H^*=(1,2,\ldots, n).$
Then according to \prettyref{alg:cycle_to_path}, $G^{e}$ with $e=(1,n)$
is a realization of the hidden Hamiltonian cycle model in \prettyref{def:model}
with ground truth $C^*=(1,2,\ldots, n, 1).$
By assumption, $\pprob{ \hat{C} \left(G^{e} \right) = C^* } \to 1.$
Hence, $\pprob{\hat{H}^{e} = H^*} \to 1.$ Define two events
\begin{align*}
E_1 =\left\{ \hat{H}^{e} = H^* \right\}  , \quad E_2  =\left\{ \hat{H}_{\rm ML} = H^*  \right\}.
\end{align*}
By assumption $\pprob{E_2} \to 1.$ Hence, 
$\pprob{E_1 \cap E_2} \to 1.$
On the event $E_2$, $H^*$ is the unique Hamiltonian path with the maximal weight.
Thus, on the event $E_1 \cap E_2,$
$$
\hat{H}^{e} = H^* = \arg \min_{H \in \calH} w(H),
$$
where $\calH$ denotes the set of all possible Hamiltonian paths. 
Since $\hat{H}^{e} \in \calH$ for all $\hat{H}_e \neq \emptyset$, it follows that on the event $E_1 \cap E_2,$
$
\hat{H} = \hat{H}^{e} = H^*.
$
Thus, $\pprob{\hat{H} = H^*} \to 1$.

Conversely, by construction in \prettyref{alg:path_to_cycle}, the time complexity of $\hat{C}(G)$ is $O(n^2T_n).$
The proof for $\pprob{\hat{C} =C^*} \to 1$ is completely analogous as before and omitted. 
\end{proof}

As an immediate corollary of \prettyref{thm:path_cycle_conversion} and \prettyref{thm:LP_opt}, 
we show the hidden Hamiltonian path can be exactly recovered with high probability in polynomial-time,
provided that $\alpha_n-\log n\to +\infty.$
\begin{corollary} 
If the optimal solution $x^*$ of the F2F LP \prettyref{eq:F2F} is an adjacency vector 
of a Hamiltonian cycle in $G$, let $\hat{C}(G)$ denote this Hamiltonian cycle; otherwise,
let $\hat{C}(G)$ denote any arbitrary Hamiltonian cycle in $G.$
Let $\hat{H}(G)$ denote the output of
\prettyref{alg:cycle_to_path} with input $\hat{C}(G).$ 
If $\alpha_n-\log n\to +\infty$, \ie, \prettyref{eq:ITlimit} holds, 
then $\pprob{\hat{H} = H^*} \to 1$ as $n \to \infty.$
\end{corollary}

\section{Empirical Results}
\label{sec:emp_results}
We empirically evaluate the performance of simple thresholding, greedy merging, belief propagation (BP) and F2F (LP) on both simulated and real datasets from DNA assembly.\footnote{Since nearest-neighbor algorithm is very similar to simple thresholding
and has the same performance guarantee as the
simple thresholding, we do not include nearest-neighbor in the empirical study.} Of the four, the first three algorithms output integer-valued solutions. 
However, the optimal solution of F2F may contain half-integral entries and we randomly round them to either $0$ or $1$ with equal probability.
We run BP for $10000$ iterations for all the experiments in this section. All the plots in this section can be reproduced from the following repository - \url{https://github.com/bagavi/HiC_LP}.

\subsection{Simulation}
\label{subsec:sim}
In this subsection, we consider a completely simulated setting of hidden Hamiltonian cycle graph consisting of $n=1000$ vertices,  whose edge weights are independently drawn from $\calP=\calN(\mu,1)$ if the edge on the Hamiltonian cycle $C^*$, and  $\calQ=\calN(0,1)$ otherwise. 

\prettyref{fig:simulations_a} plots the error probability $\prob{\hat{C} \neq C^*}$ 
of the output $\hat{C}$ for each of the four algorithms. As predicted by our theory
(cf. \prettyref{eq:alphaexample} in the Gaussian case), the error probabilities of F2F and BP exhibit a phase transition around $\mu^2=4\log n$. Moreover, the error probabilities of greedy merging and simple thresholding drop off to $f0$ around $\mu^2=6\log n$ 
and $\mu^2=8\log n,$ respectively, which is consistent with our theoretical performance analysis (cf.\ Table \prettyref{table:various_alg}).


\begin{figure}[h]
\centering
\includegraphics[width=0.9\columnwidth]{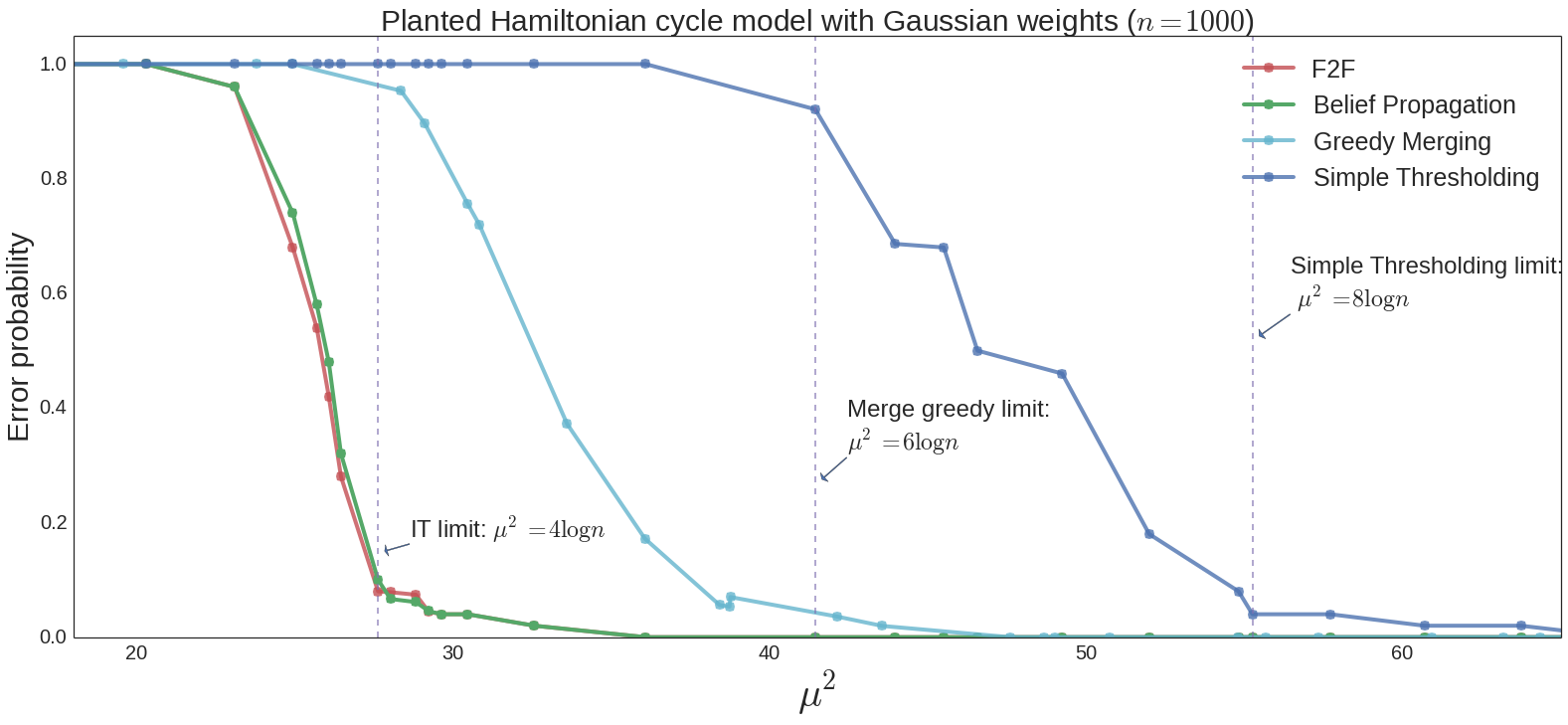}
\caption{Four algorithms are evaluated under simulated setting of hidden Hamiltonian cycle with $n=1000$ vertices, with  $\calP=\calN(\mu,1)$ and  $\calQ=\calN(0, 1)$,
for $\mu^2$ between $20$ and $65$. The y-axis is the probability of error, \ie, the probability that the hidden Hamiltonian cycle $C^*$ is not exactly recovered. For each value of $\mu^2$, the error probability is calibrated using $50$ independent runs.
 BP algorithm is run for $10000$ iterations. 
 }
\label{fig:simulations_a}
\end{figure}

\prettyref{fig:simulations_b} plots the fraction of misclassified edges $|\hat{C} \Delta C^*|/n$, of the output $\hat{C}$ for each of the four algorithms,
where $\Delta$ denotes the set difference. 
In other words, the fraction of misclassified edges is equal 
to the number of edges in $C^*$  that are absent in the $\hat{C}$ plus the number of edges in the $\hat{C}$ that are absent in $C^*$,
normalized by $n$.  From both Figures \ref{fig:simulations_a} and \ref{fig:simulations_b}, we observe that 
when the optimal solution of F2F is integral, the solution of BP coincides with the F2F solution, as predicted by
the result in~\cite{bayati2011belief}. However, when the optimal solution of F2F has half-integral entries, the solution of
BP often does not coincide with the F2F solution, and in fact has larger errors.  
The same phenomenon is also observed with real datasets as shown later.
\begin{figure}[ht]
\centering
\includegraphics[width=0.9\columnwidth]{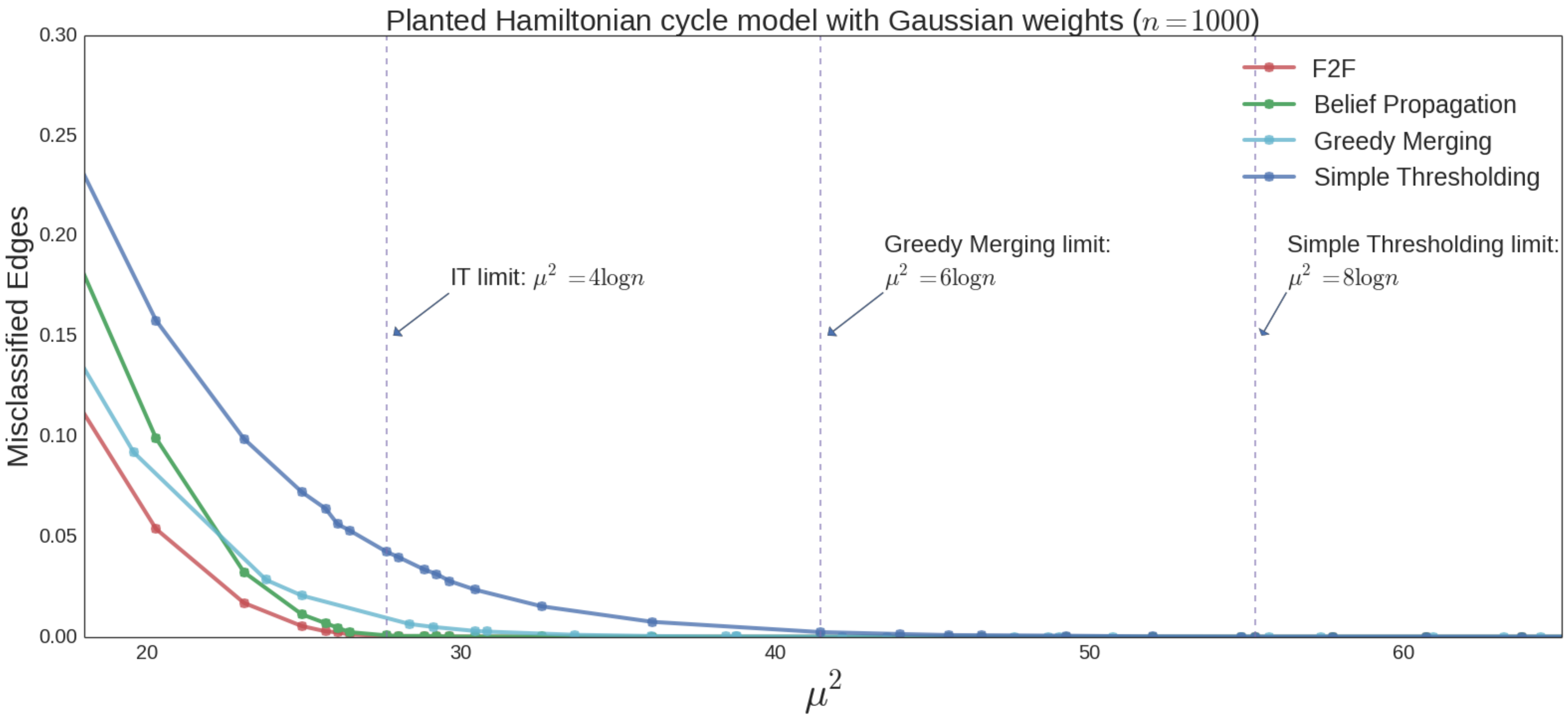}
\caption{This figure uses the same experiments in \prettyref{fig:simulations_a} to plot a different measure of error -- fraction of misclassified edges $|\hat{C} \Delta C^*|/n$. 
For each value of $\mu^2$, the error is averaged  over $50$ independent trials. BP algorithm is run for $10000$ iterations.
}
\label{fig:simulations_b}
\end{figure}

\subsection{DNA Assembly}
As described in Section \ref{sec:introduction}, ordering contigs using  HiC reads (or Chicago reads, a variant of HiC reads)~\cite{putnam2016chromosome, dudchenko2017novo} 
can be modeled as recovering hidden Hamiltonian cycle where the contigs are vertices of the graph, the hidden Hamiltonian cycle is the true ordering of the contigs on the genome, and the weights on the graph are the counts of the Hi-C reads (or Chicago) linking the contigs. 
In this section, we empirically compare the performance
of simple thresholding, greedy merging, BP and F2F 
on datasets from the following three organisms:
\begin{enumerate}
\item Homosapiens - Sequenced from a human in Iowa, USA. \cite{putnam2016chromosome}.
\item Aedes Aegypti \cite{dudchenko2017novo} - A breed of mosquito which spreads dengue, zika fevers, and found throughout the world.
\item Gasterosteus Aculeatus ~\cite{peichel2016improvement} - A variety of fish which is extensively studied to understand evolution and behavioral ecology.
\end{enumerate}
We choose these three datasets because the DNAs of these organisms have been sequenced and
assembled using other (expensive) sequencing technologies, and that gives us
the ground truth-- the true contig ordering. 
The information on data sources and the detailed procedure is provided in Appendix \ref{appendix:data}.
We evaluate the performance of all the four algorithms under two scenarios -- \textit{equal}-length contigs and \textit{variable}-length contigs. 

\subsubsection{DNA Assembly: Equal-length contigs}
For real datasets, shotgun reads are assembled to produce unordered contigs of variable lengths. However, in this subsection we consider a controlled experimental setting. We partition a chromosome into contigs of equal length  $100$ kbps and randomly shuffle these contigs. 
We choose this length because the span of HiC reads is typically between $0$ and $100$ kbps. As a consequence, 
neighboring pairs of contigs tend to have many HiC reads between them, while non-neighboring pairs
tend to have relatively fewer HiC reads, as postulated under our hidden Hamiltonian cycle model. 
Hence, we expect our algorithms -- developed for recovering the hidden Hamiltonian cycle -- to perform well
for finding the true ordering of contigs under this setting.


\paragraph*{Graph Construction} Firstly, the chromosome is partitioned into \textit{equal}-length contigs
of length $100$ kbps, and only the first $1000$ contigs are kept and shuffled randomly.
Secondly, the HiC reads are aligned to these $1000$ contigs. Lastly, we construct graph $G=([n],E)$,
where every vertex $i$ corresponds to a contig and
the weight on every edge $e=(i,j) \in E$  is
the number of HiC reads between the two contigs corresponding
to vertices $i,j$ . For example, the weighted adjacency matrix of $G$ for chromosome 1 of Homosapiens is
shown in \prettyref{fig:contact}.


\paragraph*{Performance evaluation}

The performance of all four algorithms is shown
in Figures \ref{fig:homosapiens}, \ref{fig:gasterosteus}, \ref{fig:aedes} for data from Homosapiens,  Gasterous Aculeatus, and Aedes Agyepti, respectively. We observe that across all chromosomes and organisms,
for a fixed coverage depth of HiC reads, F2F consistently has fewer errors
than the two greedy algorithms -- simple thresholding and greedy merging.
Moreover, BP also has fewer errors than the two greedy algorithms in most of the experiments. 
However, as mentioned previously, when the optimal solution of F2F has fractional entries, 
the BP solution no longer matches the optimal solution of F2F and has more errors.

\begin{figure}
\centering
\includegraphics[width=0.9\columnwidth]{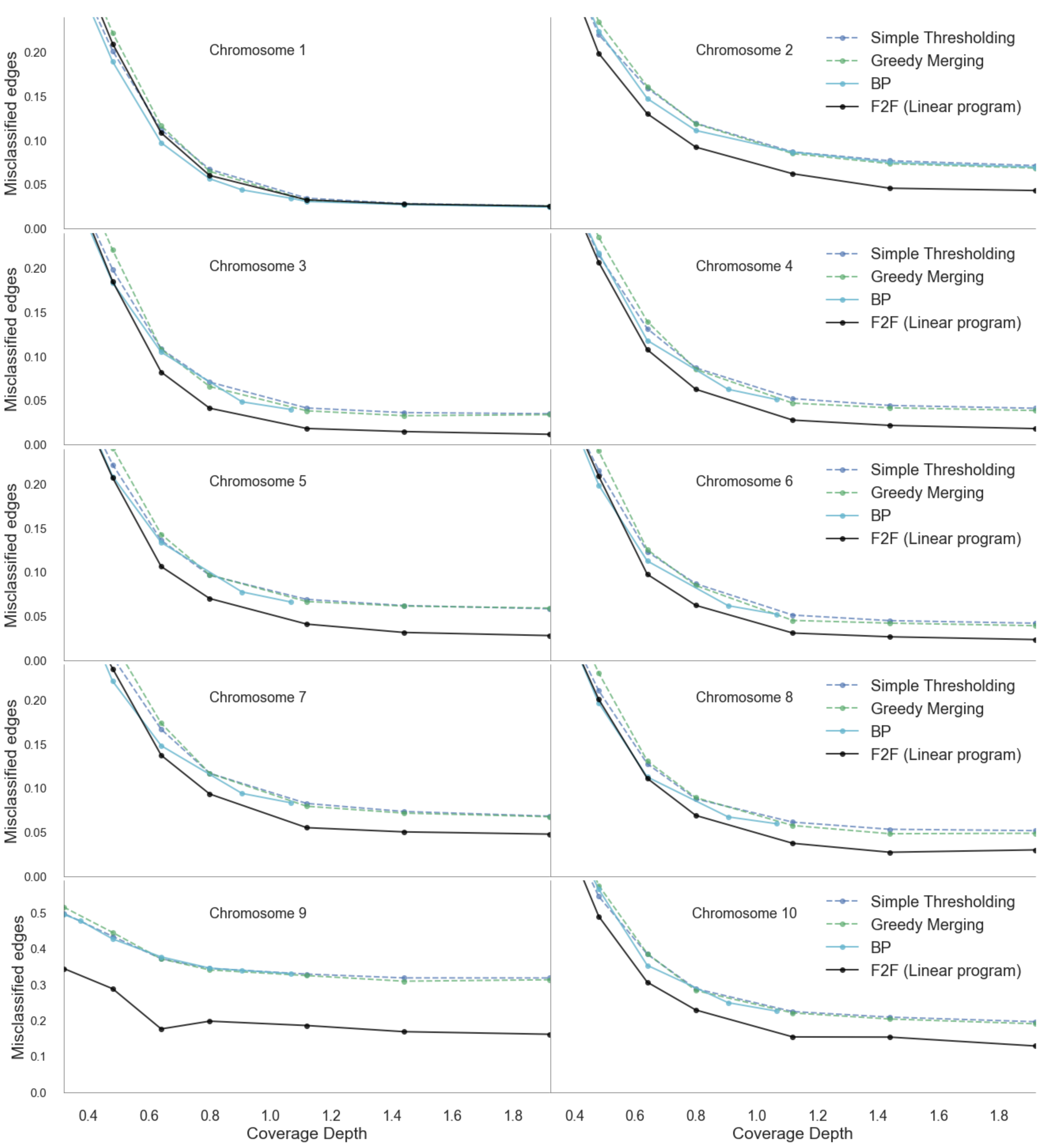}
\caption{Homosapiens [Equal length contigs]: The $x$-axis is the coverage depth of Hi-C reads i.e, expected number of Hi-C reads that include a base pair (nucleotide).
The $y$-axis is the fraction of misclassified edges.
For each value of coverage depth, the corresponding misclassification error is averaged over $25$ independent runs of experiment.
\label{fig:homosapiens}}
\end{figure}

\begin{figure}
\centering
\includegraphics[width=0.9\columnwidth]{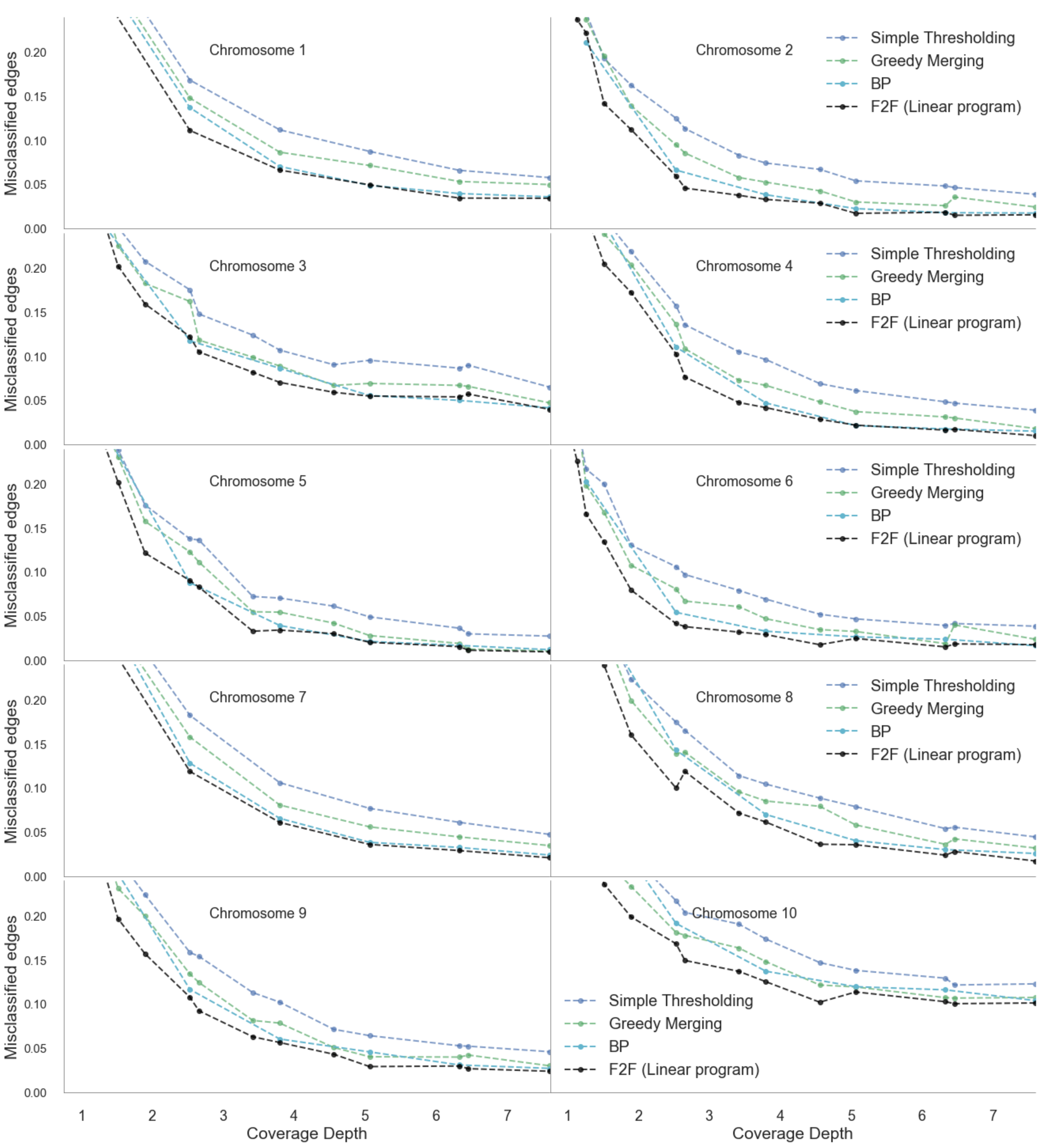}
\caption{Gasterosteus Aculeatus [Equal length contigs]: The $x$-axis is the coverage depth of Hi-C reads and the $y$-axis is the fraction of misclassified edges.
For each value of coverage depth, the corresponding misclassification error is averaged over $25$ independent runs of experiment.
\label{fig:gasterosteus}
}
\end{figure}

\begin{figure}
\centering
\includegraphics[width=0.9\columnwidth]{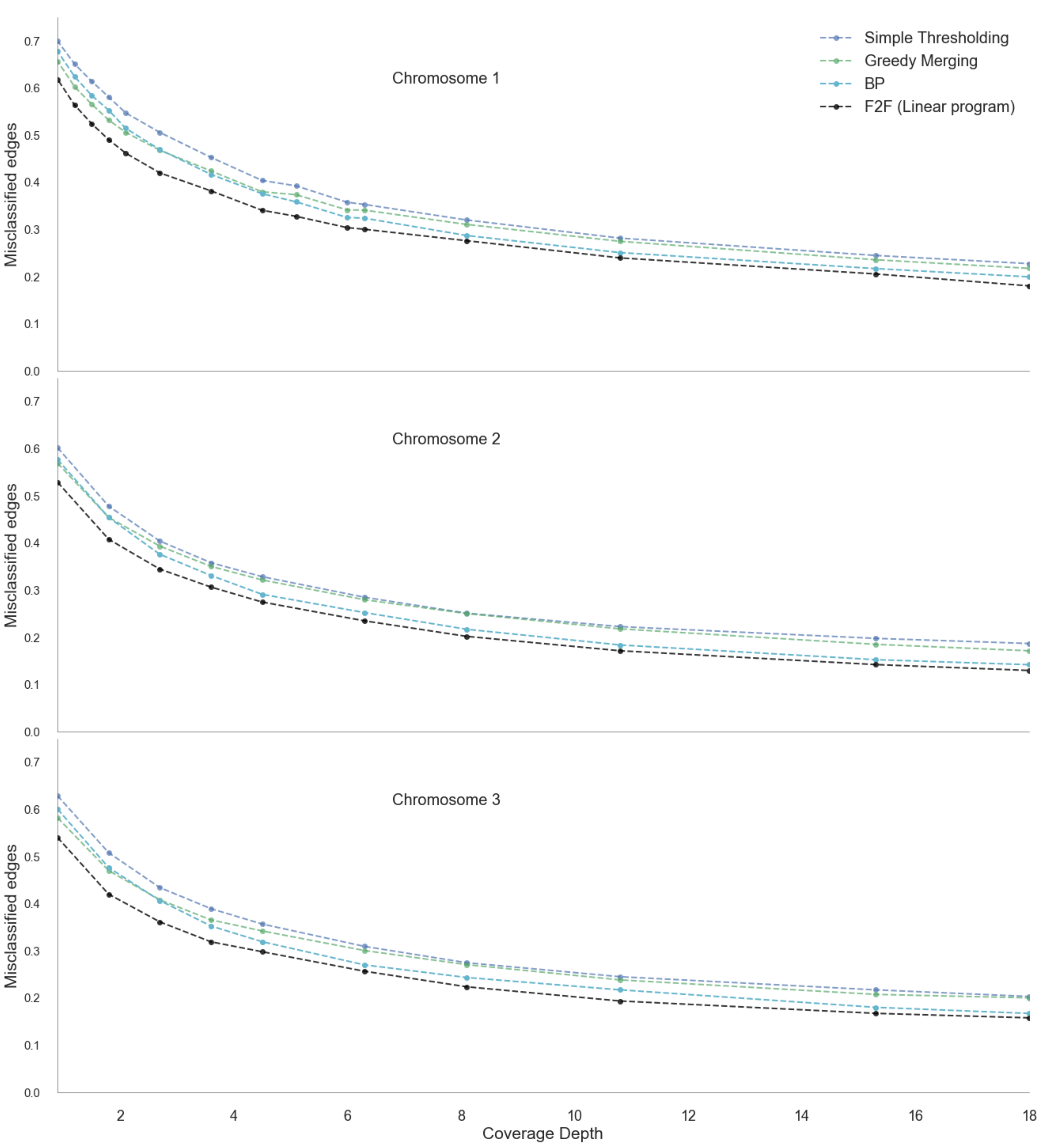}
\caption{Aedes Agyepti [Equal length contigs]: The $x$-axis is the coverage depth of Hi-C reads and the $y$-axis is the fraction of misclassified edges.
For each value of coverage depth, the corresponding misclassification error is averaged over $25$ independent runs of experiment.
\label{fig:aedes}
}
\end{figure} 

\subsubsection{DNA Assembly: Variable-length contigs}
We turn to variable-length contigs generated by shotgun assembly softwares~\cite{chapman2011meraculous, butler2008allpaths}. 
In this setting, HiC reads bias towards longer contigs. In other words, there tends to be more HiC reads linking
two long contigs, which is not captured by our hidden Hamiltonian cycle model. 
Hence, our algorithms developed for recovering the hidden Hamiltonian cycle may not perform well
for finding the true ordering of contigs. Nevertheless, we find that F2F and BP still consistently have few errors than the greedy methods, which are an integral part of the state-of-the art HiC assemblers~\cite{putnam2016chromosome, dudchenko2017novo, burton2013chromosome, ghurye2017scaffolding}. Thus we expect these assemblers can be possibly improved by replacing 
their greedy algorithm component by F2F or BP algorithm.

\paragraph*{Graph Construction}
Firstly, the unordered contigs are obtained as output of short-gun reads assemblers 
such as Meraculous ~\cite{chapman2011meraculous} and AllPaths~\cite{butler2008allpaths}. 
These obtained contigs have variable lengths ranging from $10$ kbps to $10$ mbps (million base pairs). 
Secondly, the HiC reads are aligned to these contigs to obtain the raw number of HiC reads between each pair of contigs. 
Lastly, we construct graph $G =([n], E)$, where every vertex $i$ correspond to a contig. 
The weight on every edge $e=(i,j) \in E$ is equal to $w_{ij} = b_i b_j N_{i,j}$
with $N_{i,j}$ denoting the \textit{number} of HiC reads between the two contigs corresponding to vertices $i,j$.
The values of $b_i$'s are chosen such that $\sum_{i} w_{i,j} = 1$ for all $j$, in order to remove the bias due to
the variable contig length. These values of $b_i$'s are calculated using an iterative algorithm, \href{https://github.com/hiclib/iced}{ICE}, described in `Online methods' section from ~\cite{imakaev2012iterative}.


\paragraph*{Performance evaluation}
The performance of all four algorithms is shown
in Figures \ref{fig:homosapiens_var}, \ref{fig:gasterosteus_var}, \ref{fig:aedes_var} for
the data from Homosapiens,  Gasterous Aculeatus, and Aedes Agyepti, respectively.  
Similar to the results in the setting with equal-length contigs, across all chromosomes and organisms,
for a fixed coverage depth of HiC reads, F2F consistently has fewer errors
than the two greedy algorithms. Moreover, BP also has fewer errors than greedy methods in most of the experiments. 
When the optimal solution of F2F has fractional entries, the BP solution does not match the optimal solution of F2F 
and has more errors.

\begin{figure}
\centering
\includegraphics[width=0.9\columnwidth]{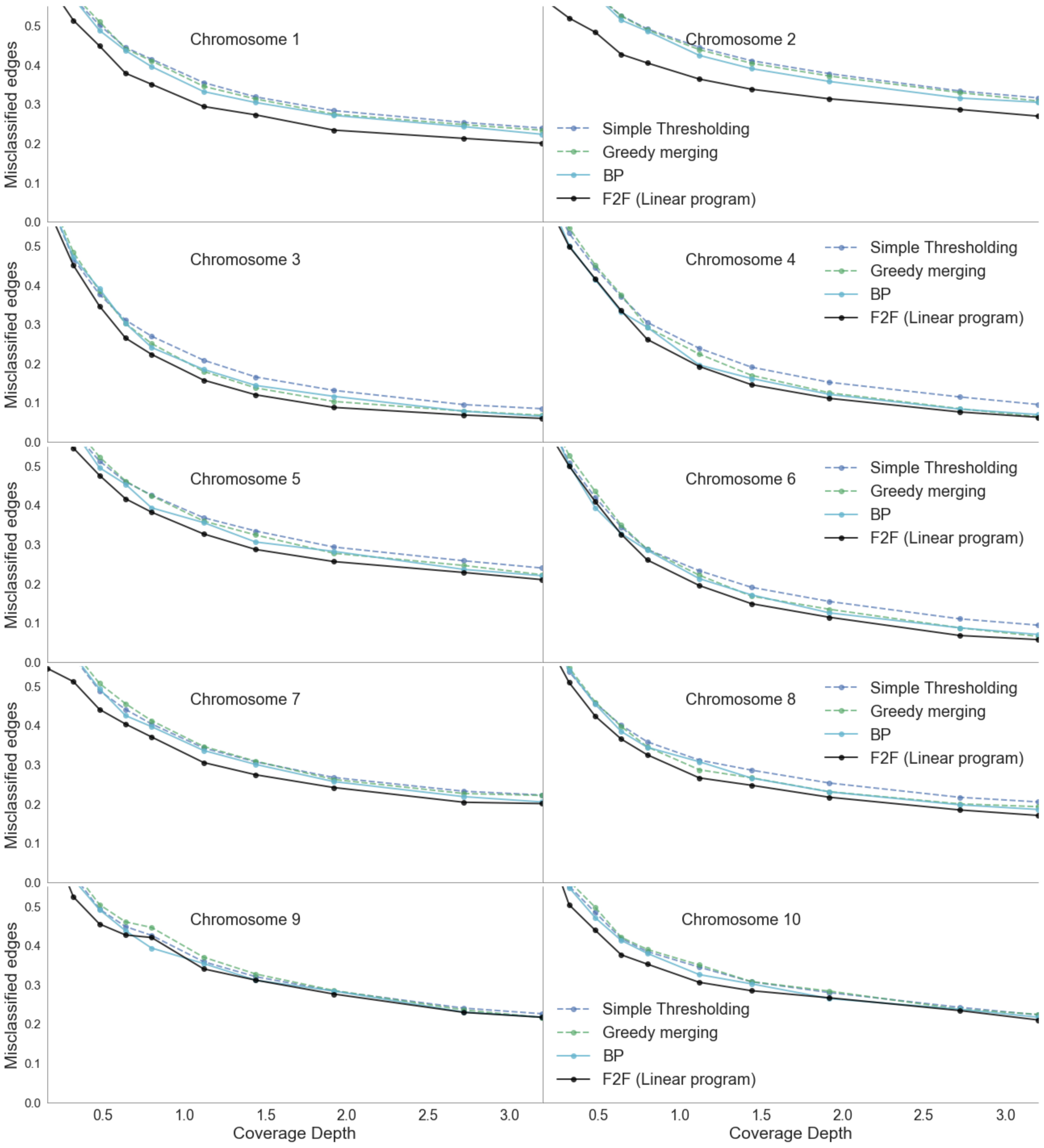}
\caption{Homosapiens [Variable length contigs]: The $x$-axis is the coverage depth of Hi-C reads i.e, expected number of HiC-reads that include a base pair (nucleotide). The $y$-axis is the fraction of misclassified edges.
For each value of coverage depth, the corresponding misclassification error is averaged over $25$ independent runs of experiment.
\label{fig:homosapiens_var}
}
\end{figure}

\begin{figure}
\centering
\includegraphics[width=0.9\columnwidth]{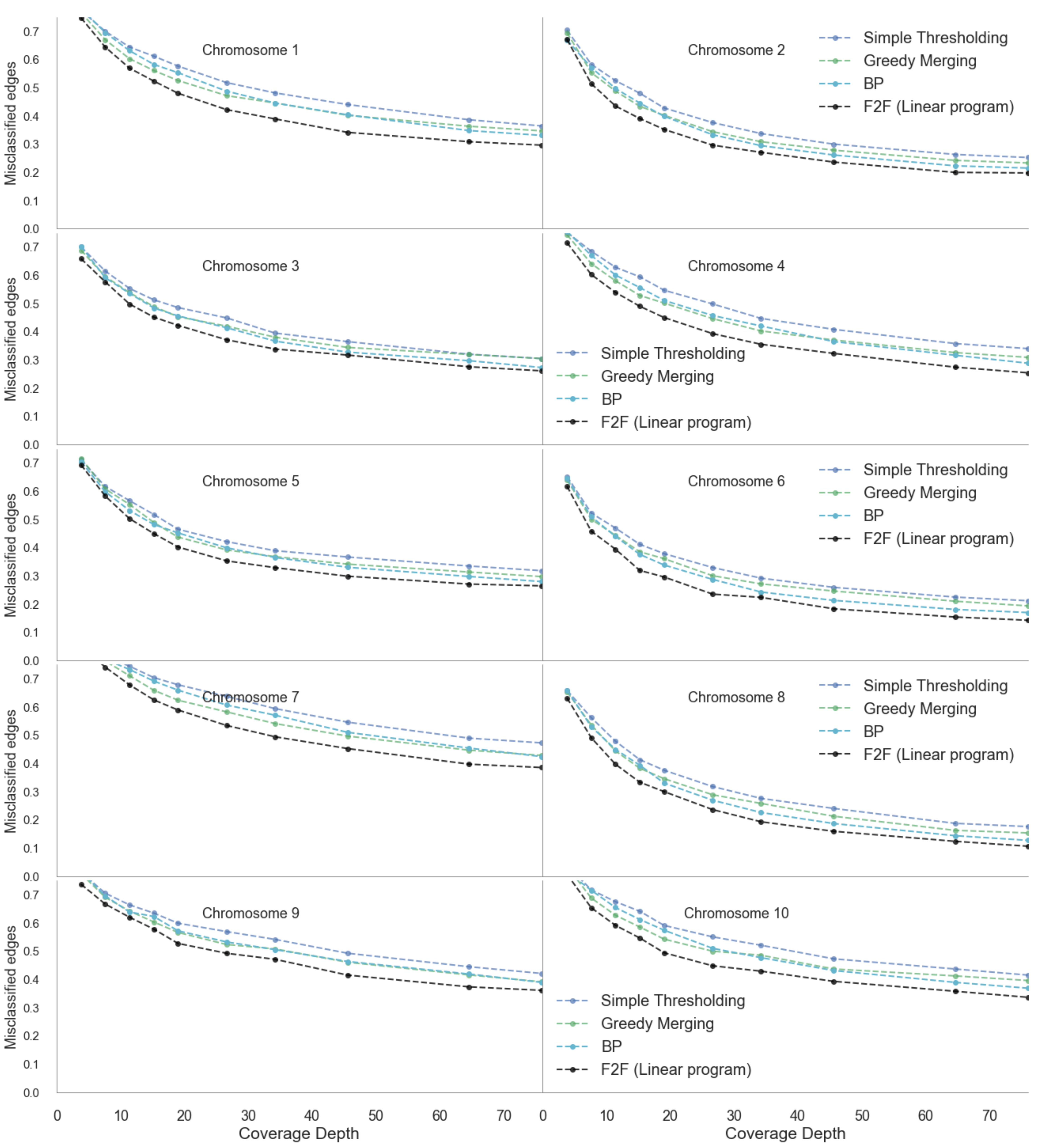}
\caption{Gasterosteus Aculeatus  [Variable length contigs]: The $x$-axis is the coverage depth of Hi-C reads and the $y$-axis is the fraction of misclassified edges.
For each value of coverage depth, the corresponding misclassification error is averaged over $25$ independent runs of experiment.
\label{fig:gasterosteus_var}
}
\end{figure}

\begin{figure}
\centering
\includegraphics[width=0.9\columnwidth]{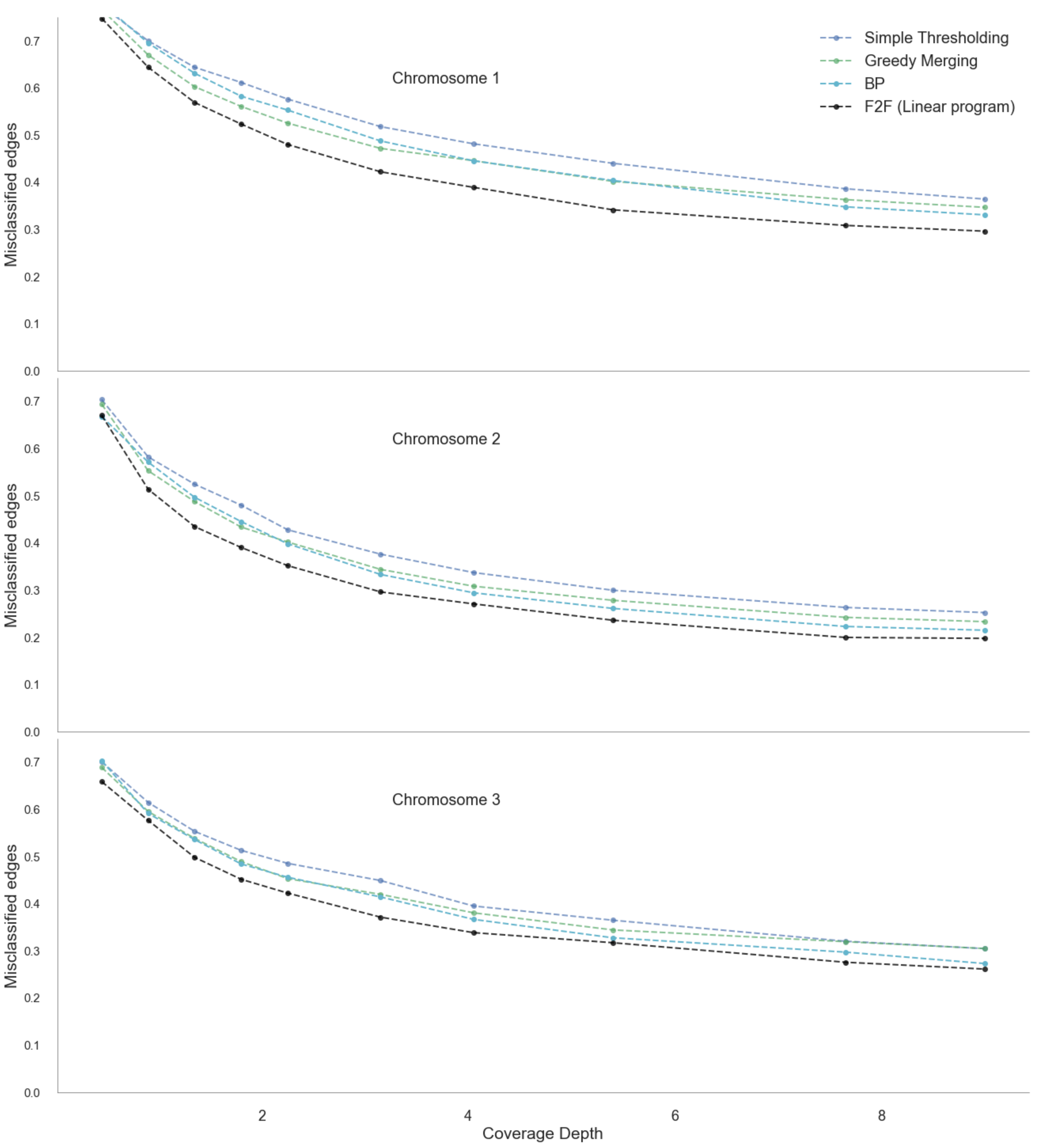}
\caption{Aedes Agyepti  [Variable length contigs]: The $x$-axis is the coverage depth of Hi-C reads and the $y$-axis is the fraction of misclassified edges.
For each value of coverage depth, the corresponding misclassification error is averaged over $25$ independent runs of experiment.
\label{fig:aedes_var}
}
\end{figure}

\newpage

\begin{appendices}

\section{Proof of \prettyref{lmm:eulerian}} \label{app:eulerian}
  By assumption, $G$ is connected and balanced. Therefore, every vertex in $G$ has an even degree.
  Hence, there exists an Eulerian circuit $T=(v_0,v_1,\ldots,v_{m-1},v_0)$, where $m$ is the number of edges in $G$. 
  Next we show that it can be modified, if necessary, to have alternating colors. 
  
  Suppose there exists a pair of two adjacent edges $(v_{i-1},v_i)$ and $(v_{i},v_{i+1})$ such that they have the same color,
  say blue. Then by the balancedness, vertex $v_i$ must be visited later on the circuit, say at step $t$, such that 
  $v_t=v_i$, and $(v_{t-1},v_t)$ and $(v_{t},v_{t+1})$ are both red. Consider the modified tour 
  $$
  T'=(v_0,\ldots,v_{i-1},v_i,v_{t-1}, v_{t-2}, \ldots, v_{i+1}, v_t, v_{t+1}, \ldots, v_{m-1}, v_0)
  $$ 
  formed by reversing the trail between $v_{i+1}$ and $v_{t-1}$. As a result, $(v_{i-1},v_{i})$ and $(v_i,v_{t-1})$ have
  alternating colors. Similarly for $(v_{i+1}, v_t)$ and $(v_t, v_{t+1})$. The ordering of colors on the other adjacent edges is same as that in $T$. Therefore, the number of pairs of adjacent edges of the same color in the modified Eulerian 
  circuit $T'$ decreases by $2$.  Hence one can apply this procedure repeatedly for a finite number of rounds
  to  transform $T$ to an alternating Eulerian circuit.

\section{Large Deviation Bounds}
\label{app:LD}

Recall that $X_i$'s and $Y_i$'s are two independent sequences of random variables, where $X_i$'s are i.i.d.\ copies of $\diff P/\diff Q$ under distribution $P$ and $Y_i $'s are i.i.d.\ copies of $\diff P/\diff Q$ under
distribution $Q$. 
Let $F$ denote the Legendre transform of the log moment generating function of $Y_1-X_1$, i.e.,
\begin{equation}
F(x) =  \sup_{\theta \ge 0} \sth{ \theta x  - \psi_P(-\theta) -\psi_Q(\theta) }.
\label{eq:F}
\end{equation}
Then from the Chernoff bound we have the following large deviation inequality: for any $ x\in\reals$ and any $\ell \in \naturals$,
\begin{equation}
\prob{\sum_{i=1}^\ell (Y_i - X_i) \geq \ell x }\leq \exp \left(-\ell F(x) \right).
\label{eq:LDXY}
\end{equation}
Moreover, it is straightforward to check that 
\begin{align}
F(0) 
=  -  \psi_P( -1/2) - \psi_Q( 1/2) =  \alpha_n \label{eq:Falpha}.
\end{align}
and $F'(0) = 1/2$. 

\begin{lemma}
\label{lmm:LD}
Fix $s,t\in\naturals$ and 
$u \in \integers$ such that $u \ge \min \{-s, -t/2\}$ and $v\in\integers $ such that $ -2(s+u) \le v \le t$. 
Let 
\[
A \triangleq \sum_{i=1}^{s+u} Y_i - \sum_{i=1}^{s} X_i, \quad B \triangleq \sum_{i=1}^{t} Y_i - \sum_{i=1}^{t+2u} X_i.
\]
Let $\{A_i: i \in [N]\}$ be identically distributed as $A$ and $\{B_{ij}:  i\in[N],j\in[M]\}$ as $B$,
where 
$$ 
M \leq n^{t - v} 4^{s+t+u-v/2}, \quad N \leq (2n)^{s+u + v/2} .
$$
Assume that for each $i$, 
	$A_i$ is independent of $\{B_{ij}: j \in [M]\}$. 
If 
\begin{equation}
\alpha_n  - \log n \geq \log 64, 
\label{eq:alphaabove}
\end{equation}
then
\begin{equation}
\prob{ \max_{i\in [N],j\in [M]}  2 A_i + B_{ij} \geq 0} \le 5 \exp \left( - (\alpha-\log n) (s+t+u-v/2)/4 \right).
\label{eq:LD}
\end{equation}
\end{lemma}
\begin{proof}

We abbreviate $\alpha_n$ as $\alpha.$
	Define the rate function for $A$ and $B$ as follows:
		for any $x\in\reals$, 
	\begin{align*}
	F_s(x) = & ~  \sup_{\theta \ge 0} \sth{ \theta x - (s+u) \psi_Q(\theta) - s \psi_P(-\theta) },\\
	G_t(x) = & ~  \sup_{\theta \ge 0} \sth{ \theta x - t \psi_Q(\theta) - (t+2u) \psi_P(-\theta) }.
	\end{align*}
	Then Chernoff's inequality gives the large deviation bound: 
\[
\prob{A \geq x }\leq \exp(-F_s(x)), \quad 
\prob{B \geq x }\leq \exp(-G_t(x)).
\]
Next we provide a lower bound on the rate function. To explain the intuition, consider the easiest situation of $u=0$,
then we have $F_s(x)=G_s(x)=s F(\frac{x}{s})$. 
In the case of large $s$ (which can be shown to be the most adversarial scenario), we have the approximation
$F(\epsilon) = F(0) + F'(0) \epsilon + o(\epsilon) $ as $\epsilon\to0$. Recall $F'(0)=1/2$. Thus we have
$F_s(x) = s \alpha +  x/2 + o(1)$ as $s\diverge$, which depends linearly on both $x$ and $s$.
This intuition can be made precise by the following non-asymptotic bound:
\begin{align}
F_s(x) \geq & ~ 	s \alpha +  x/2 - u \psi_Q(1/2) \label{eq:Fslb}\\
G_t(x) \geq & ~ 	t \alpha +  x/2 - 2 u \psi_P(-1/2 ) \label{eq:Gtlb},
\end{align}
which simply follows from substituting $\theta =1/2$ in the definition of $F_s$ and $G_t$ and applying \prettyref{eq:Falpha}.

Next we bound the probability in \prettyref{eq:LD}. 
By the union bound, we have for any $\tau \in \reals$,
\begin{align}
\prob{\max_{i\in [N],j\in [M]}  2 A_i + B_{ij}  \geq 0} \le 
N \cdot \prob{ A \geq \tau} + NM \cdot \prob{B \geq -2A, A \leq \tau}, \label{eq:union_bound_key}
\end{align}
where $A\indep B$. Notice that here instead of directly taking the union bound
over both $i$ and $j$, we first single out the event $\{\max_{i \in [N]} A_i \geq \tau \}.$
The key idea is to choose $\tau$ appropriately such that  
event $\{\max_{i \in [N]} A_i \geq \tau \}$ has low probability. In view of the assumption $N \le (2n)^{s+u+v/2}$
and \prettyref{eq:Fslb},
a straightforward choice of $\tau$ is such that 
$$
s\alpha+ \tau/2 -u \psi_Q(1/2) = \left( s+u+v/2 \right) \alpha. 
$$ 
Here, we inflate $\tau$ slightly by setting 
\begin{equation}
\tau \triangleq 2  \left(u \psi_Q(1/2 )+ (u+v/2) \alpha + \epsilon(s+t+u-v/2) \log n \right),
\label{eq:tau}
\end{equation}
where 
\begin{align}
\epsilon=\frac{\alpha-\log n}{4\log n}. \label{eq:def_epsilon_alpha}
\end{align}
Then for the first term in \prettyref{eq:union_bound_key}:
\begin{align}
 N \prob{ A \geq \tau} \nonumber \leq & ~ N \exp(-F_s(\tau))	\nonumber \\
\overset{\prettyref{eq:Fslb}}{\leq} & ~ \exp(-(s \alpha + \tau/2 - u \psi_Q(1/2 ) ) + \log (2n) (s+u+v/2))	\nonumber \\
\overset{\prettyref{eq:tau}}{=} & ~ \exp\left(- (s+u+v/2) \left(\alpha-\log (2n) \right)  - \epsilon(s+t+u-v/2) \log n \right) 	\nonumber \\
\overset{\prettyref{eq:alphaabove}}{\leq} & ~ \exp \left( -\epsilon (s+t+u-v/2) \log n \right). \label{eq:union_bound_part_1}
\end{align}
For the second term in \prettyref{eq:union_bound_key}, fix some $\delta>0$ to be chosen later and partition $(-\infty,\tau]$ into disjoint intervals $\cup_{i\geq 1} (\tau_i,\tau_{i-1}]$, where 
$\tau_i=\tau-i \delta$.
Then by union bound,
\begin{align*}
\prob{B \geq -2A, A \leq \tau}
\leq & ~ \sum_{i \geq 1} \prob{B \geq -2A, \tau_i \leq A \leq \tau_{i-1}}	 \\
\leq & ~ \sum_{i \geq 1} \prob{B \geq -2\tau_{i-1}} \prob{A \geq \tau_i }	 \\
\leq & ~ \sum_{i \geq 1} \exp\left(-F_s(\tau_i) - G_t(-2\tau_{i-1}) \right) 
\end{align*}
where, in view of \prettyref{eq:Fslb} and \prettyref{eq:Gtlb}, the combined exponent satisfies
\begin{align*}
 & ~ F_s(\tau_i) + G_t(-2\tau_{i-1}) \\
\geq & ~ s \alpha +  \tau_i/2 - u \psi_Q(1/2) + 	t \alpha -   \tau_{i-1} - 2 u \psi_P(-1/2) 	\\
= & ~ (s+t) \alpha - \tau/2  - 2 u \psi_P(-1/2) - u \psi_Q(1/2 )  + (i- 2) \delta/2	\\
\overset{\prettyref{eq:tau}}{=} & ~ (s+t) \alpha - \left(u \psi_Q(1/2)+ (u+v/2) \alpha + \epsilon(s+t+u-v/2) \log n \right)  - 2 u \psi_P(-1/2) - u \psi_Q(1/2)  + (i- 2) \delta/2	\\
\overset{\prettyref{eq:Falpha}}{=} & ~ (s+t) \alpha  + 2u \alpha - (u+v/2) \alpha - \epsilon (s+t+u-v/2)\log n  + (i- 2) \delta/2	\\
= & ~ (s+t + u-v/2) (\alpha-\epsilon)  + (i- 2) \delta/2
\overset{\prettyref{eq:def_epsilon_alpha}}{\le} (s+t + u-v/2) (1+3\epsilon)\log n + (i- 2) \delta/2.
\end{align*}
 Combining the last two displayed equation and noting that $NM \le (8n)^{s+t+u-v/2}$,  we have
\begin{align*}
NM \prob{B \geq -2A, A \leq \tau} &\leq  \exp \left(  - ( 3\epsilon \log n - \log 8)  (s+t+u-v/2) \right) 
\sum_{i \geq 1} \exp\left(-(i- 2) \delta/2 \right) \\
& \leq \frac{e^{ \delta/2}}{1-e^{-\delta/2}}  \exp \left( -\epsilon (s+t+u-v/2) \log n \right),
\end{align*}
where the last inequality holds because $2\epsilon \log n = (\alpha-\log n)/2 \ge \log 8$ in view of \prettyref{eq:alphaabove}.
Choosing $\delta = 2\log 2$ that minimizes the right hand side of the last displayed equation,
we get
\begin{align}
NM \prob{B \geq -2A, A \leq \tau} \le 4 \exp \left( -\epsilon (s+t+u-v/2) \log n \right). \label{eq:union_bound_part_2}
\end{align}
Finally, we arrive at \prettyref{eq:LD} by combining \prettyref{eq:union_bound_key}, \prettyref{eq:union_bound_part_1}, and
\prettyref{eq:union_bound_part_2}, and \prettyref{eq:def_epsilon_alpha}.
\end{proof}

\section{Proof of \prettyref{lmm:assumption_verify}}
	\label{app:assumption_verify}
\begin{proof}
In the Gaussian case,  the log likelihood ratio is given by
$$
\log \frac{\diff P}{\diff Q} (x)  =  x \mu - \mu^2/2.
$$ 
Hence, 
$$
\psi_P(\theta)= \frac{1}{2} \theta (1+\theta) \mu^2
$$
and $\alpha_n= \mu^2/4$. It follows that $\tau^*=\psi'_P(-1/2) = 0$.
Define $\Phi^c(x) \triangleq \int_x^{\infty} \frac{1}{\sqrt{2\pi}} e^{-t^2/2} dt$ as the tail probability of the standard normal.
Then
\begin{align*}
\log \prob{X \le \tau^*} & = \log \prob{ \calN(\mu,1) \le \mu/2}  = \log \Phi^c(\mu/2)  \ge - \mu^2/8 + O(\log \mu ),  \\
\log \prob{Y \ge \tau^*} & = \log \prob{ \calN(0,1) \ge \mu/2 } = \log \Phi^c(\mu/2) \ge - \mu^2/8 + O(\log \mu ). 
\end{align*}
where the last inequality holds because $\Phi^c(x)\ge \frac{x}{1+x^2} e^{-x^2/2}$. 
Hence \prettyref{ass:chernoff_inverse} holds in the Gaussian case. 

In the Poisson case, for integer $k \geq 0$,
$$
\log \frac{\diff P} {\diff Q} (k) = k \log (\lambda/\mu) - (\lambda - \mu).
$$
Hence, 
$$
\psi_P(\theta)= \lambda \left( (\lambda/\mu)^\theta -1 \right) - \theta (\lambda -\mu). 
$$
and
$$
\alpha_n=\left(\sqrt{\lambda}-\sqrt{\mu} \right)^2, \quad  \tau^* = \sqrt{\lambda \mu} \log (\lambda/\mu) - (\lambda -\mu) .
$$
Let $\eta= \sqrt{\lambda \mu}$. Then  
\begin{align*}
\log \prob{X \le \tau^*} & = \log \prob{\Pois(\lambda ) \le \eta } \\
& \ge \log \prob{ \Pois (\lambda ) =\lfloor \eta \rfloor }  \\
& = \lfloor \eta \rfloor \log \lambda  - \lambda - \log \left( \lfloor \eta \rfloor ! \right)  \\
& = \eta  \log \lambda   - \lambda  -  \eta \log (\eta/e )  +O( \log (\lambda\mu) )  +O(1) \\
& =  -  \lambda + \eta + \frac{1}{2} \eta \log \left( \lambda/\mu \right)  +o(\alpha)+o(\log n),
\end{align*}
where we used the fact that $\log n! \le n \log (n/e) + (1/2)\log n +1$ and the assumption
that $\log (\lambda \mu) = o (\alpha)  + o(\log n).$
Similarly, 
\begin{align*}
\log \prob{Y \ge \tau^*} & = \log \prob{\Pois( \mu ) \ge \eta } \\
& \ge \log \prob{ \Pois (\mu) =\lceil \eta  \rceil  }  \\
& =  \lceil \eta  \rceil  \log \mu  - \mu - \log \left( \lceil \eta  \rceil ! \right)  \\
& = \eta  \log \mu - \mu  - \eta \log (\eta/e)  + +O( \log (\lambda\mu) )  +O(1)  \\
& =  -  \mu + \eta  - \frac{1}{2} \eta \log \left( \lambda/\mu \right)    +o(\alpha)+o(\log n).
\end{align*}
Thus  \prettyref{ass:chernoff_inverse}    holds in the Poisson case. 

In the Bernoulli case, let $\bar{p}=1=p$ and $\bar{q}=1-q$. Then 
$X$ and $Y$ takes the value of either $\log (p/q)$ or $\log (\bar{p}/\bar{q})$. 
Hence, 
\begin{align*}
 \log \prob{X \le \log(p/q) } + \log \prob{ Y \ge \log(p/q) }  & = \log 1 + \log q = \log q \\
  \log \prob{X \le \log (\bar{p}/\bar{q}) } + \log \prob{ Y \ge \log (\bar{p}/\bar{q}) } & = \log \bar{p} + \log 1 = \log \bar{p}.
\end{align*}
Hence, to verify \prettyref{ass:chernoff_inverse}, it suffices to show
\begin{align}
\max \{ \log \bar{p}, \log q \}  \ge - \left( 1+o(1) \right) \alpha + o(\log n). 
  \label{eq:Bernoulli_check}
\end{align}
Note that 
 $$
 - \alpha = 2 \log \left( \sqrt{pq} + \sqrt{\bar{p}\bar{q} } \right) \le  \max \{ \log (pq), \log ( \bar{p} \bar{q} ) \} + \log 4
 \le  \max \{ \log q, \log \bar{p}   \} + \log 4. 
 $$
 Hence, the desired \prettyref{eq:Bernoulli_check} holds. 
\end{proof}

\section{Analysis for the Suboptimal Dual Certificate \prettyref{eq:dual}}
\label{app:dual}

We first present an auxiliary lemma. Recall that 
$$
\beta_n = - \frac{3}{2} \log \int  (\diff P)^{2/3} (\diff Q)^{1/3}.
$$
\begin{lemma}\label{lmm:concentration_2YXX}
$$
\prob{ 2 Y - X_1 -X_2 \ge 0 } \le \exp \left(  - 2 \beta_n  \right).
$$
\end{lemma}
\begin{proof}
It follows from the Chernoff bound that
$$
\prob{ 2 Y - X_1 -X_2 \ge 0 } \le  \exp \left( - \sup_{\theta \ge 0} \left\{ \psi_Q(2\theta ) + 2\psi_P( - \theta) \right \}\right).
$$
It turns out that the supremum is achieved at  $\theta=1/3$ so that $\psi'_Q(2/3) = \psi'_P(-1/3)$ and
$$
\psi_Q(2/3 ) + 2\psi_P( - 1/3) = - 3 \log \int  (\diff P)^{2/3} (\diff Q)^{1/3} =2\beta_n.
$$
\end{proof}

\begin{theorem} \label{thm:F2F_duality}
If 
$$
\beta_n - \log n \to \infty,
$$ 
then the 
optimal solution of F2F \prettyref{eq:F2F}
satisfies 
$
\min_{x^*} \prob{ \hat x_{\rm F2F} = x^* } \to 1 
$ as $n \to \infty$. 
\end{theorem}
\begin{proof}
We use the standard dual certificate argument.
The Lagrangian function is given by
$$
L(x, b, h, u)
= - \iprod{w}{x} - \iprod{b}{x} + \iprod{h}{x- \ones} 
+ \sum_i u_i \left( x \left(\delta(i) \right) -2   \right),
$$
where $\ones$ denotes a $\binom{n}{2}$-dimensional all-one vector; 
$b \ge 0$, $h \ge 0$, and $u \in \reals^n$ are dual variables. 
To show that $x^*$ is an optimal solution to F2F \prettyref{eq:F2F}, it suffices
to construct the dual variables to satisfy the following KKT conditions: for every edge $e = (i,j) $, 
\begin{align}
- w_e - b_e +h_e + u_i + u_j = & 0,   \label{eq:F2FLP_first}\\
h_e (x^*_e -1)   = & 0,  \label{eq:F2FLP_second} \\
b_e x^*_e = & 0  \label{eq:F2FLP_third},
\end{align}
where \prettyref{eq:F2FLP_first} is the first-order condition,
and \prettyref{eq:F2FLP_second}  and \prettyref{eq:F2FLP_third}
are complementary slackness conditions. 

It follows from \prettyref{eq:F2FLP_second} that 
if $x^*_e=0$, then $h_{e} =0$. Similarly, it
follows from \prettyref{eq:F2FLP_third} 
that if $x^*_e=1$, then $b_{e }=0$.  Thus,  in view of 
\prettyref{eq:F2FLP_first}, for every edge $e=(i,j)$, we must have
\begin{align*}
h_{e}  =  w_{e }  - u_i - u_j  \quad & \text{ if }  x^*_e=1,   \\
b_{e}   = - w_{e} +  u_i +u_j   \quad & \text{ if } x^*_e=0, 
\end{align*}
Then to satisfy all the KKT conditions,  it reduces to constructing $u$ such that
that for every $e=(i,j)$, 
\begin{align}
u_i +u_j &\le w_e \quad  \text{ if } x^*_e=1,   \\
 u_i +u_j  & \ge w_e  \quad  \text{ if } x^*_e=0. 
\end{align}

Let
$$
u_i =  \frac{1}{2} \min \left\{  w_e :  e \in \delta(i), x^*_e =1 \right\}.
$$
Then for every $e=(i,j)$ such that $x^*_e=1$, $u_i \le w_e /2$ and
$u_j \le w_e/2$, and thus $ u_i +u_j \le w_e$. Define
event
$$
\calE = \cup_{ e=(i,j): x^*_e=0 }  \left\{ w_e -  u_i/2 - u_j/2     \ge 0 \right\} 
$$
We are left to show that  as $n \to \infty$, 
$$
\prob{ \calE } \to 0. 
$$ 
For each $e=(i,j)$ such that $x^*_e=0$, 
by the union bound, 
\begin{align*}
\prob{   w_e  -   u_i/2  - u_j/2   \ge  0  }
&= \prob{ \cup_{f \in \delta(i): x^*_{f } =1} \cup_{g \in \delta(j): x^*_{g}=1 }  \left\{  2 w_e - w_{f} - w_{g}  \ge 0 \right\}  }  \\
& \le \sum_{ f \in \delta(i): x^*_f =1 } \sum_{g \in \delta(j): x^*_g =1} \prob{ 2 w_e - w_{f} - w_{g} \ge 0 } \\
& =  4   \prob{ 2Y - X_1 -X_2 \ge 0} \\
& \le 4  \exp \left( -2 \beta_n  \right), 
\end{align*}
where the last inequality holds due to \prettyref{lmm:concentration_2YXX}. 
Hence, by the union bound
\begin{align*}
\prob{ \calE }
& \le \sum_{ e=(i,j): x^*_e=0 } \prob{   w_e  -   u_i/2  - u_j/2   \ge  0  } \\
& \le 4 n^2 \exp \left(  -2\beta_n \right) =o(1),
\end{align*}
where the last inequality holds due to the assumption that $\beta_n - \log n \to +\infty$. 

Finally, we show that $x^*$ is the unique optimal solution to F2F on event $\calE^c$. 
Let $\tilde{x}$ is another optimal solution to F2F. Since KKT
conditions are also necessary, it follows that $\tilde{x}$ must satisfy 
$
\iprod{\tilde{x}}{b} =0.
$
On event $\calE^c$, for every $e=(i,j)$ such that $x^*_e=0$, 
$b_e=-w_{e} + u_i +u_j >0$ and hence  $\tilde{x}_{e}  =0$. Moreover,
since $\tilde{x}$ is feasible, it follows that $\tilde{x}_{e} =1$ for every $e$ such that $x^*_e=1$ 
and thus $\tilde{x} = x$. 
\end{proof}

\section{Greedy Methods} \label{app:greedy}

We first consider the simple thresholding algorithm. 
\begin{algorithm}
\caption{Simple Thresholding}\label{alg:thresholding}
\begin{algorithmic}[1]
\STATE For every vertex keep the two edges with the largest weights and delete the rest
of $n-3$ edges, with ties broken arbitrarily. 
\STATE Output the resulting graph $\hat{C}_{\rm TH}$. 
\end{algorithmic}
\end{algorithm}

The following theorem provides an sufficient condition under which the simple thresholding algorithm
to attain exact recovery. 
\begin{theorem}[Simple Thresholding]
If 
\begin{align}
\alpha_n - 2 \log n \to + \infty, \label{eq:suff_simple_thresholding}
\end{align}
then $\min_{C^*} \prob{\hat{C}_{\rm TH}  = C^*} \to 1$ as $n \to \infty.$
\end{theorem}
\begin{proof}
For any vertex $i$, define event 
$$
\calH_i =\left\{  \min\{ w_e : e \in \delta(i), x^*_e=1 \} >   \max \{ w_e : e \in \delta(i), x^*_e=0 \} \right\}. 
$$
Then by union bound and Chernoff's inequality~\prettyref{eq:LDXY}, 
$$
\prob{\calH_i^c}  \le 2 (n-3)  \prob{ X  \le Y } \le 2 ne^{  - \alpha }. 
$$
Let $\calH=\cap_{i=1}^n \calH_i$.  Hence, applying union bound, we get that 
$$
\prob{\calH^c} \le \sum_{i=1}^n \prob{\calH_i^c} \le 2 n^2 e^{  - \alpha  } =o(1),
$$
where the last equality holds in view of the assumption \prettyref{eq:suff_simple_thresholding}.
The proof is complete by noting that  $ \calH \subset \{ \hat{C}_{\rm TH}  = C^* \} $.
\end{proof}

Next we consider the following nearest neighbor algorithm. 
\begin{algorithm}[htb]
\caption{Nearest Neighbor}\label{alg:greedy_search}
\begin{algorithmic}[1]
\STATE Start on an arbitrary vertex as current vertex.
\STATE Find out the edge with the largest weight connecting current vertex and an unvisited vertex $v$.\STATE Set current vertex to $v$ and mark $v$ as visited.
\STATE If all the vertices in domain are visited, then output Hamiltonian cycle $\hat{C}_{\rm NN}=(v_1, \ldots, v_n, v_1)$, where
$v_1, \ldots, v_n$ is the sequence of visited vertices; otherwise go to Step 2.
\end{algorithmic}
\end{algorithm}

The following theorem provides an sufficient condition under which the nearest neighbor algorithm 
to attain exact recovery. 
\begin{theorem}[Nearest Neighbor]
If 
\begin{align}
\alpha_n - 2 \log n \to + \infty, \label{eq:suff_greedy_search}
\end{align}
then $\min_{C^*} \prob{\hat{C}_{\rm NN}  = C^*} \to 1$ as $n \to \infty.$
\end{theorem}
\begin{proof}
At step $t = 1, \ldots, n-2$,  let $V_t$ denote the set of visited vertices and  $v_t$ denote the current vertex,
and define event 
$$
\calH_t =\left\{  \max \{ w_{v_t, j} : j \notin V_t, \; x^*_{v_t,j}=1 \} >   \max \{ w_{v_t,j} : j \notin V_t,  \; x^*_{v_t,j}=0 \} \right\}. 
$$
Define event $\calH = \cap_{t=1}^{n-2}\calH_t.$.
Since on event $\calH_t$,  $v_{t+1}$ must be neighbor of $v_t$ in the true Hamiltonian cycle $C^*$,
it follows that  $\{ \hat{C}_{\rm NN} = C^*\} = \calH$. Hence, the prove the theorem, it suffices to show
$\prob{\calH} \to 1$.  Using the property of conditional probability, we have
$$
\prob{\calH } = \prob{\calH_1} \prod_{t=2}^{n-2} \prob{\calH_t | \calH_{t-1}, \ldots, \calH_1}.
$$
Note that 
\begin{align*}
 \prob{\calH_1^c} & = \prob{ \max \{ X_1, X_2 \}   \le \max_{1 \le j \le n-3} Y_j  } \\
 &  \le  \sum_{j=1}^{n-3} \prob{ \max \{X_1, X_2\} \le Y_j}  \\
 & \le (n-3)  \prob{X \le  Y}  \le (n-3) e^{-\alpha}.
\end{align*}
Fix any $2 \le t \le n-2$ and condition on $(\calH_{t-1}, \ldots, \calH_1)$. 
Then there exists exactly one one unvisited neighbor of $v_t$ in the true Hamiltonian cycle $C^*$.
Moreover, since the edges connecting current vertex $v_t$ and all unvisited vertices are 
independent of $(\calH_{t-1}, \ldots, \calH_1)$, it follows that 
$$
\prob{ \calH_t^c | \calH_{t-1}, \ldots, \calH_1} 
= \prob{ X \le \max_{1 \le j \le n-1-t} Y_j }
\le (n-1-t) \prob{ X \le Y} \le (n-1-t) e^{-\alpha}.
$$
Combining the last three displayed equations yields that
\begin{align*}
\prob{\calH} & = \left( 1- (n-3) e^{-\alpha} \right) \prod_{t=2}^{n-2} \left( 1-  (n-1-t) e^{-\alpha} \right) \\
& \ge 1 - (n-3) e^{-\alpha} - \sum_{t=2}^{n-2} (n-1-t) e^{-\alpha} \\
& = 1-  (n-3) e^{-\alpha} - \frac{(n-2)(n-3)}{2} e^{-\alpha}  \to 1,
\end{align*}
where the first inequality holds by iteratively applying $(1-x)(1-y) \ge 1-x-y$ for $x,y\ge 0$;
the last limit holds in view of assumption \prettyref{eq:suff_greedy_search}.
\end{proof}

Finally, we consider a greedy merging algorithm, which improves over the
previous two greedy algorithms.

\begin{algorithm}[htb]
\caption{Greedy Merging}\label{alg:greedy_merge}
\begin{algorithmic}[1]
\STATE Start with $n$ isolated vertices. 
\STATE Among all vertices with degree strictly less than $2$, select
two of them, say $i$ and $j$, with the largest weight $w_{ij}$ and then
connect them. 
\STATE If all vertices have degree $2$, output the resulting graph $\hat{C}_{\rm GM}$; 
otherwise go to Step 2.
\end{algorithmic}
\end{algorithm}

\begin{theorem}[Greedy Merging]
If 
\begin{align}
\beta_n -  \log n \to + \infty, \label{eq:suff_greedy_merging}
\end{align}
then $\min_{C^*} \prob{\hat{C}_{\rm GM}  = C^*} \to 1$ as $n \to \infty.$
\end{theorem}

\begin{proof}
For all edge $e=(i,j)$ such that $x^*_e=0$, define an event
$$
\calH_e = \left\{ \min_{ f \in \delta(i): x^*_f=1}  w_f > w_e   \right\} \cup \left\{ \min_{g \in \delta(j): x^*_g =1}  > w_e  \right\}
$$
We claim that on event $\calH_e$, edge $e$ will  not  be added in $\hat{C}_{\rm}$.
Suppose not. Then by the definition of $\calH_e$ and the greedy merging algorithm,
either $i$ or $j$ must already have degree $2$ before adding edge $e$ and hence $e$
cannot be added. Define $\calH=\cap_{e:x^*_e=0} \calH_e.$ It follows that 
$
\{\hat{C}_{\rm GM}  = C^*\}=\calH.
$
Hence, to prove the theorem, it suffices to show $\prob{\calH^c} =o(1).$
Note that
\begin{align*}
\calH_e^c & = 
\left\{ \min_{ f \in \delta(i): x^*_f=1}  w_f \le  w_e   \right\} \cap \left\{ \min_{g \in \delta(j): x^*_g =1}  \le  w_e  \right\} \\
& = \cup_{f \in \delta(i):x^*_f=1} \cup_{g \in \delta(j):x^*_g=1} \left\{ w_f \le w_e , w_g \le w_e \right\}.
\end{align*}
Hence, by the union bound,
\begin{align*}
\prob{\calH_e^c} & \le \sum_{f \in \delta(i):x^*_f=1} \sum_{g \in \delta(j):x^*_g=1} \prob{ w_f \le w_e , w_g \le w_e } \\
& \le 4 \prob{ X_1 \le Y, X_2 \le Y} \\
&  \le 4 \prob{ X_1 +X_2 \le 2Y}  \le 4 e^{-2\beta},
\end{align*}
where the last inequality follows from \prettyref{lmm:concentration_2YXX}.
Applying union bound again, we get that
\begin{align*}
\prob{\calH^c} \le \sum_{e: x^*_e=0} \prob{\calH_e^c} \le 4n^2 e^{-2\beta} =o(1),
\end{align*}
where the last inequality holds due to assumption \prettyref{eq:suff_greedy_merging}.
\end{proof}

\section{Exact Recovery under the Partial Observation Model}
\label{app:partial}
Consider the hidden Hamiltonian cycle recovery problem under a partial observation model, where the weight of every edge in $G$ is erased independently with probability $\eta_n$. This is equivalent to recovering the hidden Hamiltonian cycle from $G'$ generated under hidden Hamiltonian cycle 
model in \prettyref{def:model} with $P'_n=\eta_n \delta_{*} +  (1-\eta_n) P_n$ and  $Q_n'=\eta_n \delta_{*} +  (1-\eta_n) Q_n$, where
$*$ is any element outside the support of $P_n$ and $Q_n$ denoting the erasure. 
In this case, the log-likelihood weight matrix is the same as before with each erasure replaced with 0. 
Straightforward calculation shows that
\[
\alpha'_n \triangleq -2 \log \int \sqrt{d P_n' dQ_n'} = -2 \log \left(\eta_n + (1-\eta_n) \exp(-\alpha_n/2) \right).
\] 
Note that $\alpha'_n - \log n \to +\infty$ if and only if $\eta_n=o(1/\sqrt{n})$ and $\alpha_n - \log n \to \infty$.
Hence, \prettyref{thm:LP_opt} implies that if $\eta_n =o(1/\sqrt{n})$ and $\alpha_n-\log n \to +\infty$,
then the F2F LP recovers the hidden Hamiltonian cycle with high probability.

For the converse, suppose MLE recovers the hidden Hamiltonian cycle with high probability. 
Then as shown in \prettyref{fig:c4}, with high probability there is no 4-cycle $(i,i+1,j+1,j,i)$ such that the weights on all $4$ edges are $*$,
because otherwise the likelihood function has at least two maximizers which renders MLE to fail with probability at least $1/2$. 
The probability that
all $4$ edges in $(i,i+1,j+1,j,i)$ have weights $*$ is $\eta_n^4$, and there are roughly $n^2/2$ different
such $4$-cycles. If the weights of all such $4$-cycles were mutually independent, then it would immediately 
follow that $\eta_n^4 n^2 \to 0$, \ie, $\eta_n=o(1/\sqrt{n}).$  However, the weights of all such $4$-cycles 
are dependent because of the multiple appearances of $(i,i+1)$ edges. The dependency can be dealt with 
in the same way as in the proof of \prettyref{thm:converse}. Specifically, it directly follows from 
\prettyref{eq:converse_key_equation} that for any $\tau \in \reals$,
$$
e^{  - \prob{Y \ge \tau}^2 \binom{n \prob{X \le \tau} /4}{2} } +  e^{ - n \prob{X \le \tau} /8 } \ge 1+o(1).
$$
Setting $\tau=0$ and noting that $\prob{X \le \tau} \ge \prob{X =0} \ge \eta_n$ and $\prob{Y\ge \tau} \ge \prob{Y \ge 0} \ge \eta_n$,
we get that 
$$
e^{  - \eta_n^2 \binom{n \eta_n /4}{2} } +  e^{ - n \eta_n/8 } \ge 1+o(1).
$$
It follows that $\eta_n=o(1/\sqrt{n})$ is necessary for exact recovery.

\section{Time Complexity of the Max-Product BP algorithm}
\label{app:BPtime}

In this appendix we bound the number of iterations $t_f$ of Max-Product BP algorithm mentioned in the end of \prettyref{sec:other_efficient_algorithms} in both the Gaussian and the Poisson model under the assumption \prettyref{eq:ITlimit} that ensures the F2F LP succeeds. Recall that $t_f = \lceil  2 n w^*/\epsilon \rceil$, where $w^*$ is the weight of the optimal $2$-factor and $\epsilon$ is the difference between the weight of the optimal $2$-factor and the second largest weight of $b$-factors. 

First we show that $\epsilon \to \infty$, which holds for any weight distribution. 
Following the analysis of the 2F ILP in \prettyref{sec:2f}, we prove a slightly stronger separation between the weight of the true Hamiltonian cycle and the weight of any other $2$-factor. Let $\rho_n \triangleq \alpha_n - \log (2n) $, 
which tends to infinity. Then 
it follows from the Chernoff's inequality that
\begin{align*}
\prob{ w(B) \ge -\rho_n} = \prob{ \sum_{i=1}^{\ell/2} Y_i - \sum_{i=1}^{\ell/2} X_i \ge -\rho_n}  
& \le \exp \left( - (\ell/2) F ( - 2\rho_n /\ell ) \right) \\
& \le \exp \left( - \alpha_n \ell/2 + \rho_n \right).
\end{align*}
where $F$ is given in  \prettyref{eq:F} and the last inequality 
follows because by convexity of $F$, $F(x) \ge F(0)+F'(0)x=\alpha_n+x/2$. Then following the exactly same argument as in \prettyref{sec:2f}, we get that 
\begin{align*}
\prob{ \max_{B \in \calB^*} w(B) \ge -\rho_n} \le e^{\rho_n} \frac{\exp \left\{ -2 \left(\alpha_n - \log (2n) \right) \right\} }{1- \exp \left\{ - \left(\alpha_n-\log (2n) \right)/2 \right\} } \to 0.
\end{align*}
It follows that $\epsilon \ge \rho_n \to +\infty$ with high probability.

Finally, for the Gaussian and Poisson weight distribution (cf.~\prettyref{eq:alphaexample}), the total weight $w^*$ of the true Hamiltonian cycle is distributed as $N(n\mu,n)$ and $\Pois(n\lambda)$, respectively. At the information-theoretic limit, we have $\mu = \Theta(\sqrt{\log n})$ and $\lambda = \Theta(\log n)$. Hence $w^* = O_P(n \log n)$ under both cases.

\section{Data Source and Preparation for DNA Assembly}\label{appendix:data}
\subsection{Data Source}
Table \ref{table:data_source} contains links to all the datasets used for experiments in \prettyref{sec:emp_results}. Each organism requires the following three datasets:
\begin{enumerate}
    \item Contigs -- Fasta file of unordered contigs generated from denovo assembly softwares such as Meraculous ~\cite{chapman2011meraculous}, AllPaths ~\cite{butler2008allpaths} etc.
    \item HiC/Chicago -- Fastq file containing HiC/Chicago  reads obtained from the DNA.
    \item Chromosome -- Chromosomes pre-assembled using other expensive data. These chromosome comprise of the above contigs in the correct order and thus this dataset gives us the  ground truth ordering of the contigs.
\end{enumerate}
\begin{table}[h!]

\centering
 \begin{tabular}{||c |c| c| c||} 
 \hline
 Organism & Contigs & HiC/Chicago & Chromosome \\ [0.5ex] 
 \hline\hline
 HomoSapiens ~\cite{putnam2016chromosome} & \href{https://www.ncbi.nlm.nih.gov/sra/?term=SRR2911057}{SRR2911057} & \href{https://www.ncbi.nlm.nih.gov/nuccore/GL582980.1?report=fasta}{GL582980.1} & \href{https://www.ncbi.nlm.nih.gov/assembly/GCF_000001405.13/}{GRCh37} \\ 
 \hline
  Gasterosteus Aculeatus ~\cite{peichel2016improvement} & \href{https://www.ebi.ac.uk/ena/data/view/SRR4007584}{SRR4007584} & 
  \href{http://datadryad.org/resource/doi:10.5061/dryad.h7h32/1}{dryad.h7h32}& \href{ftp://ftp.ncbi.nlm.nih.gov/genomes/all/GCA/000/180/675/GCA_000180675.1_ASM18067v1/GCA_000180675.1_ASM18067v1_genomic.fna.gz}{GCA000180675} \\
 \hline
  Aedes Agyepti ~\cite{dudchenko2017novo} & 
  \href{https://www.ncbi.nlm.nih.gov/sra?term=SRX2618715}{SRX2618715}
  & \href{https://www.ncbi.nlm.nih.gov/genome/?term=txid7159[Organism:noexp]}{txid7159} & \href{ftp://ftp.ncbi.%nlm.nih.gov/genomes/all/GCA/000/180/675/GCA_000180675.1_ASM18067v1/GCA_000180675.1_ASM18067v1_genomic.fna.gz}{GCA000180675}\\
 \hline
\end{tabular}
 \caption{Datasets.
\label{table:data_source}}
\end{table}


\textbf{Goal}: Order the contigs using the HiC/Chicago reads and check the accuracy of the contig ordering using the true answer obtained from  pre-assembled chromosome.

\subsection{Procedure}
For each chromosome of an organism, the following steps were performed to generate the contigs graph $G$ and to evaluate the performance of F2F, BP, greedy merging and simple thresholding.
\begin{enumerate}
    \item Align the HiC/Chicago reads to the contigs using HiC-Pro ~\cite{servant2015hic}. This gives us the raw number of HiC reads between each pair of contigs.
    \item Construct the graph $G$ using the data obtained in the previous step.
    \item On graph $G$, run F2F, BP, greedy merging and simple thresholding algorithms, with outputs denoted by 
    $\tilde{C}_{\rm F2F}$, $\tilde{C}_{\rm BP}$, $\tilde{C}_{\rm GM}$, and $\tilde{C}_{\rm TH}$, respectively.
    \item Align the contigs on the pre-assembled chromosome to obtain the true ordering of the contigs to obtain the ground truth Hamiltonian path $C^*$.
    \item Evaluate the misclassification error of the outputs of the four algorithms with respect to the ground truth path $C^*$.
\end{enumerate}
\end{appendices}

\section*{Acknowledgment}
D.~Tse and J.~Xu would like to thank the support of the Simons Institute, where this collaboration began.  J.~Xu would also like to thank Mohit Tawarmalani for inspiring discussions
on fractional $2$-factor LP relaxation of TSP. 
Y.~Wu is grateful to Dan Spielman for helpful comments and to David Pollard and Dana Yang for an extremely thorough reading of the manuscript and various corrections.

V.~Bagaria and D.~Tse are  supported by the Center for Science of Information, an NSF Science and Technology Center, under grant agreement CCF-0939370, as well as the National Human Genome Research Institute of the National Institutes of Health under award number R01HG008164.
J.~Ding was supported in part by the NSF Grant DMS-1757479 and an Alfred Sloan fellowship.
Y.~Wu was supported in part by the NSF Grant IIS-1447879, CCF-1527105, an NSF CAREER award CCF-1651588, and an Alfred Sloan fellowship.
J.~Xu was supported in part by a Simons-Berkeley Research Fellowship and the NSF Grant CCF-1755960.


\newcommand{\etalchar}[1]{$^{#1}$}

\end{document}